\documentclass[copyright,creativecommons]{eptcs}
\usepackage{amssymb,stmaryrd}
\usepackage{breakurl}
\usepackage{enumerate}
\usepackage{mathrsfs}
\usepackage{fancybox}
\newfont{\bbb}{bbm10 scaled 1100}        
\def\text#1{\textrm{#1}}

\def\precond#1{{\vphantom{#1}}^\bullet #1}
\def\postcond#1{{#1}^\bullet}

\def\Production#1{\stackrel{#1}{\Longrightarrow}}
\def\production#1{\stackrel{#1}{\longrightarrow}}
\def\equivalent{\Leftrightarrow}
\newfont{\fsc}{eusm10 scaled 1100}      
\def\powermultiset#1{\nat^{#1}}
\def\implies{\Rightarrow}

\def\equivalent{\Leftrightarrow}
\def\mathrlap{\mathpalette\mathrlapinternal}
\def\mathrlapinternal#1#2{%
  \rlap{$\mathsurround=0pt#1{#2}$}}
\def\mathllap{\mathpalette\mathllapinternal}
\def\mathllapinternal#1#2{%
  \llap{$\mathsurround=0pt#1{#2}$}}

\def\into{\rightarrow}

\def\defitem#1{\emph{#1}}

\def\rpair#1{\langle#1\rangle}

\def\onespace#1{\let\argument=#1\ifx\onespace#1\else~\fi\argument}

\def\readyset{\mathcal{R}}

\def\structuralM{\mbox{\sf M}}
\let\origmin\min
\def\min{\mathord{\origmin}}
\let\origmax\max
\def\max{\mathord{\origmax}}

\def\quireunderscore{_}
\def\quire#1{%
  \def\tmp{#1}%
  \ifx\tmp\quireunderscore%
    \def\tmp{\quireindexed_}
  \else%
    \def\tmp{\mathcal{Q}#1}
  \fi\tmp}
\def\quireindexed_#1{\mathcal{Q}_{\text{#1}}}

\makeatletter
\def\goesto{\@transition\rightarrowfill}
\def\Goesto{\@transition\Rightarrowfill}
\def\ngoesto{\@transition\nrightarrowfill}
\def\nGoesto{\@transition\nRightarrowfill}
\def\@transition#1{\@ifnextchar[{\@@transition{#1}}{\@@transition{#1}[]}}
\newbox\@transbox
\newbox\@arrowbox
\def\rightarrowfill{$\m@th\mathord-\mkern-6mu%
  \cleaders\hbox{$\mkern-2mu\mathord-\mkern-2mu$}\hfill
  \mkern-6mu\mathord\rightarrow$}
\def\Rightarrowfill{$\m@th\mathord=\mkern-6mu%
  \cleaders\hbox{$\mkern-2mu\mathord=\mkern-2mu$}\hfill
  \mkern-6mu\mathord\Rightarrow$}
\def\@@transition#1[#2]%
   {\setbox\@transbox\hbox
       {\vrule height 1.5ex depth 1ex width 0ex\hskip0.25em$\scriptstyle#2$\hskip0.25em}
   \ifdim\wd\@transbox<1.5em
      \setbox\@transbox\hbox to 1.5em{\hfil\box\@transbox\hfil}\fi
   \setbox\@arrowbox\hbox to \wd\@transbox{#1}
   \ht\@arrowbox\z@\dp\@arrowbox\z@
   \setbox\@transbox\hbox{$\mathop{\box\@arrowbox}\limits^{\box\@transbox}$}
   \ht\@transbox\z@\dp\@transbox\z@
   \mathrel{\box\@transbox}}
\makeatother

\def\alignedcaption[#1&#2]{\mbox{\scriptsize $\mathllap{#1{}}\mathrlap{#2}$}}

\def\ie{i.e.\ }

\def\varnothing{\emptyset}
\def\Act{{\rm Act}}
\def\Loc{\textrm{\upshape Loc}}
\def\concurrent{\smile}

\def\restrictedto{\mathop\upharpoonright}

\newcommand{\dcup}{\stackrel{\mbox{\huge .}}{\cup}}   
\newcommand{\plat}[1]{\raisebox{0pt}[0pt][0pt]{#1}}   
\newcommand{\inp}{\mathbin\in}                        
\def\idx#1#2#3#4#5{
  \def\argone{#1}
  \def\argtwo{#2}
  \def\argthree{#3}
  \def\argfour{#4}
  \def\argfive{#5}
  \def\testprime{'}
  \def\testdprime{''}
  \def\testtprime{'''}
  {\vphantom{\argthree}}_%
    {\vphantom{\argfour}\argone}^%
    {\vphantom{\argfive}{\argtwo}}%
  \argthree_%
    {\vphantom{\argone}\argfour}%
    \ifx\argfive\testprime\argfive\else%
    \ifx\argfive\testdprime\argfive\else%
    \ifx\argfive\testtprime\argfive\else%
    ^{\vphantom{\argtwo}\argfive}\fi\fi\fi%
}

\def\indexset{\mathfrak{K}}

\def\impl#1{\mathcal{I}(#1)}

\newcommand{\visible}{}

\newenvironment{itemise}{\begin{list}{$\bullet$}{\leftmargin 12pt \labelwidth\leftmargin\advance\labelwidth-\labelsep \topsep 4pt \itemsep 2pt \parsep 2pt}}{\end{list}}

\def\justempty{}

\DeclareFontFamily{T1}{la}{}
\DeclareFontShape{T1}{la}{m}{n}{<->s*[0.8571428571]la14}{}

\newcommand{\Rel}{\mathcal{B}\,}

\def\titlerunning{On Distributability of Petri Nets}
\title{\titlerunning\thanks{This work was
    partially supported by the DFG (German Research Foundation).\newline
    An extended abstract of this paper appeared in L.~Birkedal, ed.: Proc.\ 15th
    Int.\ Conf.\ on Foundations of Software Science and
    Computation Structures (FoSSaCS 2012), LNCS 7213,
    Springer, 2012, pp. 331--345, doi:\urlalt{http://dx.doi.org/10.1007/978-3-642-28729-9_22}
    {10.1007/978-3-642-2872-9_22}.}}
\def\authorrunning{R.J. van Glabbeek, U. Goltz \& J.-W. Schicke-Uffmann}
\providecommand{\publicationstatus}{Technical Report 2011-10\\
  Institut f\"ur Programmierung und Reaktive Systeme\\
  Technical University of Braunschweig}

\author{
Rob van Glabbeek
\institute{NICTA, Sydney, Australia}
\institute{School of Computer Science and Engineering\\
 Univ.\ of New South Wales,
 Sydney, Australia}
\email{rvg@cs.stanford.edu}
 \and
Ursula Goltz  \qquad~
Jens-Wolfhard Schicke-Uffmann
\institute{Institute for Programming and Reactive Systems\\
 TU Braunschweig, Germany}
 \email{goltz@ips.cs.tu-bs.de  \qquad \qquad drahflow@gmx.de}
}

\newtheorem{defi}{Definition}
\newtheorem{theo}{Theorem}
\newtheorem{prop}{Proposition}
\newtheorem{lemm}{Lemma}
\newtheorem{coro}{Corollary}
\newtheorem{exam}{Example}
\newtheorem{obse}{Observation}
\newtheorem{clai}{Claim}

\newenvironment{definition}[1]{\begin{defi} \rm \label{df-#1}} {\end{defi}}
\newenvironment{theorem}[1]{\begin{theo} \rm \label{thm-#1}} {\end{theo}}
\newenvironment{proposition}[1]{\begin{prop} \rm \label{pr-#1}} {\end{prop}}
\newenvironment{lemma}[1]{\begin{lemm} \rm \label{lem-#1}} {\end{lemm}}
\newenvironment{corollary}[1]{\begin{coro} \rm \label{cor-#1}} {\end{coro}}

\newenvironment{observation}[1]{\begin{obse} \rm \label{obs-#1}} {\end{obse}}
\newenvironment{claim}[1]{\begin{clai} \rm \label{cl-#1}} {\end{clai}}
\newenvironment{proof}{\begin{trivlist} \item[\hspace{\labelsep}\bf Proof:]}{\hfill $\Box$\end{trivlist}}
\newenvironment{proofNobox}{\begin{trivlist} \item[\hspace{\labelsep}\bf Proof:]}{\end{trivlist}}
\newcommand{\filledbox}{\rule{1.2ex}{1.2ex}}
\newenvironment{proofclaim}{\begin{trivlist} \item[\hspace{\labelsep}\it Proof:]}{\hfill\filledbox\end{trivlist}}
\newenvironment{proofclaimNobox}{\begin{trivlist} \item[\hspace{\labelsep}\it Proof:]}{\end{trivlist}}

\newcommand{\refdf}[1]{Definition~\ref{df-#1}}
\newcommand{\refthm}[1]{Theorem~\ref{thm-#1}}
\newcommand{\refpr}[1]{Proposition~\ref{pr-#1}}
\newcommand{\reflem}[1]{Lemma~\ref{lem-#1}}

\newcommand{\reffig}[1]{Figure~\ref{fig-#1}}

\newcommand{\refobs}[1]{Observation~\ref{obs-#1}}
\newcommand{\refcl}[1]{Claim~\ref{cl-#1}}
\newcommand{\refsec}[1]{Section~\ref{sec-#1}}
\newcommand{\reftab}[1]{Table~\ref{tab-#1}}
\newcommand{\UI}{\Omega}             
\newcommand{\UIij}{\UI_{i}}          
\newcommand{\ui}{\mbox{\it \i}}      
\newcommand{\dist}[1][p]{\textsf{distribute}_#1}
\newcommand{\ini}[1][j]{\textsf{initialise}_#1}
\newcommand{\trans}[2][h]{\textsf{transfer}^#1_#2}
\newcommand{\exec}[2][i]{\textsf{execute}^#1_#2}
\newcommand{\fetch}[1][p,c]{\textsf{fetch}_{i,j}^{#1}}
\newcommand{\fetched}[2][i]{\textsf{fetched}^#1_#2}
\newcommand{\comp}[2][i]{\textsf{finalise}^#1}

\newcommand{\fire}{\textsf{fire}}
\newcommand{\undo}[1][i]{\textsf{undo}_{#1}}
\newcommand{\und}[1][f]{\textsf{undo}(#1)}
\newcommand{\undone}{\textsf{undone}}
\newcommand{\reset}[1][i]{\textsf{reset}_{#1}}
\newcommand{\elide}[1][i]{\textsf{elide}_{#1}}
\newcommand{\ack}[1][i]{\textsf{ack}_{#1}}
\newcommand{\keep}[1][i]{\rho_{#1}}
\newcommand{\Fired}{\textsf{fired}}
\newcommand{\take}{\textsf{take}}
\newcommand{\took}{\textsf{took}}

\newcommand{\Pre}{\textsf{pre}}
\newcommand{\transin}[2][h]{\textsf{trans}^#1_#2\textsf{-in}}
\newcommand{\transout}[2][h]{\textsf{trans}^#1_#2\textsf{-out}}
\newcommand{\fetchin}[1][p,c]{\textsf{fetch}_{i,j}^{#1}\textsf{-in}}
\newcommand{\fetchout}[1][p,c]{\textsf{fetch}_{i,j}^{#1}\textsf{-out}}

\newcommand{\weight}[1]{\hfill\mbox{\scriptsize $F'(#1)$}}

\newcommand{\leqc}{\leq^\#\!}
\newcommand{\confeq}{\mathbin{\plat{$\stackrel{\#}{=}$}}}
\newcommand{\confeqscript}{\stackrel{\#}{=}}
\newcommand{\nat}{\mbox{\bbb N}}
\newcommand{\Int}{\mbox{\bbb Z}}
\newcommand{\fin}{\in_f}
\newcommand{\marking}[1]{\llbracket#1\rrbracket}  
\newcommand{\NF}{\mbox{\it NF}}                          
\newcommand{\opt}[1]{\mbox{\tiny\rm(}#1\mbox{\tiny\rm)}} 
\usepackage{pstricks}
\usepackage{pst-node}
\usepackage{graphicx}

\newrgbcolor{darkblue}{0 0 0.5}
\newrgbcolor{darkred}{0.7 0 0}

\newcounter{netimage}

\def\p#1:#2;{\cnode #1{0.3}{n\thenetimage-#2}}
\def\P#1:#2;{\p #1:#2;\pscircle*#1{0.1}}
\def\q#1:#2:#3;{\p #1:#2;\rput#1{\rput[l](0.45,0){\large #3}}}
\def\Q#1:#2:#3;{\P #1:#2;\rput#1{\rput[l](0.45,0){\large #3}}}
\def\qq#1:#2:#3;{\p #1:#2;\rput#1{\rput[t](0,-0.5){\large #3}}}
\def\ql#1:#2:#3;{\p #1:#2;\rput#1{\rput[r](-0.45,0){\large #3}}}
\def\qr#1:#2:#3;{\p #1:#2;\rput#1{\rput[l](0.45,0){\large #3}}}
\def\qt#1:#2:#3;{\p #1:#2;\rput#1{\rput[b](0,0.45){\large #3}}}
\def\qb#1:#2:#3;{\p #1:#2;\rput#1{\rput[t](0,-0.45){\large #3}}}
\def\Ql#1:#2:#3;{\P #1:#2;\rput#1{\rput[r](-0.45,0){\large #3}}}
\def\Qr#1:#2:#3;{\P #1:#2;\rput#1{\rput[l](0.45,0){\large #3}}}
\def\Qt#1:#2:#3;{\P #1:#2;\rput#1{\rput[b](0,0.45){\large #3}}}
\def\Qb#1:#2:#3;{\P #1:#2;\rput#1{\rput[t](0,-0.45){\large #3}}}
\def\qx#1:#2:#3:#4;{\p #1:#2;\rput#1{\rput#4{\large #3}}}
\def\QXX#1:#2:#3:#4:#5;{\p #1:#2;\rput#1{\rput#4{\large #3}}\pscircle*#5{0.1}}
\def\s#1:#2:#3;{\p #1:#2;\rput#1{\rput(-0.03,0){\large #3}}}
\def\t#1:#2:#3;{\rput#1{\rnode{n\thenetimage-#2}{\psframebox{%
  \vbox to 0.6cm{\vfil\hbox to 0.6cm{\hfil\Large #3\hfil}\vfil}}}}}
\def\u#1:#2:#3:#4;{\rput#1{\rnode{n\thenetimage-#2}{\psframebox{%
  \vbox to 0.6cm{\vfil\hbox to 0.6cm{\hfil\Large #3\hfil}\vfil}}}}%
  \rput#1{\rput[l](0.6,0){\large #4}}}
\def\ut#1:#2:#3:#4;{\rput#1{\rnode{n\thenetimage-#2}{\psframebox{%
  \vbox to 0.6cm{\vfil\hbox to 0.6cm{\hfil\Large #3\hfil}\vfil}}}}%
  \rput#1{\rput[b](0,0.6){\large #4}}}
\def\ub#1:#2:#3:#4;{\rput#1{\rnode{n\thenetimage-#2}{\psframebox{%
  \vbox to 0.6cm{\vfil\hbox to 0.6cm{\hfil\Large #3\hfil}\vfil}}}}%
  \rput#1{\rput[t](0,-0.6){\large #4}}}
\def\ul#1:#2:#3:#4;{\rput#1{\rnode{n\thenetimage-#2}{\psframebox{%
  \vbox to 0.6cm{\vfil\hbox to 0.6cm{\hfil\Large #3\hfil}\vfil}}}}%
  \rput#1{\rput[r](-0.6,0){\large #4}}}
\def\ur#1:#2:#3:#4;{\rput#1{\rnode{n\thenetimage-#2}{\psframebox{%
  \vbox to 0.6cm{\vfil\hbox to 0.6cm{\hfil\Large #3\hfil}\vfil}}}}%
  \rput#1{\rput[l](0.6,0){\large #4}}}
\def\a#1->#2;{\ncline{->}{n\thenetimage-#1}{n\thenetimage-#2}}
\def\aBack#1->#2;{\ncline{-}{n\thenetimage-#1}{n\thenetimage-#2}\ncline[nodesep=0.2]{-<}{n\thenetimage-#1}{n\thenetimage-#2}}
\def\aEarly#1->#2;{\ncline{->}{n\thenetimage-#1}{n\thenetimage-#2}\ncline[nodesep=0.25]{-<}{n\thenetimage-#1}{n\thenetimage-#2}}
\def\aLate#1->#2;{\ncline{->}{n\thenetimage-#1}{n\thenetimage-#2}\ncline[nodesep=0.6]{-<}{n\thenetimage-#1}{n\thenetimage-#2}}
\def\aUndo#1->#2;{\ncline[arrowinset=0]{->}{n\thenetimage-#1}{n\thenetimage-#2}}
\def\aReset#1->#2;{\ncline[arrowinset=0]{->}{n\thenetimage-#1}{n\thenetimage-#2}\ncline[linecolor=white,arrowinset=0,arrowscale=0.6,linestyle=none,nodesep=0.07]{->}{n\thenetimage-#1}{n\thenetimage-#2}}
\let\aFar\aLate

\let\aAck\aReset
\def\A#1->#2;{\ncarc[arcangle=22]{->}{n\thenetimage-#1}{n\thenetimage-#2}}
\def\B#1->#2;{\ncarc[arcangle=-18]{->}{n\thenetimage-#1}{n\thenetimage-#2}}
\def\BB#1->#2;{\ncarc[arcangle=-32]{->}{n\thenetimage-#1}{n\thenetimage-#2}}
\def\avlinearc{0.2}
\def\av#1[#2]-#3->[#4]#5;{
  \SpecialCoor
  \psline[linearc=\avlinearc]{->}([angle=#2]n\thenetimage-#1)#3([angle=#4]n\thenetimage-#5)
}
\def\interface#1:#2:#3:#4;{\rput#1{\rnode{n\thenetimage-#2}{\psframebox*[fillstyle=solid,fillcolor=black]{\vbox to 0.5cm{\vfill\hbox to 0.05cm{}}}}\rput[lt](0.05,-0.1){#4}}\ncline{-}{n\thenetimage-#2}{n\thenetimage-#3}}

\makeatletter
\long\def\petrinet(#1)#2\end{\psscalebox{0.8}{\pspicture(#1)\stepcounter{netimage}#2\endpspicture}\end}

\makeatother

\psset{arrowsize=5pt}

\begin{document}

\maketitle

\begin{abstract}
We formalise a general concept of distributed systems as sequential
components interacting asynchronously.  We define a corresponding
class of Petri nets, called LSGA nets, and precisely characterise
those system specifications which can be implemented as LSGA nets up
to branching ST-bisimilarity with explicit divergence.
\end{abstract}

\section{Introduction}
The aim of this paper is to contribute to a fundamental understanding
of the concept of a distributed reactive system and the paradigms of
synchronous and asynchronous interaction. We start by giving an
intuitive characterisation of the basic features of distributed
systems. In particular we assume that distributed systems consist of
components that reside on different locations, and that any signal
from one component to another takes time to travel.  Hence the only
interaction mechanism between components is asynchronous
communication.

Our aim is to characterise which system specifications may be implemented as
distributed systems.  In many formalisms for system specification or design,
synchronous communication is provided as a basic notion; this happens for
example in process algebras.  Hence a particular challenge is that it may be
necessary to simulate synchronous communication by asynchronous communication.
 
Trivially, any system specification may be implemented distributedly by locating
the whole system on one single component. Hence we need to pose some additional
requirements. One option would be to specify locations for system activities and
then to ask for implementations satisfying this distribution and still
preserving the behaviour of the original specification. This is done in
\cite{BCD02}. Here we pursue a different approach. We add another requirement to
our notion of a distributed system, namely that its components only allow
sequential behaviour. We then ask whether an arbitrary system specification may
be implemented as a distributed system consisting of sequential components in an
optimal way, that is without restricting the concurrency of the original
specification. This is a particular challenge when synchronous communication
interacts with concurrency in the specification of the original system. We will
give a precise characterisation of the class of distributable systems, which
answers in particular under which conditions synchronous communication may be
implemented in a distributed setting.
  
For our investigations we need a model which is expressive enough to represent
concurrency. It is also useful to have an explicit representation of the
distributed state space of a distributed system, showing in particular the local
control states of components. We choose Petri nets, which offer these
possibilities and additionally allow finite representations of infinite
behaviours. We work within the class of \emph{structural conflict nets}
\cite{glabbeek11ipl}---a proper generalisation of the class of
one-safe place/transition systems, where conflict and concurrency are clearly separated.

For comparing the behaviour of systems with their distributed implementation we
need a suitable equivalence notion. Since we think of open systems interacting
with an environment, and since we do not want to restrict concurrency in
applications, we need an equivalence that respects branching time and
concurrency to some degree. Our implementations use transitions which are
invisible to the environment, and this should be reflected in the equivalence
by abstracting from such transitions. However, we do not want implementations to
introduce divergence. In the light of these requirements we work with two
semantic equivalences. \emph{Step readiness equivalence} is one of the
weakest equivalences that captures branching time, concurrency and divergence to some
degree; whereas \emph{branching ST-bisimilarity with explicit divergence}
fully captures branching time, divergence, and those aspects of concurrency that
can be represented by concurrent actions overlapping in time.
We obtain the same characterisation for both notions of equivalence, and thus
implicitly for all notions in between these extremes.

We model distributed systems consisting of sequential components as an
appropriate class of Petri nets, called \emph{LSGA nets}.  These are obtained by
composing nets with sequential behaviour by means of an asynchronous parallel
composition. We show that this class corresponds exactly to a more abstract
notion of distributed systems, formalised as \emph{distributed nets}
\cite{glabbeek08syncasyncinteractionmfcs}.

We then consider distributability of system specifications which are represented
as structural conflict nets. A net $N$ is \emph{distributable} if there exists a
distributed implementation of $N$, that is a distributed net which is
semantically equivalent to $N$.  In the implementation we allow unobservable
transitions, and labellings of transitions, so that single actions of the
original system may be implemented by multiple transitions. However, the system
specifications for which we search distributed implementations are \emph{plain}
nets without these features.

We give a precise characterisation of distributable nets in terms of a
semi-structural property.  This characterisation provides a formal proof that
the interplay between choice and synchronous communication is a key issue for
distributability.

To establish the correctness of our characterisation we develop a new
method for rigorously proving the equivalence of two Petri nets, one of 
which known to be plain, up to branching ST-bisimilarity with explicit divergence.
 
\section{Basic Notions}
\label{sec-basic}
In this paper we employ \emph{signed multisets}, which generalise
multisets by allowing elements to occur in it with a negative multiplicity.

\begin{definition}{multiset}
Let $X$ be a set.
\begin{list}{{\bf --}}{\leftmargin 18pt
                        \labelwidth\leftmargini\advance\labelwidth-\labelsep
                        \topsep 0pt \itemsep 0pt \parsep 0pt}
\item A {\em signed multiset} over $X$ is a function $A\!:X \rightarrow \Int$,
  \ie $A\in \Int^{X}$.\\
  It is a \emph{multiset} iff $A\in \nat^X$, \ie iff $A(x)\geq 0$ for all $x\in X$.
\item $x \in X$ is an \defitem{element of} a signed multiset $A\in\nat^X$, notation $x \in
  A$, iff $A(x) \neq 0$. 
\item For signed multisets $A$ and $B$ over $X$ we write $A \leq B$ iff
 \mbox{$A(x) \leq B(x)$} for all $x \inp X$;
\\ $A\cup B$ denotes the signed multiset over $X$ with $(A\cup B)(x):=\text{max}(A(x), B(x))$,
\\ $A\cap B$ denotes the signed multiset over $X$ with $(A\cap B)(x):=\text{min}(A(x), B(x))$,
\\ $A + B$ denotes the signed multiset over $X$ with $(A + B)(x):=A(x)+B(x)$,
\\ $A - B$ denotes the signed multiset over $X$ with $(A - B)(x):=A(x)-B(x)$, and\\
for $k\inp\nat$ the signed multiset $k\cdot A$ is given by $(k \cdot A)(x):=k\cdot A(x)$.
\item The function $\emptyset\!:X\rightarrow\nat$, given by
  $\emptyset(x):=0$ for all $x \inp X$, is the \emph{empty} multiset over $X$.
\item If $A$ is a signed multiset over $X$ and $Y\subseteq X$ then
  $A\restrictedto Y$ denotes the signed multiset over $Y$ defined by
  $(A\restrictedto Y)(x) := A(x)$ for all $x \inp Y$.
\item The cardinality $|A|$ of a signed multiset $A$ over $X$ is given by
  $|A| := \sum_{x\in X}|A(x)|$.
\item A signed multiset $A$ over $X$ is \emph{finite}
  iff $|A|<\infty$, i.e.,
  iff the set $\{x \mid x \inp A\}$ is finite.\\
  We write $A\fin\Int^X$ or $A\fin\nat^X$ to indicate that $A$ is a finite (signed)
  multiset over $X$.
\item Any function $f:X\rightarrow\Int$ or $f:X\rightarrow\Int^Y$ from
  $X$ to either the integers or the signed multisets over some set $Y$
  extends to the finite signed multisets $A$ over $X$ by $f(A)=\sum_{x\in X}A(x)\cdot f(x)$.
\end{list}
\end{definition}
Two signed multisets $A\!:X \rightarrow \Int$ and $B\!:Y\rightarrow \Int$
are \emph{extensionally equivalent} iff
$A\restrictedto (X\cap Y) = B\restrictedto (X\cap Y)$,
$A\restrictedto (X\setminus Y) = \emptyset$, and
$B \restrictedto (Y\setminus X) = \emptyset$.
In this paper we often do not distinguish extensionally equivalent
signed multisets. This enables us, for instance, to use $A + B$ even
when $A$ and $B$ have different underlying domains.
A multiset $A$ with $A(x) \in\{0,1\}$ for all $x$ is
identified with the set $\{x \mid A(x)=1\}$.
A signed multiset with elements $x$ and $y$, having
multiplicities $-2$ and $3$, is denoted as $-2\cdot\{x\}+3\cdot\{y\}$.

We consider here general labelled place/transition systems with arc weights. Arc weights
are not necessary for the results of the paper, but are included for the sake of generality.

\begin{definition}{Petri net}
  Let \Act{} be a set of \emph{visible actions} and
  $\tau\mathbin{\not\in}\Act$ be an \emph{invisible action}. Let $\Act_\tau:=\Act \dcup \{\tau\}$.
  A (\emph{labelled}) \defitem{Petri net} (\emph{over $\Act_\tau$}) is a tuple
  $N = (S, T, F, M_0, \ell)$ where
  \begin{list}{{\bf --}}{\leftmargin 18pt
                        \labelwidth\leftmargini\advance\labelwidth-\labelsep
                        \topsep 0pt \itemsep 0pt \parsep 0pt}
    \item $S$ and $T$ are disjoint sets (of \defitem{places} and \defitem{transitions}),
    \item $F: (S \times T \cup T \times S) \rightarrow \nat$
      (the \defitem{flow relation} including \defitem{arc weights}),
    \item $M_0 : S \rightarrow \nat$ (the \defitem{initial marking}), and
    \item \plat{$\ell: T \into \Act_\tau$} (the \defitem{labelling function}).
  \end{list}
\end{definition}

\noindent
Petri nets are depicted by drawing the places as circles and the
transitions as boxes, containing their label.  Identities of places
and transitions are displayed next to the net element.  When
$F(x,y)>0$ for $x,y \inp S\cup T$ there is an arrow (\defitem{arc})
from $x$ to $y$, labelled with the \emph{arc weight} $F(x,y)$.
Weights 1 are elided.  When a Petri net represents a concurrent
system, a global state of this system is given as a \defitem{marking},
a multiset $M$ of places, depicted by placing $M(s)$ dots
(\defitem{tokens}) in each place $s$.  The initial state is $M_0$.

To compress the graphical notation, we also allow universal
quantifiers of the form $\forall x. \phi(x)$ to appear
in the drawing (cf.~\reffig{conflictrepl}).
A quantifier replaces occurrences of $x$ in element identities
with all concrete values for which $\phi(x)$ holds, possibly
creating a set of elements instead of the depicted single one.
An arc of which only one end is replicated by a given quantifier
results in a fan of arcs, one for each replicated element.
If both ends of an arc are affected by the same quantifier,
an arc is created between pairs of elements corresponding to the same $x$,
but not between elements created due to differing values of $x$.

The behaviour of a Petri net is defined by the possible moves between
markings $M$ and $M'$, which take place when a finite multiset $G$ of
transitions \defitem{fires}.  In that case, each occurrence of a
transition $t$ in $G$ consumes $F(s,t)$ tokens from each 
place $s$.  Naturally, this can happen only if $M$ makes all these
tokens available in the first place.  Next, each $t$ produces $F(t,s)$ tokens
in each $s$.  \refdf{firing} formalises this notion of behaviour.

\begin{definition}{preset}
Let $N = (S, T, F, M_0, \ell)$ be a Petri net and $x\inp S\cup T$.\\
The multisets $\precond{x},~\postcond{x}: S\cup T \rightarrow
\nat$ are given by $\precond{x}(y)=F(y,x)$ and
$\postcond{x}(y)=F(x,y)$ for all $y \inp S \cup T$.
If $x\in T$, the elements of $\precond{x}$ and $\postcond{x}$ are
called \emph{pre-} and \emph{postplaces} of $x$, respectively, and if
$x\in S$ we speak of \emph{pre-} and \emph{posttransitions}.
The \emph{token replacement function} $\marking{\_\!\_}:T\rightarrow \Int^S$
is given by $\marking{t}=\postcond{t}-\precond{t}$ for all $t\in T$.
These functions extend to finite signed multisets
as usual (see \refdf{multiset}).
\end{definition}

\begin{definition}{firing}
Let $N \mathbin= (S, T, F, M_0, \ell)$ be a Petri net,
$G \inp \nat^T\!$, $G$ non-empty and finite, and $M, M' \in \nat^S$.\\
$G$ is a \defitem{step} from $M$ to $M'$,
written \plat{$M~[G\rangle_N~ M'$}, iff
\begin{list}{{\bf --}}{\leftmargin 18pt
                        \labelwidth\leftmargini\advance\labelwidth-\labelsep
                        \topsep 0pt \itemsep 0pt \parsep 0pt}
  \item $\precond{G} \leq M$ ($G$ is \defitem{enabled}) and
  \item $M' = (M - \precond{G}) + \postcond{G} = M + \marking{G}$.
\end{list}
\end{definition}
Note that steps are (finite) multisets, thus allowing self-concurrency,
\ie the same transition can occur multiple times in a single step.
We write $M~[t\rangle_N~ M'$ for $M\mathrel{[\{t\}\rangle_N} M'$, whereas
$M [G\rangle_N$ abbreviates $\exists M'.~ M \mathrel{[G\rangle_N} M'$.
We may omit the subscript $N$ if clear from context.

In our nets transitions are labelled with \emph{actions} drawn from a
set \plat{$\Act \dcup \{\tau\}$}. This makes it possible to see these
nets as models of \defitem{reactive systems} that interact with their
environment. A transition $t$ can be thought of as the occurrence of
the action $\ell(t)$. If $\ell(t)\inp\Act$, this occurrence can be
observed and influenced by the environment, but if $\ell(t)\mathbin=\tau$,
it cannot and $t$ is an \defitem{internal} or \defitem{silent} transition.
Transitions whose occurrences cannot be distinguished by the
environment carry the same label. In particular, since
the environment cannot observe the occurrence of internal
transitions at all, they are all labelled $\tau$.

The labelling function $\ell$ extends to finite multisets of transitions $G\in\Int^T$
by $\ell(G):=\sum_{t\in T}G(t)\cdot\{\ell(t)\}$. For $A,B\in\Int^{\Act_\tau}$
we write $A\equiv B$ iff $\ell(A)(a)=\ell(B)(a)$ for all $a\in\Act$, i.e.\ iff $A$ and $B$
contain the same (numbers of) visible actions, allowing $\ell(A)(\tau)\neq \ell(B)(\tau)$.
Hence $\ell(G)\equiv\emptyset$ indicates that $\ell(t)=\tau$ for all transitions $t\in T$ with $G(t)\neq 0$.

\begin{definition}{onesafe}
  Let $N = (S, T, F, M_0, \ell)$ be a Petri net.
\begin{list}{{\bf --}}{\leftmargin 18pt
                        \labelwidth\leftmargini\advance\labelwidth-\labelsep
                        \topsep 0pt \itemsep 0pt \parsep 0pt}
\item
  The set $[M_0\rangle_N$ of \defitem{reachable markings of $N$} is defined as the
  smallest set containing $M_0$ that is closed under $[G\rangle_N$, meaning that if
  $M \in [M_0\rangle_N$ and $M \mathrel{[G\rangle_N} M'$ then $M' \in [M_0\rangle_N$.
\item
  $N$ is \defitem{one-safe}
  iff $M \in [M_0\rangle_N \implies \forall s \in S.~ M(s) \leq 1$.
\item
  The \defitem{concurrency relation} $\mathord{\concurrent} \subseteq
  T^2$ is given by $t \concurrent u \equivalent \exists
  M \inp [M_0\rangle.~ M [\{t\}\mathord+\{u\}\rangle$.
\item
  $N$ is a \hypertarget{scn}{\defitem{structural conflict net}} iff
  for all $t,u\in T$ with $t\smile u$ we have $\precond{t} \cap \precond{u} = \emptyset$.
\end{list}
\end{definition}
We use the term \hypertarget{plain}{\defitem{plain nets}} for Petri nets where $\ell$ is
injective and no transition has the label $\tau$, \ie essentially unlabelled nets. 

This paper first of all aims at studying finite Petri nets: nets with finitely many places
and transitions. However, our work also applies to infinite nets with the properties that
$\precond{t} \ne \varnothing$ for all transitions $t\in T$, and
any reachable marking (a) is finite, and (b) enables only finitely many transitions.
Henceforth, we call such nets \hypertarget{finitary}{\emph{finitary}}.
Finitariness can be ensured by requiring $|M_0| \mathbin< \infty \wedge \forall t \in T.\,
\precond{t} \ne \varnothing \wedge \forall x \in S\cup T.\, |\postcond{x}| < \infty$, \ie
that the initial marking is finite, no transition has an empty set of preplaces, and each
place and transition has only finitely many outgoing arcs.

\section{Semantic Equivalences}\label{sec-equivalences}

In this section, we give an overview on some semantic equivalences for reactive systems. Most of these may be defined  formally for Petri nets in a uniform way, by first defining equivalences for transition systems and then associating different transition systems with a Petri net. This yields in particular different non-interleaving equivalences for Petri nets.

\newcommand{\lts}{\mathfrak{L}}
\newcommand{\st}{\mathfrak{S}}
\newcommand{\tr}{\mathfrak{T}}
\newcommand{\inist}{\mathfrak{M_0}}
\newcommand{\mm}{\mathfrak{M}}
\newcommand{\act}{\mathfrak{Act}}

\begin{definition}{LTS}
Let $\act$ be a set of \emph{visible actions} and
$\tau\mathbin{\not\in}\act$ be an \emph{invisible action}. Let $\act_\tau:=\act \dcup \{\tau\}$.
A \emph{labelled transition system} (LTS) (\emph{over $\act_\tau$}) is a triple
$\lts=(\st,\tr,\inist)$ with
\begin{list}{{\bf --}}{\leftmargin 18pt
                        \labelwidth\leftmargini\advance\labelwidth-\labelsep
                        \topsep 0pt \itemsep 0pt \parsep 0pt}
\item $\st$ a set of \emph{states},
\item $\tr\subseteq \st\times \act_\tau \times \st$ a \emph{transition relation}
\item and $\inist\in\st$ the \emph{initial state}.
\end{list}
\end{definition}
Given an LTS $(\st,\tr,\inist)$ with $\mm,\mm'\in\st$ and $\alpha\in\act_\tau$,
we write $\mm \goesto[\alpha] \mm'$ for $(\mm,\alpha,\mm')\in \tr$.
We write $\mm \goesto[\alpha]$ for $\exists \mm'.~ \mm \goesto[\alpha] \mm'$ and
$\mm \arrownot\goesto[\alpha]$ for $\nexists \mm'.~ \mm \goesto[\alpha] \mm'$.
Furthermore, $\mm \goesto[\opt{\alpha}] \mm'$ denotes
$\mm \goesto[\alpha] \mm' \vee (\alpha\mathbin=\tau \wedge \mm\mathbin=\mm')$,
meaning that in case \mbox{$\alpha\mathbin=\tau$} performing a $\tau$-transition is optional.
      For $\,a_1 a_2 \cdots a_n \in \act^*$ we write
      $\mm \Goesto[\,a_1 a_2 \cdots a_n~] \mm'$ when\vspace{-4pt}
      \[
      \mm
      \Goesto \production{a_1}
      \Goesto \production{a_2}
      \Goesto \cdots
      \Goesto \production{a_n}
      \Goesto
      \mm'\vspace{-4pt}
      \]
      where $\Goesto$ denotes the reflexive and transitive closure of $\goesto[\tau]$.
  A state $\mm \in \st$ is said to be \defitem{reachable} iff there is a
  $\sigma \in \act^*$ such that $\inist \Production{\sigma} \mm$. The set of all
  reachable states is denoted by $[\inist\rangle$.
In case there are $\mm_i\in [\inist\rangle$ for all $i\geq 1$ with
$\mm_1 \goesto[\tau] \mm_2 \goesto[\tau] \cdots$
the LTS is said to display \emph{divergence}.

Many semantic equivalences on LTSs that in some way abstract from internal transitions are
defined in the literature; an overview can be found in \cite{vanglabbeek93linear}.  On
divergence-free LTSs, the most discriminating semantics in the spectrum of equivalences of
\cite{vanglabbeek93linear}, and the only one that fully respects the branching structure
of related systems, is \emph{branching bisimilarity}, proposed in \cite{GW96}.

\begin{definition}{branching LTS}
Two LTSs $(\st_1,\tr_1,\inist_1)$ and $(\st_2,\tr_2,\inist_2)$ are
\emph{branching bisimilar} iff there exists a relation $\Rel
\subseteq \st_1 \times \st_2$---a \emph{branching bisimulation}---such
that, for all $\alpha\inp\act_\tau$:
\begin{enumerate}[~~1.]
\item $\mathfrak{M_0}_1\Rel \mathfrak{M_0}_2$;
\item if $\mathfrak{M}_1\Rel \mathfrak{M}_2$ and
  $\mathfrak{M}_1\!\goesto[\alpha]\mathfrak{M}'_1$
  then $\exists \mathfrak{M}^\dagger_2,\mathfrak{M}'_2$ such that
  $\mathfrak{M}_2\Goesto[] \mathfrak{M}^\dagger_2 \!\goesto[\opt{\alpha}] \mathfrak{M}'_2$,
  ~$\mathfrak{M}_1\Rel \mathfrak{M}^\dagger_2$ and $\mathfrak{M}'_1\Rel \mathfrak{M}'_2$;
\item if $\mathfrak{M}_1\Rel \mathfrak{M}_2$ and
  $\mathfrak{M}_2\!\goesto[\alpha]\mathfrak{M}'_2$
  then $\exists \mathfrak{M}^\dagger_1,\mathfrak{M}'_1$ such that
  $\mathfrak{M}_1\Goesto[] \mathfrak{M}^\dagger_1 \!\goesto[\opt{\alpha}] \mathfrak{M}'_1$,
  ~$\mathfrak{M}^\dagger_1\Rel \mathfrak{M}_2$ and $\mathfrak{M}'_1\Rel \mathfrak{M}'_2$.
\end{enumerate}
\end{definition}
Branching bisimilarity \emph{with explicit divergence} \cite{GW96,GLT09}, is a variant of
branching bisimilarity that fully respects the diverging behaviour of related
systems. Since in this paper we mainly compare systems of which one admits no divergence at all, the
definition simplifies to the requirement that the other system may not diverge either.

One of the semantics reviewed in \cite{vanglabbeek93linear} that respects branching time
and divergence only to a small extent, is \emph{readiness equivalence}, proposed in \cite{OH86}.

\begin{definition}{readiness}
  Let $\lts = (\st,\tr,\inist)$ be an LTS, $\sigma \in \act^*$ and $X \subseteq \act$.
  $\rpair{\sigma, X}$ is a \defitem{ready pair} of $\lts$ iff\vspace{-3pt}
  $$\exists \mm.~ \inist \Production{\sigma} \mm \wedge \mm \arrownot\production{\tau}
  \wedge \, X = \{a\inp \act \mid \mm \goesto[a]\}.$$\\[-3ex]
  We write $\mathfrak{R}(\lts)$ for the set of all ready pairs of $\lts$.\\
  Two LTSs $\lts_1$ and $\lts_2$ are \defitem{readiness equivalent}
  iff $\mathfrak{R}(\lts_1) = \mathfrak{R}(\lts_2)$.
\end{definition}

As indicated in \cite{vanglabbeek01refinement}, see in particular the diagram on Page 317 (or 88),
equivalences on LTSs have been ported to Petri nets and other causality respecting models
of concurrency chiefly in five ways: we distinguish \emph{interleaving semantics},
\emph{step semantics}, \emph{split semantics}, \emph{ST-semantics} and \emph{causal semantics}.
Causal semantics fully respect the causal relationships between the actions of related
systems, whereas interleaving semantics fully abstract from this information.
Step semantics differ from interleaving semantics by taking into account the possibility
of multiple actions to occur simultaneously (in \emph{one step}); this carries a minimal
amount of causal information.
ST-semantics respect causality to the extent that it can be expressed in terms of
the possibility of durational actions to overlap in time. They are formalised by executing
a visible action $a$ in two phases: its start $a^+$ and its termination $a^-$.
Moreover, terminating actions are properly matched with their starts. Split semantics are a
simplification of ST-semantics in which the matching of starts and terminations is dropped.

Interleaving semantics on Petri nets can be formalised by associating to each net
$N=(S,T,F,M_0,\ell)$ the LTS $(\st,\tr,M_0)$ with $\st$ the set of markings of $N$
and $\tr$ given by 
$$M_1 \production{\alpha} M_2 :\equivalent \exists\, t \mathbin\in T
      .~ \alpha \mathbin= \ell(t) \wedge M_1~[t\rangle~ M_2.$$
Here we take $\act := \Act$.
Now each equivalence on LTSs from \cite{vanglabbeek93linear} induces a corresponding
interleaving equivalence on nets by declaring two nets equivalent iff the associated LTSs are.
For example, \emph{interleaving branching bisimilarity} is the relation of \refdf{branching LTS}
with the $\mm$'s denoting markings, and the $\alpha$'s actions from $\Act_\tau$.

Step semantics on Petri nets can be formalised by associating another LTS to each net.
Again we take $\st$ to be the markings of the net, and $\inist$ the initial marking,
but this time $\act$ consists of the \emph{steps} over $\Act$,
the non-empty, finite multisets $A$ of visible actions from $\Act$,
and the transition relation $\tr$ is given by\vspace{-1ex}
      $$M_1 \production{A} M_2 :\equivalent \exists\, G \fin\nat^T
      .~ A = \ell(G)\wedge M_1~[G\rangle~ M_2$$
with $\tau$-transitions defined just as in the interleaving case.
In particular, the step version of readiness equivalence would be the relation of \refdf{readiness}
with the $\mm$'s denoting markings, the $a$'s steps over $\Act$, and the $\sigma$'s sequences of steps.
However, variations in this type of definition are possible.
In this paper, following \cite{glabbeek08syncasyncinteractionmfcs}, we employ 
a form of step readiness semantics that is 
a bit closer to interleaving semantics: $\sigma$ is
a sequence of single actions, whereas the menu $X$ of possible continuations after
$\sigma$ is a set of steps.

\begin{definition}{step readiness}
  {Let $N = (S, T, F, M_0, \ell)$ be a Petri net, $\sigma \in \Act^*$ and
  $X \subseteq \powermultiset{\Act}$.}
  $\rpair{\sigma, X}$ is a \defitem{step ready pair} of $N$ iff\vspace{-3pt}
  $$\exists M. M_0 \Production{\sigma} M \wedge M \arrownot\production{\tau}
  \wedge \, X = \{A\inp \nat^\Act \mid M \goesto[A]\}.$$\\[-3ex]
  We write $\readyset(N)$ for the set of all step ready pairs of $N$.\\
  Two Petri nets $N_1$ and $N_2$ are \defitem{step readiness equivalent},
  $N_1 \approx_\mathscr{R} N_2$, iff $\readyset(N_1) = \readyset(N_2)$.
\end{definition}

Next we propose a general definition on Petri nets of ST-versions of each of the semantics
of \cite{vanglabbeek93linear}. Again we do this through a mapping from nets to a suitable LTS\@.
An \emph{ST-marking} of a net $(S,T,F,M_0,\ell)$ is a pair $(M,U)\inp\nat^S \mathord\times T^*$
of a normal marking, together with a sequence of transitions \emph{currently firing}.
The \emph{initial} ST-marking is $\mathfrak{M_0}:=(M_0,\varepsilon)$.
The elements of $\Act^\pm:=\{a^+,\, a^{-n} \mid a \inp \Act, ~n\mathbin> 0\}$ are called
\emph{visible action phases}, and \plat{$Act^\pm_\tau:=\Act^\pm\dcup\{\tau\}$}.
For $U\in T^*$, we write $t\in^{(n)}U$ if $t$ is the  \plat{$n^{\it th}$}
element of $U$. Furthermore $U^{-n}$ denotes $U$ after removal of the \plat{$n^{\it th}$}
transition.

\begin{definition}{ST-marking}
Let $N=(S,T,F,M_0,\ell)$ be a Petri net, labelled over \plat{$\Act_\tau$}.

The \emph{ST-transition relations} $\goesto[\eta]$ for $\eta\inp\Act^\pm_\tau$ between ST-markings are given by\vspace{1pt}

$(M,U)\goesto[a^+](M',U')$ iff $\exists t\inp T.~ \ell(t)=a \wedge M[t\rangle
\wedge M'=M-\precond{t} \wedge U'=U t$.

$(M,U)\goesto[a^{-n}](M',U')$ iff $\exists t\in^{(n)} U.~\ell(t)=a \wedge
 U'=U^{-n} \wedge M'=M+\postcond{t}$.

$(M,U)\goesto[\tau](M',U')$ iff $M\goesto[\tau]M' \wedge U'=U$.
\end{definition}
Now the ST-LTS associated to a net $N$ is $(\st,\tr,\inist)$ with $\st$ the set of
ST-markings of $N$, $\act:=\Act^\pm$, $\tr$ as defined in \refdf{ST-marking}, and $\inist$ the initial ST-marking.
Again, each equivalence on LTSs from \cite{vanglabbeek93linear} induces a corresponding
ST-equivalence on nets by declaring two nets equivalent iff their associated LTSs are.
In particular, \emph{branching ST-bisimilarity} is the relation of \refdf{branching LTS}
with the $\mm$'s denoting ST-markings, and the $\alpha$'s action phases from $\Act^\pm_\tau$.
We write $N_1\approx^\Delta_{bSTb}N_2$ iff $N_1$ and $N_2$ are branching ST-bisimilar with explicit divergence.

\emph{ST-bisimilarity} was originally proposed in \cite{GV87}. It was extended to a setting
with internal actions in \cite{Vo93}, based on the notion of \emph{weak bisimilarity} of
\cite{Mi89}, which is a bit less discriminating than branching bisimilarity.
The above can be regarded as a reformulation of the same idea; the notion of weak
ST-bisimilarity defined according to the recipe above agrees with the ST-bisimilarity of \cite{Vo93}.

The next proposition says that branching ST-bisimilarity with explicit divergence is more
discriminating than (\ie \emph{stronger} than, \emph{finer} than, or included in) step readiness equivalence.

\begin{proposition}{step ready ST}
Let $N_1$ and $N_2$ be Petri nets. If $N_1\approx^\Delta_{bSTb}N_2$ then $N_1 \approx_\mathscr{R} N_2$.
\end{proposition}
\begin{proof}
Suppose $N_1\approx^\Delta_{bSTb}N_2$ and $\rpair{\sigma, X}\in\readyset(N_1)$.
By symmetry it suffices to show that $\rpair{\sigma, X}\in\readyset(N_2)$.

There must be a branching bisimulation $\Rel$ between the ST-markings of
$N_1=(S_1,T_1,F_1,{M_0}_1,\ell_1)$ and $N_2=(S_2,T_2,F_2,{M_0}_2,\ell_2)$.
In particular, $({M_0}_1,\epsilon)\Rel({M_0}_2,\epsilon)$.
Let $\sigma := a_1 a_2 \cdots a_n \in \Act^*$.
Then \plat{$
      {M_0}_1
      \Goesto \production{a_1}
      \Goesto \production{a_2}
      \Goesto \cdots
      \Goesto \production{a_n}
      \Goesto
      M'_1
      $}\vspace{-4pt}
for a marking $M'_1\inp\nat^{S_1}$ with
\plat{$X = \{A\inp \nat^\Act \mid M'_1 \goesto[A]\}$} and $M'_1 \arrownot\production{\tau}$.
Hence $
      ({M_0}_1,\epsilon)
      \Goesto \production{a_1^+}\production{a_1^{-1}}
      \Goesto \production{a_2^+}\production{a_2^{-1}}
      \Goesto \cdots
      \Goesto \production{a_n^+}\production{a_n^{-1}}
      \Goesto
      (M'_1,\epsilon)
      $.
Thus, using the properties of a branching bisimulation on the ST-LTSs associated to $N_1$
\vspace{-3pt}and $N_2$, there must be a marking \plat{$M'_2\inp\nat^{S_2}$} such that $
      ({M_0}_2,\epsilon)\!
      \Goesto \production{a_1^+}\production{a_1^{-1}}
      \Goesto \production{a_2^+}\production{a_2^{-1}}
      \Goesto \cdots
      \Goesto \production{a_n^+}\production{a_n^{-1}}
      \Goesto\!
      (M'_2,\epsilon)
      $
and $(M'_1,\epsilon)\Rel(M'_2,\epsilon)$.
Since \plat{$(M'_1,\epsilon) \arrownot\production{\tau}$}, the ST-marking $(M'_1,\epsilon)$
admits no divergence. As \plat{$\approx^\Delta_{bSTb}$} respects this property, also $(M'_2,\epsilon)$
admits no divergence, and there must be an $M''_2\inp\nat^{S_2}$ with $M''_2 \arrownot\production{\tau}$
and $(M'_2,\epsilon)\Goesto(M''_2,\epsilon)$. Clause 3.\ of a branching bisimulation
gives $(M'_1,\epsilon)\Rel(M''_2,\epsilon)$, and
\refdf{ST-marking} yields ${M_0}_2 \Goesto[\sigma] M''_2$.\vspace{-3pt}

Now let $B=\{b_1,\ldots,b_n\}\in X$. Then $M'_1\production{B}$, so
$(M'_1,\epsilon)\production{b_1^+}\production{b_2^+}\cdots \production{b_m^+}$.\vspace{-3pt}
Property 2.\ of a branching bisimulation implies
$(M''_2,\epsilon)\production{b_1^+}\production{b_2^+}\cdots \production{b_m^+}$
and hence $M''_2\production{B}$. Likewise, with Property 3., $M''_2\production{B}$ implies 
\plat{$M'_1\production{B}$} for all $B\in\nat^\Act$.
It follows that $\rpair{\sigma, X}\in\readyset(N_2)$.
\end{proof}

In this paper we employ both step readiness equivalence and 
branching ST-bisimilarity with explicit divergence.
Fortunately it will turn out that for our purposes the latter equivalence coincides
with its split version
(since always one of the compared nets is plain, see \refpr{split}).

A \emph{split marking} of a net $N=(S,T,F,M_0,\ell)$ is a pair $(M,U)\in\nat^S \times\nat^T$
of a normal marking $M$, together with  a multiset of transitions currently firing.
The \emph{initial} split marking is $\mathfrak{M_o}:=(M_0,\emptyset)$.
A split marking can be regarded as an abstraction from an ST-marking, in which the total
order on the (finite) multiset of transitions that are currently firing has been dropped.
Let $\Act^\pm_{\rm split}:=\{a^+,\, a^- \mid a \in \Act\}$.

\begin{definition}{split marking}
Let $N=(S,T,F,M_0,\ell)$ be a Petri net, labelled over $\Act_\tau$.

The \emph{split transition relations} $\goesto[\zeta]$ for $\zeta\inp\Act^\pm_{\rm split}\dcup\{\tau\}$
between split markings are given by

$(M,U)\goesto[a^+](M',U')$ iff $\exists t\inp T.~ \ell(t)=a \wedge M[t\rangle
\wedge M'=M-\precond{t} \wedge U'=U + \{t\}$.

$(M,U)\goesto[a^{-}](M',U')$ iff $\exists t\inp U.~\ell(t)=a \wedge
 U'=U-\{t\} \wedge M'=M+\postcond{t}$.

$(M,U)\goesto[\tau](M',U')$ iff $M\goesto[\tau]M' \wedge U'=U$.
\end{definition}
Note that $(M,U)\goesto[a^+]$ iff $M\goesto[a]$, whereas $(M,U)\goesto[a^-]$
iff $a\in\ell(U)$.
With induction on reachability of markings it is furthermore easy to check that
$(M,U) \in [\inist\rangle$ iff $\ell(U)\in\nat^\Act$ and $M+\!\precond{U}\in [M_0\rangle$.

The split LTS associated to a net $N$ is $(\st,\tr,\inist)$ with $\st$ the set of split
markings of $N$, $\act:=\Act^\pm$, $\tr$ as defined in \refdf{split marking}, and $\inist$ the initial split marking.
Again, each equivalence on LTSs from \cite{vanglabbeek93linear} induces a corresponding
split equivalence on nets by declaring two nets equivalent iff their associated LTSs are.
In particular, \emph{branching split bisimilarity} is the relation of \refdf{branching LTS}
with the $\mm$'s denoting split markings, and the $\alpha$'s action phases from
\plat{$\Act^\pm_{\rm split}\dcup\{\tau\}$}.
\vspace{2pt}

For $\mathfrak{M}=(M,U)\in \nat^S\times T^*$ an ST-marking, let
$\overline{\mathfrak{M}}=(M,\overline{U})\in \nat^S \times \nat^T$ be the split marking obtained by
converting the sequence $U$ into the multiset $\overline{U}$, where $\overline{U}(t)$ is
the number of occurrences of the transition $t\in T$ in $U$.
Moreover, define $\ell(\mathfrak{M})$ by $\ell(M,U) := \ell(U)$ and
$\ell(t_1 t_2 \cdots t_k) := \ell(t_1) \ell(t_2) \cdots \ell(t_k)$.
Furthermore, for $\eta\in\Act^\pm_\tau$, let \plat{$\overline{\eta}\in \Act^\pm_{\rm split}\dcup\{\tau\}$}
be given by \plat{$\overline{a^+}:= a^+$}, \plat{$\overline{a^{-n}}:= a^-$} and $\overline{\tau}:= \tau$.

\begin{observation}{match}
Let $\mathfrak{M},\mathfrak{M}'$ be ST-markings, $\mathfrak{M}^\dagger$ a split marking,
$\eta\inp\Act^\pm_\tau$ and \plat{$\zeta\in\Act^\pm_{\rm split}\cup\{\tau\}$}. Then
\begin{list}{{\bf --}}{\leftmargin 18pt
                        \labelwidth\leftmargini\advance\labelwidth-\labelsep
                        \topsep 0pt \itemsep 0pt \parsep 0pt}
\item $\mathfrak{M}\in \nat^S\times T^*$ is the initial ST-marking of $N$ iff
$\overline\mathfrak{M}\in \nat^S\times \nat^T$ is the initial split marking of $N$;
\item if $\mathfrak{M} \goesto[\eta] \mathfrak{M}'$ then
$\overline\mathfrak{M} \goesto[\overline{\eta}] \overline{\mathfrak{M}'}$;
\item if $\overline\mathfrak{M} \goesto[\zeta] \mathfrak{M}^\dagger$ then there
is a $\mathfrak{M}'\in \nat^S\times T^*$ and \plat{$\eta\in\Act^\pm_\tau$} such that $\mathfrak{M} \goesto[\eta]
\mathfrak{M}'$, $\overline{\eta}=\zeta$ and $\overline{\mathfrak{M}'}=\mathfrak{M}^\dagger$;
\item if $\mathfrak{M} \goesto[\opt{\eta}] \mathfrak{M}'$ then
$\overline\mathfrak{M} \goesto[\overline{\opt{\eta}}] \overline{\mathfrak{M}'}$;
\item if $\overline\mathfrak{M} \goesto[\opt{\zeta}] \mathfrak{M}^\dagger$ then there
is a $\mathfrak{M}'\in \nat^S\times T^*$ and \plat{$\eta\in\Act^\pm_\tau$} such that $\mathfrak{M} \goesto[\opt{\eta}]
\mathfrak{M}'$, $\overline{\eta}=\zeta$ and $\overline{\mathfrak{M}'}=\mathfrak{M}^\dagger$;
\item if $\mathfrak{M} \Goesto \mathfrak{M}'$ then
$\overline\mathfrak{M} \Goesto \overline{\mathfrak{M}'}$;
\item if $\overline\mathfrak{M} \Goesto \mathfrak{M}^\dagger$ then there
is a $\mathfrak{M}'\in \nat^S\times T^*$ such that $\mathfrak{M} \Goesto
\mathfrak{M}'$ and $\overline{\mathfrak{M}'}=\mathfrak{M}^\dagger$.\hfill $\Box$
\end{list}
\end{observation}

\begin{lemma}{label sequence}
Let $N_1=(S_1,T_1,F_1,{M_0}_1,\ell)$ and $N_2=(S_2,T_2,F_2,{M_0}_2,\ell_2)$ be two nets, $N_2$ being
\hyperlink{plain}{plain};
let $\mathfrak{M}_1,\mathfrak{M}'_1$ be ST-markings of $N_1$, and
$\mathfrak{M}_2,\mathfrak{M}'_2$ ST-markings of $N_2$.
  If $\ell(\mathfrak{M}_2)=\ell(\mathfrak{M}_1)$,
  $\mathfrak{M}_1\goesto[\eta] \mathfrak{M}'_1$ and
  $\mathfrak{M}_2\goesto[\opt{\eta'}] \mathfrak{M}'_2$ with $\overline{\eta'}=\overline{\eta}$,
  then there is an $\mathfrak{M}''_2$ with $\mathfrak{M}_2\goesto[\opt{\eta}] \mathfrak{M}''_2$,
  $\ell(\mathfrak{M}''_2)=\ell(\mathfrak{M}'_1)$,
  and $\overline{\mathfrak{M}''_2}=\overline{\mathfrak{M}'_2}$.
\end{lemma}

\begin{proof}
If $\mathfrak{M}\goesto[\eta] \mathfrak{M}'$ or $\mathfrak{M}\goesto[\opt{\eta}] \mathfrak{M}'$
then $\ell( \mathfrak{M}')$ is completely determined by $\ell(\mathfrak{M})$ and $\eta$.
For this reason the requirement
$\ell(\mathfrak{M}''_2)=\ell(\mathfrak{M}'_1)$ will hold as soon as
the other requirements are met. 

First suppose $\eta$ is of the form $\tau$ or $a^+$. Then
$\overline{\eta}=\eta$ and moreover $\overline{\eta'}=\overline{\eta}$ implies $\eta'=\eta$.
Thus we can take $\mathfrak{M}''_2:=\mathfrak{M}'_2$.

Now suppose $\eta:=a^{-n}$ for some $n>0$. Then $\eta'=a^{-m}$ for
some $m>0$. As $\mathfrak{M}_1\goesto[\eta]$, the $n^{\it th}$ element
of $\ell(\mathfrak{M}_1)$ must (exist and) be $a$.
Since $\ell(\mathfrak{M}_2)=\ell(\mathfrak{M}_1)$, also the $n^{\it th}$ element
of $\ell(\mathfrak{M}_2)$ must be $a$, so there is an
$\mathfrak{M}''_2$ with $\mathfrak{M}_2\goesto[\opt{\eta}]
\mathfrak{M}''_2$. Let $\mathfrak{M}_2:=(M_2,U_2)$. Then $U_2$ is a
sequence of transitions of which the $n^{\it th}$ and the $m^{\it th}$ elements
are both labelled $a$. Since the net $N_2$ is plain, those two
transitions must be equal. Let $\mathfrak{M}'_2:=(M'_2,U'_2)$ and $\mathfrak{M''}_2:=(M''_2,U''_2)$.
We find that $M''_2\mathbin=M'_2$ and $\overline{U''_2}\mathbin=\overline{U'_2}$. It
follows that $\overline{\mathfrak{M}''_2}=\overline{\mathfrak{M}'_2}$.
\end{proof}

\begin{observation}{label sequence}
If $\mathfrak{M}\Goesto \mathfrak{M}'$ for ST-markings
$\mathfrak{M},\mathfrak{M}'$ then $\ell(\mathfrak{M}')=\ell(\mathfrak{M})$.
\end{observation}

\begin{observation}{end-phase determinism 1}
  If $\ell(\mathfrak{M}_1)=\ell(\mathfrak{M}_2)$ and
  $\mathfrak{M}_2\goesto[a^{-n}]$ for some $a\in\Act$ and $n>0$, then $\mathfrak{M}_1\goesto[a^{-n}]$.
\end{observation}

\begin{observation}{end-phase determinism 2}
  If $\mathfrak{M}\goesto[a^{-n}]\mathfrak{M}'$ and $\mathfrak{M}\goesto[a^{-n}]\mathfrak{M}''$
 for some $a\in\Act$ and $n>0$, then $\mathfrak{M}'_1=\mathfrak{M}'_2$.
\end{observation}

\begin{proposition}{split}
Let $N_1=(S_1,T_1,F_1,{M_0}_1,\ell)$ and
$N_2=(S_2,T_2,F_2,{M_0}_2,\ell_2)$ be two nets, $N_2$ being \hyperlink{plain}{plain}.
Then $N_1$ and $N_2$ are branching ST-bisimilar (with explicit
divergence) iff they are branching split bisimilar (with explicit divergence).
\end{proposition}

\begin{proof}
Suppose $\Rel $ is a branching ST-bisimulation between $N_1$ and $N_2$.
Then, by \refobs{match}, the relation $\Rel _{\rm split} := \{(\overline{\mathfrak{M}_1},\overline{\mathfrak{M}_2}) \mid
(\mathfrak{M}_1,\mathfrak{M}_2)\in \Rel \}$ is a branching split bisimulation between $N_1$ and $N_2$.

Now let $\Rel $ be a branching split bisimulation between $N_1$ and $N_2$.
Then, using \refobs{match}, the relation $\Rel _{\rm ST} := \{(\mathfrak{M}_1,\mathfrak{M}_2) \mid
\ell_1(\mathfrak{M}_1)=\ell_2(\mathfrak{M}_2) \wedge 
(\overline{\mathfrak{M}_1},\overline{\mathfrak{M}_2})\in \Rel \}$ turns out to be a
branching ST-bisimulation between $N_1$ and $N_2$:
\begin{enumerate}[~~1.]
\item $\mathfrak{M_0}_1\Rel_{\rm ST} \mathfrak{M_0}_2$ follows from \refobs{match}, using that
  $\overline{\mathfrak{M_0}_1}\Rel \overline{\mathfrak{M_0}_2}$ and
  $\ell(\mathfrak{M_0}_1)\mathbin=\ell(\mathfrak{M_0}_2)\mathbin=\epsilon$.
\item Suppose $\mathfrak{M}_1\Rel_{\rm ST} \mathfrak{M}_2$ and
  $\mathfrak{M}_1\!\goesto[\eta]\mathfrak{M}'_1$.
  Then $\overline{\mathfrak{M}_1}\Rel \overline{\mathfrak{M}_2}$ and
  $\overline{\mathfrak{M}_1}\!\goesto[\overline\eta]\overline{\mathfrak{M}'_1}$.
  Hence $\exists \mathfrak{M}^\dagger_2,\mathfrak{M}^\ddagger_2$ such that
  $\overline{\mathfrak{M}_2}\Goesto[] \mathfrak{M}^\dagger_2 \!\goesto[\opt{\overline\eta}] \mathfrak{M}^\ddagger_2$,
  ~$\overline{\mathfrak{M}_1}\Rel \mathfrak{M}^\dagger_2$ and
  $\overline{\mathfrak{M}'_1}\Rel \mathfrak{M}^\ddagger_2$.
  As $N_2$ is plain, $\mathfrak{M}^\dagger_2=\overline{\mathfrak{M}_2}$.
  By \refobs{match}, using that $\overline{\mathfrak{M}_2}\goesto[\opt{\overline\eta}] \mathfrak{M}^\ddagger_2$,
  $\exists \mathfrak{M}'_2,\,\eta'$ such that
  $\mathfrak{M}_2\goesto[\opt{\eta'}] \mathfrak{M}'_2$, $\overline{\eta'}=\overline{\eta}$ and
  $\overline{\mathfrak{M}'_2} = \mathfrak{M}^\ddagger_2$.
  By \reflem{label sequence}, there is an ST-marking $\mathfrak{M}''_2$ such that
  $\mathfrak{M}_2\goesto[\opt{\eta}] \mathfrak{M}''_2$,
  $\ell(\mathfrak{M}''_2)=\ell(\mathfrak{M}'_1)$,
  and $\overline{\mathfrak{M}''_2}=\overline{\mathfrak{M}'_2} = \mathfrak{M}^\ddagger_2$.
  It follows that $\mathfrak{M}'_1\Rel_{\rm ST} \mathfrak{M}''_2$.
\item Suppose $\mathfrak{M}_1\Rel_{\rm ST} \mathfrak{M}_2$ and
  $\mathfrak{M}_2\!\goesto[\eta]\mathfrak{M}'_2$.
  Then $\overline{\mathfrak{M}_1}\Rel \overline{\mathfrak{M}_2}$ and
  $\overline{\mathfrak{M}_2}\!\goesto[\overline\eta]\overline{\mathfrak{M}'_2}$.
  Hence $\exists \mathfrak{M}^\dagger_1,\mathfrak{M}^\ddagger_1$ such that
  $\overline{\mathfrak{M}_1}\Goesto[] \mathfrak{M}^\dagger_1 \!\goesto[\opt{\overline\eta}] \mathfrak{M}^\ddagger_1$,
  ~$\mathfrak{M}^\dagger_1\Rel \overline{\mathfrak{M}_2}$ and
  $\mathfrak{M}^\ddagger_1\Rel \overline{\mathfrak{M}'_2}$.
  By \refobs{match}, $\exists \mathfrak{M}^*_1$ such that
  $\mathfrak{M}_1\Goesto[] \mathfrak{M}^*_1$ and
  $\overline{\mathfrak{M}^*_1} = \mathfrak{M}^\dagger_1$.
  By \refobs{label sequence},
  $\ell(\mathfrak{M}^*_1)=\ell(\mathfrak{M}_1)=\ell(\mathfrak{M}_2)$,
  so $\mathfrak{M}^*_1\Rel_{\rm ST} \mathfrak{M}_2$.
  Since $N_2$ is plain, $\eta\neq\tau$.
  \begin{itemize}
  \item Let $\eta=a^+$ for some $a\in\Act$.
  Using that $\overline{\mathfrak{M}^*_1}\goesto[\opt{\overline\eta}] \mathfrak{M}^\ddagger_1$,
  by \refobs{match} $\exists \mathfrak{M}'_1,\,\eta'$ such that
  $\mathfrak{M}^*_1\goesto[\opt{\eta'}] \mathfrak{M}'_1$, $\overline{\eta'}=\overline{\eta}$ and
  $\overline{\mathfrak{M}'_1} = \mathfrak{M}^\ddagger_1$.
  It must be that $\eta'=\eta=a^+$ and
  $\ell(\mathfrak{M}'_1)=\ell(\mathfrak{M}^*_1)a=\ell(\mathfrak{M}_2)a=\ell(\mathfrak{M}'_2)$.
  Hence $\mathfrak{M}'_1\Rel_{\rm ST} \mathfrak{M}'_2$.
\item Let $\eta=a^{-n}$ for some $a\in\Act$ and $n>0$.
  By \refobs{end-phase determinism 1}, $\exists \mathfrak{M}'_1$ with
  $\mathfrak{M}^*_1\goesto[\eta]\mathfrak{M}'_1$.
  By Part 2.\ of this proof, $\exists \mathfrak{M}''_2$ such that
  $\mathfrak{M}_2\goesto[\opt{\eta}] \mathfrak{M}''_2$ and
  $\mathfrak{M}'_1\Rel_{\rm ST} \mathfrak{M}''_2$.
  By \refobs{end-phase determinism 2} $\mathfrak{M}''_2=\mathfrak{M}'_2$.
  \end{itemize}
\end{enumerate}
Since the net $N_2$ is plain, it has no divergence. In such a case,
the requirement ``with explicit divergence'' requires $N_1$ to be free
of divergence as well, regardless of whether split or ST-semantics is
in used.
\end{proof}
In this paper we will not consider causal semantics.
The reason is that our distributed implementations will not fully preserve the causal behaviour of nets.
We will further comment on this in the conclusion.

\section{Distributed Systems}
\label{sec-distributed systems}
In this section, we stipulate what we understand by a distributed
system, and subsequently formalise a model of distributed systems in
terms of Petri nets.
\begin{list}{{\bf --}}{\leftmargin 18pt
                        \labelwidth\leftmargini\advance\labelwidth-\labelsep
                        \topsep 3pt \itemsep 0pt \parsep 0pt}
\item A distributed system consists of components residing on different locations.
\item Components work concurrently.
\item Interactions between components are only possible by explicit communications.
\item Communication between components is time consuming and asynchronous.
\end{list}
Asynchronous communication is the
only interaction mechanism in a
distributed system for exchanging signals or information.
\begin{list}{{\bf --}}{\leftmargin 18pt
                        \labelwidth\leftmargini\advance\labelwidth-\labelsep
                        \topsep 3pt \itemsep 0pt \parsep 0pt}
\item The sending of a message happens always strictly before its receipt (there is a causal relation between sending and receiving a message).
\item A sending component sends without regarding the state of the
  receiver; in particular there is no need to synchronise with a receiving component.
  After sending the sender continues its behaviour independently of receipt of the message.
\end{list}
As explained in the introduction, we will add another requirement to
our notion of a distributed system, namely that its components only
allow sequential behaviour.

Formally, we model distributed systems as nets consisting of component
nets with sequential behaviour and interfaces in terms of input and
output places.

\begin{definition}{component}
  Let $N \mathbin= (S, T, F, M_0, \ell)$ be a Petri net,
  $I, O \mathbin\subseteq S$,
  $I\mathop\cap O\mathbin=\emptyset$ and
  $\postcond{O} = \varnothing$.
  \begin{list}{{\bf --}}{\leftmargin 18pt
                        \labelwidth\leftmargini\advance\labelwidth-\labelsep
                        \topsep 0pt \itemsep 0pt \parsep 0pt}
    \item[1.] $(N, I, O)$ is a \defitem{component with interface $(I, O)$}.
    \item[2.] $(N, I, O)$ is a \defitem{sequential} component with interface $(I,
      O)$ iff\\ $\exists Q \mathbin\subseteq S \mathord\setminus (I \cup O)$ with
      $\forall t \in T. |\precond{t} \restrictedto Q| = 1 \wedge
          |\postcond{t}\! \restrictedto Q| = 1$ and
        $|M_0 \restrictedto Q| = 1$.
  \end{list}
\end{definition}
An input place $i\inp I$ of a component $\mathcal{C}\mathbin=(N,I,O)$ can be regarded as a mailbox of
$\mathcal{C}$ for a specific type of messages.  An output place $o\inp
O$, on the other hand, is an address outside $\mathcal{C}$ to which $\mathcal{C}$
can send messages. Moving a token into $o$ is like posting a
letter. The condition $\postcond{o}=\varnothing$ says that a message,
once posted, cannot be retrieved by the component.

A set of places like $Q$ above is called an \emph{$S$-invariant}.
The requirements guarantee that the number of tokens in these places
remains constant, in this case $1$. It follows that no two transitions
can ever fire concurrently (in one step).
Conversely, whenever a net is sequential, in the sense that no two
transitions can fire in one step, it is easily converted into a
behaviourally equivalent net with the required $S$-invariant, namely by
adding a single marked place with a self-loop to all transitions.
This modification preserves virtually all semantic equivalences on
Petri nets from the literature, including $\approx^\Delta_{bSTb}$.

Next we define an operator for combining components with asynchronous
communication by fusing input and output places.

\begin{definition}{parcomp}
  Let $\indexset$ be an index set.\\
  Let $((S_k, T_k, F_k, {M_0}_k, \ell_k), I_k, O_k)$ with $k \in \indexset$
  be components with interface such that
  $(S_k \cup T_k) \cap (S_l \cup T_l) = (I_k \cup O_k) \cap (I_l \cup O_l)$
  for all $k, l \in \indexset$ with $k\neq l$
  (components are disjoint except for interface places)
  and $I_k \cap I_l = \varnothing$ for all $k, l \in \indexset$ with $k\neq l$
  (mailboxes cannot be shared; any message has a unique recipient).

\noindent
  Then the \defitem{asynchronous parallel composition} of
  these components is defined by\vspace{-.5ex}
  \[
  \Big\|_{i \in \indexset} ((S_k, T_k, F_k, {M_0}_k, \ell_k), I_k, O_k) =
    ((S, T, F, {M_0}, \ell), I, O)\vspace{-.5ex}
  \]
  with
  $S \mathord= \bigcup_{k \in \indexset} S_k,~
  T \mathord= \bigcup_{k \in \indexset} \!T_k,~
  F \mathord= \bigcup_{k \in \indexset} F_k,~
  M_0 \mathord= \sum_{k \in \indexset} {M_0}_k,~
  \ell \mathord= \bigcup_{k \in \indexset} \ell_k$
  (componentwise union of all nets),
  $I \mathord= \bigcup_{k \in \indexset} I_k$
  (we accept additional inputs from outside), and
  $O \mathord= \bigcup_{k \in \indexset} O_k \setminus \bigcup_{k \in \indexset} I_k$
  (once fused with an input, $o\inp O_I$ is no longer an output).
\end{definition}

\begin{observation}{associativity}
  $\|$ is associative.
\end{observation}
  This follows directly from the associativity of the (multi)set union
  operator.\hfill$\Box$

\noindent
We are now ready to define the class of nets representing systems 
of asynchronously communicating sequential components.

\begin{definition}{LSGA}
  $\!$A Petri net $N$ is an \defitem{LSGA net} (a \defitem{locally sequential
  globally asynchronous net}) iff there exists an index set
  $\indexset$ and sequential components with interface
  $\mathcal{C}_k,~ k \inp \indexset$, such that $(N, I, O) = \|_{k \in \indexset} \mathcal{C}_k$
  for some $I$ and $O$.
\end{definition}

\noindent
Up to $\approx^\Delta_{bSTb}$---or any reasonable equivalence
preserving causality and branching time but abstracting from internal
activity---the same class of LSGA systems would have been obtained if
we had imposed, in \refdf{component}, that $I$, $O$ and $Q$ form
a partition of $S$ and that $\precond{I}=\emptyset$.
However, it is essential that our definition allows multiple
transitions of a component to read from the same input place.

In the remainder of this section we give a more abstract
characterisation of Petri nets representing distributed systems,
namely as \emph{distributed} Petri nets, which we introduced in
\cite{glabbeek08syncasyncinteractionmfcs}. This will be useful in
\refsec{distributable}, where we investigate
distributability using this more semantic characterisation. We show
below that the concrete characterisation of distributed systems as
LSGA nets and this abstract characterisation agree.

Following \cite{BCD02}, to arrive at a class of nets representing distributed systems,
we associate \emph{localities} to the elements of a net $N=(S,T,F,M_0,\ell)$.
We model
this by a function \mbox{$D: S\cup T \rightarrow\Loc$}, with $\Loc$ a
set of possible locations.  We refer to such a
function as a \defitem{distribution} of $N$.  Since the identity of
the locations is irrelevant for our purposes, we can just as well
abstract from $\Loc$ and represent $D$ by the equivalence relation
$\equiv_D$ on $S\cup T$ given by $x \equiv_D y$ iff $D(x)=D(y)$.

Following \cite{glabbeek08syncasyncinteractionmfcs}, we impose a fundamental
restriction on distributions, namely that when two transitions
can occur in one step, they cannot be co-located. This reflects our
assumption that at a given location \visible actions can only occur
sequentially.

In \cite{glabbeek08syncasyncinteractionmfcs} we observed that
Petri nets incorporate a notion of synchronous interaction, in that a
transition can fire only by synchronously taking the tokens from all
of its preplaces. In general the behaviour of a net would change
radically if a transition would take its input tokens one by one---in
particular deadlocks may be introduced. Therefore we insist that in a
distributed Petri net, a transition and all its input places reside on
the same location. There is no reason to require the same for the
output places of a transition, for the behaviour of a net would not
change significantly if transitions were to deposit their output tokens
one by one \cite{glabbeek08syncasyncinteractionmfcs}.
\pagebreak[3]

This leads to the following definition of a distributed Petri net.

\begin{definition}{distributed}$\!$\cite{glabbeek08syncasyncinteractionmfcs}\,
    A Petri net $N = (S, T, F, M_0, \ell)$ is \defitem{distributed}
    iff there exists a distribution $D$ such that
  \begin{list}{{\bf --}}{\leftmargin 25pt
                        \labelwidth\leftmargini\advance\labelwidth-\labelsep
                        \topsep 0pt \itemsep 0pt \parsep 0pt}
    \item[(1)]
      $\forall s \in S, ~t \in T.~\hspace{1pt}s \in \precond{t}
      \implies t \equiv_D s$,
    \item[(2)]
     $\forall t,u\in T.~t \concurrent u \implies t\not\equiv_D u$.
  \end{list}
\end{definition}

\noindent
A typical example of a net which is not distributed is shown in
\reffig{fullM} on Page \pageref{fig-fullM}.
Transitions $t$ and $v$ are concurrently executable
and hence should be placed on different locations. However,
both have preplaces in common with $u$ which would enforce putting all
three transitions on the same location. In fact, distributed nets can
be characterised in the following semi-structural way.

\begin{observation}{distributed}
A Petri net is distributed iff there is no sequence $t_0,\ldots,t_n$ of
transitions with $t_0 \smile t_n$ and
$\precond{t_{i-1}}\cap\precond{t_{i}}\neq\emptyset$ for $i=1,\ldots,n$.\hfill$\Box$
\end{observation}

\noindent
We proceed to show that the classes of LSGA nets and distributable
nets essentially coincide. That every LSGA net is distributed follows
because we can place each sequential component on a
separate location. The following two lemmas constitute a formal argument.
Here we call a component with interface $(N,I,O)$ distributed iff $N$
is distributed.

\begin{lemma}{sequential component distributed}
Any sequential component with interface is distributed.
\end{lemma}
\begin{proof}
As a sequential component displays no concurrency,
it suffices to co-locate all places and transitions.
\end{proof}

\noindent
\reflem{parcompdistributed} states that the class of distributed nets is
closed under asynchronous parallel composition.

\begin{lemma}{parcompdistributed}
  Let $\mathcal{C}_k=(N_k, I_k, O_k)$, $k \inp \indexset$, be components with
  interface, satisfying the requirements of \refdf{parcomp},
  which are all distributed.
  Then $\|_{k \in \indexset} \mathcal{C}_k$ is distributed.
\end{lemma}
\begin{proof}
  We need to find a distribution $D$ satisfying the requirements of
  \refdf{distributed}.

  Every component $\mathcal{C}_k$
  is distributed and hence comes with a distribution $D_k$.
  Without loss of generality the codomains of all $D_k$ can be assumed
  disjoint.

  Considering each $D_k$ as a function from net elements onto locations,
  a partial function $D_k'$ can be defined which does not map any places in
  $O_k$, denoting that the element may be located arbitrarily, and
  behaves as $D_k$ for all other elements.
  As an output place has no posttransitions within a component, any total
  function larger than (i.e. a superset of) $D_k'$ is still a valid
  distribution for $N_k$.

  Now $D' = \bigcup_{k \in \indexset} D_k'$ is a (partial) function, as every place
  shared between components is an input place of at most one.
  The required distribution $D$ can be chosen as any total function extending $D'$;
  it satisfies the requirements of \refdf{distributed} since
  the $D_k$'s do.
\end{proof}

\begin{corollary}{LSGA distributed}
Every LSGA net is distributed.\hfill$\Box$
\end{corollary}

\noindent
Conversely, any distributed net $N$ can be transformed in an LSGA net by
choosing co-located transitions with their pre- and postplaces as
sequential components and declaring any place that belongs to multiple
components to be an input place of component $N_k$ if it is a preplace
of a transition in $N_k$, and an output place of component $N_l$ if it
is a postplace of a transition in $N_l$ and not an input place of $N_l$.
Furthermore, in order to guarantee that the components are sequential
in the sense of \refdf{component}, an explicit control place is
added to each component---without changing behaviour---as explained
below \refdf{component}. It is straightforward to check that the
asynchronous parallel composition of all so-obtained components is an
LSGA net, and that it is equivalent to $N$ (using $ \approx_\mathscr{R}$,
$\approx^\Delta_{bSTb}$, or any other reasonable equivalence).

\begin{theorem}{bothdistributedequal}
  For any distributed net $N$ there is an LSGA net $N'$ with $N'\approx^\Delta_{bSTb} N$.
\end{theorem}
\begin{proof}
  Let $N = (S, T, F, M_0, \ell)$ be a distributed net with a distribution $D$.
  Then an equivalent LSGA net $N'$ can be constructed
  by composing sequential components with interfaces as follows.
  
  For each equivalence class $[x]$ of net elements according to $D$ a
  sequential component $(N_{[x]}, I_{[x]}, O_{[x]})$ is created.
  Each such component contains one new and initially marked place $p_{[x]}$
  which is connected via self-loops to all transitions in $[x]$.
  The interface of the component is formed by
  $I_{[x]} := (S \cap [x])$\footnote{Alternatively, we could take
  $I_{[x]} := \postcond{(T\backslash[x])}\cap [x]$.}
  and
  $O_{[x]} := \postcond{([x] \cap T)} \setminus [x]$.
  Formally,
  $N_{[x]} := (S_{[x]}, T_{[x]}, F_{[x]}, {M_0}_{[x]}, \ell_{[x]})$
  with
  \begin{itemize}
    \item $S_{[x]} = ((S \cap [x]) \cup O_{[x]} \cup \{p_{[x]}\}$,
    \item $T_{[x]} = T \cap [x]$,
    \item $F_{[x]} = F \restrictedto (S_{[x]} \cup T_{[x]})^2 \cup
      \{(p_{[x]}, t), (t, p_{[x]}) \mid t \in T_{[x]}\}$,
    \item ${M_0}_{[x]} = (M_0 \restrictedto [x]) \cup \{p_{[x]}\}$, and
    \item $\ell_{[x]} = \ell \restrictedto [x]$.
  \end{itemize}
  All components overlap at interfaces only, as the sole
  places not in an interface are the newly created $p_{[x]}$.
  The $I_{[x]}$ are disjoint as the equivalence classes $[x]$ are, so
  $(N',I',O') := \|_{[x] \in (S \cup T) / D} (N_{[x]}, O_{[x]}, I_{[x]})$ is well-defined.
  It remains to be shown that $N' \approx^\Delta_{bSTb} N$.
  The elements of $N'$ are exactly those of $N$ plus the new places $p_{[x]}$,
  which stay marked continuously except when a transition from $[x]$ is firing,
  and never connect two concurrently enabled transitions. Hence there exists a bijection
  between the ST-markings of $N'$ and $N$ that preserves the ST-transition relations
  between them, \ie the associated ST-LTSs are isomorphic. From this it follows that
  $N' \approx^\Delta_{bSTb} N$.
\end{proof}

\begin{observation}{distributed-structuralconflict}
  Every distributed Petri net is a \hyperlink{scn}{structural conflict net}.\hfill$\Box$
\end{observation}

\begin{corollary}{LSGA-structuralconflict}
  Every LSGA net is a structural conflict net.\hfill$\Box$
\end{corollary}

\noindent
Further on, we use  a more liberal definition of a
distributed net, called \emph{essentially distributed}.
We will show that up to $\approx^\Delta_{bSTb}$ any essentially
distributed net can be converted into a distributed net.
In \cite{glabbeek08syncasyncinteractionmfcs} we employed an even more liberal definition of a
distributed net, which we call here \emph{externally distributed}.
Although we showed that up to step readiness equivalence any externally
distributed net can be converted into a distributed net, this does not hold for
$\approx^\Delta_{bSTb}$.

\begin{definition}{externally distributed}
  A net $N=(S,T,F,M_0,\ell)$ is \emph{essentially distributed} iff
  there exists a distribution $D$ satisfying (1) of
  \refdf{distributed} and
  \begin{list}{{\bf --}}{\leftmargin 25pt
                        \labelwidth\leftmargini\advance\labelwidth-\labelsep
                        \topsep 0pt \itemsep 0pt \parsep 0pt}
    \item[($2'$)] $\forall t,u\in T.~t \concurrent u \wedge \ell(t)\neq\tau \implies t\not\equiv_D u$.
  \end{list}
  It is \emph{externally distributed} iff there exists a distribution $D$ satisfying (1) and
  \begin{list}{{\bf --}}{\leftmargin 25pt
                        \labelwidth\leftmargini\advance\labelwidth-\labelsep
                        \topsep 0pt \itemsep 0pt \parsep 0pt}
    \item[($2''$)] $\forall t,u\in T.~t \concurrent u \wedge \ell(t),\ell(u)\neq\tau \implies t\not\equiv_D u$.
  \end{list}
\end{definition}
Instead of ruling out co-location of concurrent transitions in general,
essentially distributed nets permit concurrency of internal transitions---labelled $\tau$---at the same location.
Externally distributed nets even allow concurrency between external and internal transitions at the same location.
If the transitions $t$ and $v$ in the net of \reffig{fullM} would both be labelled $\tau$,
the net would be essentially distributed, although not distributed; in case only $v$ would
be labelled $\tau$ the net would be externally distributed but not essentially distributed.
Essentially distributed nets need not be structural conflict nets; in fact, \emph{any} net
without external transitions is essentially distributed.

The following proposition says that up to $\approx^\Delta_{bSTb}$ any essentially
distributed net can be converted into a distributed net.
\begin{proposition}{essentiallydistributedequal}
  For any essentially distributed net $N$ there is a distributed net $N'$ with $N'\approx^\Delta_{bSTb} N$.
\end{proposition}
\begin{proof}
The same construction as in the proof of \refthm{bothdistributedequal} applies: $N'$
differs from $N$ by the addition, for each location $[x]$, of a marked place $p_{[x]}$
that is connected through self-loops to all transitions at that location.  This time there
exists a bijection between the \emph{reachable} ST-markings of $N'$ and $N$ that preserves the
ST-transition relations between them. This bijection exists because a reachable ST-marking
is a pair $(M,U)$ with $U$ a sequence of \emph{external} transitions only; this follows by
a straightforward induction on reachability by ST-transitions. From this it follows that
$N' \approx^\Delta_{bSTb} N$.
\end{proof}
Likewise, up to $\approx_\mathscr{R}$ any externally distributed net can be converted into a distributed net.
\begin{proposition}{externallydistributedequal}\cite{glabbeek08syncasyncinteractionmfcs}\,\,
  For any externally distributed net $N$ there is a distributed net $N'$ with $N'\approx_\mathscr{R} N$.
\end{proposition}
\begin{proof}
  Again the same construction applies. This time there exists a bijection
  between the markings of $N'$ and $N$ that preserves the step transition relations
  between them, \ie the associated step transition systems are isomorphic.
  Here we use that the transitions in the associated LTS involve either a multiset of
  concurrently firing \emph{external} transitions, or a single internal one.
  From this, step readiness equivalence follows.
\end{proof}
The counterexample in \reffig{externally distributed} shows that up to
$N' \approx^\Delta_{bSTb} N$ not any externally distributed net can be
converted into a distributed net.
Sequentialising the component with actions $a$, $b$
and $\tau$ would disable the execution $\goesto[a^+]\Goesto\goesto[c^+]$.

\begin{figure}[h]
\begin{minipage}[b]{0.38\linewidth}
  \begin{center}
    \begin{petrinet}(6,3,4)
      \Q (2,3):p1:$p$;
      \Q (4,3):p2:$q$;
      \u (1,1):t1:$a$:$t$;
      \u (3,1):t2:$b$:$u$;
      \u (5,1):t3:$c$:$v$;

      \a p1->t1;
      \a p1->t2;
      \a p2->t2;
      \a p2->t3;
    \end{petrinet}
  \end{center}
  \vspace{-4ex}
  \caption{A fully marked \structuralM.}
  \label{fig-fullM}
\end{minipage}\hfill
\begin{minipage}[b]{0.55\linewidth}
  \begin{center}
    \begin{petrinet}(10,3,4)
      \Q (2,3):p1:$p$;
      \Q (4,3):p2:$q$;
      \u (1,1):t1:$a$:$t$;
      \u (3,1):t2:$b$:$u$;
      \ub (5,1):t3:$\tau$:$v$;
      \qb (7,1):p3:$r$;
      \ub (9,1):t4:$c$:$w$;

      \a p1->t1;
      \a p1->t2;
      \a p2->t2;
      \a p2->t3;
      \a t3->p3;
      \a p3->t4;
    \end{petrinet}
  \end{center}
  \vspace{-4ex}
  \caption{Externally distributed, but not distributable.}
  \label{fig-externally distributed}
\end{minipage}
  \vspace{-2ex}
\end{figure}

\begin{definition}{canonical}
Given any Petri net $N$, the \emph{canonical co-location relation} $\equiv_C$ on $N$ is the
equivalence relation on the places and transitions of $N$
\emph{generated} by Condition (1) of \refdf{distributed},
i.e.\ the smallest equivalence relation $\equiv_D$ satisfying (1).
The \emph{canonical distribution} of $N$ is the distribution $C$ that maps
each place or transition to its $\equiv_C$-equivalence class.
\end{definition}
\begin{observation}{canonical}
A Petri net that is distributed (resp.\ essentially or externally
distributed) w.r.t.\ any distribution $D$, is
distributed (resp.\ essentially or externally distributed)
w.r.t.\ its canonical distribution.
\end{observation}
Hence a net is distributed (resp.\ essentially or externally distributed) iff its canonical
distribution $D$ satisfies Condition (2) of \refdf{distributed}
(resp.\ Condition ($2'$) or ($2''$) of \refdf{externally distributed}).

\section{Distributable Systems}
\label{sec-distributable}

We now consider Petri nets as specifications of concurrent
systems and ask the question which of those specifications can be
implemented as distributed systems. This question can be formalised as
\begin{quote}\em
Which Petri nets are semantically equivalent to distributed nets?
\end{quote}
Of course the answer depends on the choice of a suitable
semantic equivalence. 
Here we will answer this question using the two equivalences discussed in the
introduction.
We will give a precise characterisation of those nets for which we can find semantically equivalent distributed nets. For the negative part of this characterisation, stating that certain nets are not distributable, we will use step readiness
equivalence, which is one of the simplest and least discriminating
equivalences imaginable that abstracts from internal actions, but
preserves branching time, concurrency and divergence to some small
degree. As explained in \cite{glabbeek08syncasyncinteractionmfcs}, giving
up on any of these latter three properties would make any Petri net
distributable, but in a rather trivial and
unsatisfactory way.
For the positive part, namely that all other nets are indeed distributable, 
we will use the most discriminating equivalence for which our implementation works, namely branching
ST-bisimilarity with explicit divergence, which is finer than
step readiness equivalence. Hence we will obtain the strongest possible results for both directions and it turns out that the concept of distributability is fairly robust w.r.t.\ the choice of
a suitable equivalence: any equivalence notion between step readiness
equivalence and branching ST-bisimilarity with explicit divergence will yield
the same characterisation.

\begin{definition}{distributable}
  A Petri net $N$ is \defitem{distributable} up to an equivalence
  $\approx$ iff there exists a distributed net $N'$ with $N' \approx N$.
\end{definition}

\noindent
Formally we give our characterisation of distributability by classifying which
finitary plain structural conflict nets can be implemented as distributed nets,
and hence as LSGA nets. In such implementations, we use invisible
transitions. We study the concept ``distributable'' for plain nets only, but in
order to get the largest class possible we allow non-plain implementations,
where a given transition may be split into multiple transitions carrying the
same label.

It is well known that sometimes a global protocol is necessary to implement
synchronous interaction present in system specifications. In
particular, this may be needed for deciding choices in a coherent
way, when these choices require agreement of multiple components.
The simple net in \reffig{fullM} shows a typical situation of this
kind. Independent decisions of the two choices might lead to a
deadlock. As remarked in \cite{glabbeek08syncasyncinteractionmfcs}, for
this particular net there exists no satisfactory distributed
implementation that fully respects the reactive behaviour of the
original system. Indeed such \structuralM-structures, representing
interference between concurrency and choice, turn out to play a crucial
r\^ole for characterising distributability.

\begin{definition}{fullM}
  Let $N = (S, T, F, M_0, \ell)$ be a Petri net.
  $N$ has a \hypertarget{M}{\defitem{fully reachable \visible pure \structuralM}} iff\\
  $\exists t,u,v \in T.
  \precond{t} \cap \precond{u} \ne \varnothing \wedge
  \precond{u} \cap \precond{v} \ne \varnothing \wedge
  \precond{t} \cap \precond{v} = \varnothing \wedge
  \exists M \in [M_0\rangle.
  \precond{t} \cup \precond{u} \cup \precond{v} \subseteq M$.
\end{definition}

\noindent
Note that \refdf{fullM} implies that $t\neq u$, $u\neq v$ and $t\neq v$.

\smallskip

We now give an upper bound on the class of distributable nets by adopting a
result from \cite{glabbeek08syncasyncinteractionmfcs}.

\begin{theorem}{trulysyngltfullm}
  Let $N$ be a plain structural conflict Petri net. 
  If $N$ has a fully reachable \visible pure {\structuralM}, then $N$ is
  not distributable up to step readiness equivalence.
\end{theorem}
\begin{proof}
In \cite{glabbeek08syncasyncinteractionmfcs} this theorem was obtained for
plain one-safe nets.\footnote{In \cite{glabbeek08syncasyncinteractionmfcs} the
theorem was claimed and proven only for plain nets with a fully reachable
\emph{visible} pure {\structuralM}; however, for plain nets the
requirement of visibility is irrelevant.} The proof applies verbatim to plain
structural conflict nets as well.
\end{proof}

\noindent
Since $\approx^\Delta_{bSTb}$ is finer than $\approx_\mathscr{R}$,
this result holds also for distributability up to $\approx^\Delta_{bSTb}$
(and any equivalence between $\approx_\mathscr{R}$ and $\approx^\Delta_{bSTb}$).

In the following, we establish that this upper bound is tight, and hence a finitary plain structural conflict net is distributable iff it has no fully reachable \visible
pure \structuralM. For this, it is helpful to first introduce macros in Petri nets for reversibility of transitions.

\subsection{Petri nets with reversible transitions}

A \emph{Petri net with reversible transitions} generalises the notion
of a Petri net; its semantics is given by a translation to an
ordinary Petri net, thereby interpreting the reversible transitions
as syntactic sugar for certain net fragments.  It is defined as a tuple
$(S,T,\UI,\ui,F,M_0,\ell)$ with $S$ a set of places, $T$ a set of (reversible)
transitions, labelled by \plat{$\ell:T\rightarrow\Act\dcup\{\tau\}$}, $\UI$ a set of
\emph{undo interfaces} with the relation $\ui\subseteq \UI\times T$
linking interfaces to transitions, $M_0 \inp \nat^S$ an initial marking, and
$$F\!: (S\times T\times \{{\scriptstyle \it in,~early,~late,~out,~far}\} \rightarrow \nat)$$
the flow relation.
When $F(s,t,{\scriptstyle \it type})>0$ for ${\scriptstyle \it type} \in
\{{\scriptstyle \it in,~early,~late,~out,~far}\}$, this is depicted by drawing an arc from $s$ to
$t$, labelled with its arc weight $F(s,t,{\scriptstyle \it type})$, of the form\
\psscalebox{0.7}{\begin{pspicture}(1.5,0.2)
  \def\thenetimage{arrowexamples}
  \pnode(0,0.1){narrowexamples-a}
  \pnode(1.5,0.1){narrowexamples-b}
  \a a->b;
\end{pspicture}},
\psscalebox{0.7}{\begin{pspicture}(1.5,0.2)
  \def\thenetimage{arrowexamples}
  \pnode(0,0.1){narrowexamples-a}
  \pnode(1.5,0.1){narrowexamples-b}
  \aEarly a->b;
\end{pspicture}},
\psscalebox{0.7}{\begin{pspicture}(1.5,0.2)
  \def\thenetimage{arrowexamples}
  \pnode(0,0.1){narrowexamples-a}
  \pnode(1.5,0.1){narrowexamples-b}
  \aLate a->b;
\end{pspicture}},
\psscalebox{0.7}{\begin{pspicture}(1.5,0.2)
  \def\thenetimage{arrowexamples}
  \pnode(0,0.1){narrowexamples-a}
  \pnode(1.5,0.1){narrowexamples-b}
  \a b->a;
\end{pspicture}},
\psscalebox{0.7}{\begin{pspicture}(1.5,0.2)
  \def\thenetimage{arrowexamples}
  \pnode(0,0.1){narrowexamples-a}
  \pnode(1.5,0.1){narrowexamples-b}
  \aFar b->a;
\end{pspicture}},
respectively. 
For $t\inp T$ and ${\scriptstyle \it type} \in
\{{\scriptstyle \it in,~early,~late,~out,~far}\}$, the multiset of places
$t^{\it type}\inp\nat^S$ is given by $t^{\it type}(s) = F(s,t,{\scriptstyle \it type})$.
When $s\inp t^{\it type}$ for ${\scriptstyle \it type} \in
\{{\scriptstyle \it in,~early,~late}\}$, the place $s$ is called a
\emph{preplace} of $t$ of type {\scriptsize \it type\/};
when $s\inp t^{\it type}$ for ${\scriptstyle \it type} \in
\{{\scriptstyle \it out,~far}\}$, $s$ is called a
\emph{postplace} of $t$ of type {\scriptsize \it type}.
For each undo interface $\omega\inp \UI$ and transition $t$ with $\ui(\omega,t)$ there must be places
$\undo[\omega](t)$, $\reset[\omega](t)$ and $\ack[\omega](t)$ in $S$.
A transition with a nonempty set of interfaces is called \emph{reversible};
the other (\emph{standard}) transitions may have pre- and postplaces
of types {\scriptsize \it in} and {\scriptsize \it out} only---for these
transitions $t^{\it in}\mathbin=\precond{t}$ and $t^{\it out}\mathbin=\postcond{t}$.
In case $\UI=\emptyset$, the net is just a normal Petri net.

A global state of a Petri net with reversible transitions is given by a marking
$M\inp\nat^S$, together with the state of each reversible transition ``currently
in progress''.  Each transition in the net can fire as usual. A reversible transition can
moreover take back (some of) its output tokens, and be \emph{undone} and
\emph{reset}.  When a transition $t$ fires, it consumes $\sum_{{\scriptstyle \it
type}\in\{{\scriptstyle \it in,~early,~late}\}} F(s,t,{\scriptstyle \it type})$
tokens from each of its preplaces $s$ and produces $\sum_{{\scriptstyle \it
type}\in\{{\scriptstyle \it out,~far}\}} F(s,t,{\scriptstyle \it type})$
tokens in each of its postplaces $s$.  A reversible transition $t$ that has fired
can start its reversal by consuming a token from $\undo[\omega](t)$ for one of
its interfaces $\omega$.  Subsequently, it can
take back one by one a token from its postplaces of type {\scriptsize \it
far}. After it has retrieved all its output of type
{\scriptsize \it far}, the transition is undone, thereby returning
$F(s,t,{\scriptstyle \it early})$ tokens in each of its preplaces $s$ of type
{\scriptsize \it early}.
Afterwards, by
consuming a token from $\reset[\omega](t)$, for the same interface $\omega$ that
started the undo-process, the transition terminates its chain of activities by
returning $F(s,t,{\scriptstyle \it late})$ tokens in each of its {\scriptsize
\it late} preplaces $s$.  At that occasion it also produces a token in
$\ack[\omega](t)$. Alternatively, two tokens in $\undo[\omega](t)$ and
$\reset[\omega](t)$ can annihilate each other without involving the transition
$t$; this also produces a token in $\ack[\omega](t)$. The latter mechanism comes in action
when trying to undo a transition that has not yet fired.

\reffig{reversible} shows the translation of a reversible transition
$t$ with $\ell(t)\mathbin=a$ into an ordinary net fragment.
\begin{figure}[ht]
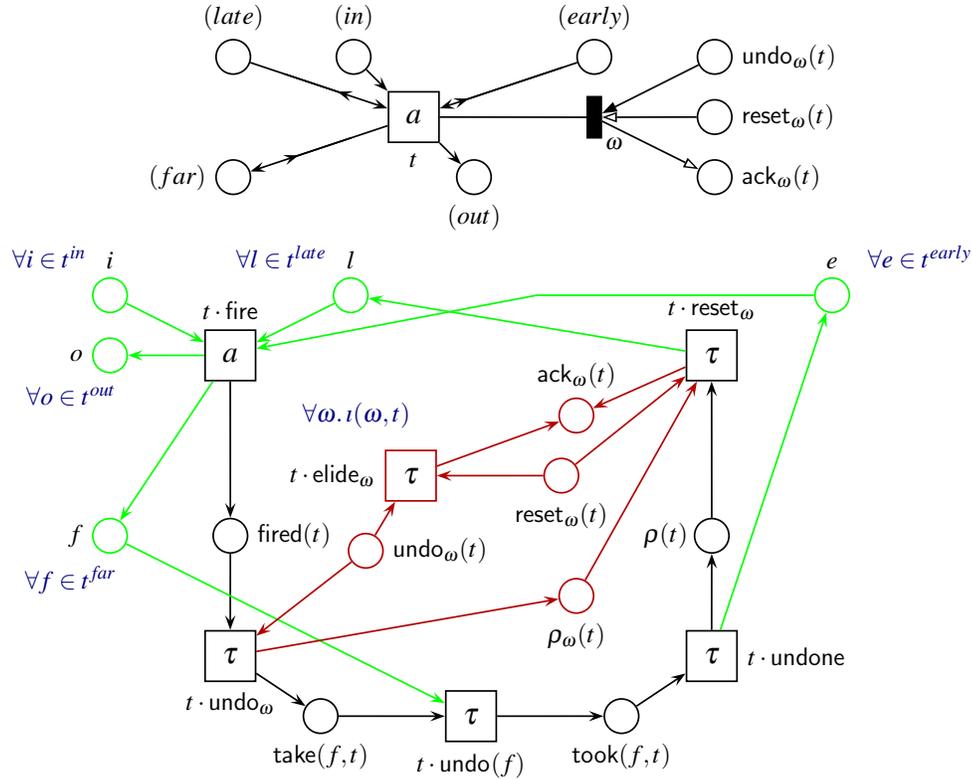

  \begin{center}
    \begin{petrinet}(10,4)
      \qt(3,3):in:$(in)$;
      \qt(1,3):late:$(late)$;
      \qt(7,3):early:$(early)$;
      \q(9,3):undo:$\undo[\omega](t)$;
      \q(9,2):reset:$\reset[\omega](t)$;
      \q(9,1):ack:$\ack[\omega](t)$;
      \ql(1,1):far:$(far)$;
      \qb(5,1):out:$(out)$;
      \ub(4,2):t:$a$:$t$;

      \interface(7,2):ti:t:$\omega$;

      \a in->t;
      \aLate late->t;
      \aEarly early->t;
      \a t->out;
      \aFar t->far;
      \aUndo undo->ti;
      \aReset reset->ti;
      \aAck ti->ack;
    \end{petrinet}
  \end{center}
  \vspace{1ex}
  \begin{center}
    \begin{petrinet}(14,8)
      \psset{linecolor=green}
      \ql(1,4):far:$f$;
      \ql(1,7):out:$o$;
      \qt(1,8):in:$i$;
      \qt(5,8):late:$l$;
      \qt(13,8):early:$e$;

      \psset{linecolor=black}
      \qb(4.5,1):take:$\take(f,t)$;
      \ub(7,1):tundop:$\tau$:$t \cdot \undo[](f)$;
      \qb(9.5,1):took:$\took(f,t)$;

      \ub(3,2):tundoa:$\tau$:$t \cdot \undo[\omega]$;
      \ur(11,2):tundone:$\tau$:$t \cdot \undone$;

      \qr(3,4):fired:$\Fired(t)$;
      \ql(11,4):p2:$\rho(t)\!$;

      \ut(3,7):tfire:$a$:$t \cdot \fire$;
      \ut(11,7):treseta:$\tau$:$t \cdot \reset[\omega]$;

      \psset{linecolor=darkred}
      \qr(5.25,3.75):undoa:$\undo[\omega](t)$;
      \qb(8.75,3):pa:$\keep[\omega](t)$;
      \ul(6,5):elidea:$\tau$:$t \cdot \elide[\omega]$;
      \qt(8.75,6):acka:$\ack[\omega](t)$;
      \qb(8.5,5):reseta:$\reset[\omega](t)$;

      {
        \darkblue
        \rput[tr](1,3.5){\large $\forall f \in t^{\,\it far}$}
        \rput[tr](1,6.5){\large $\forall o \in t^{out}$}
        \rput[rb](0.5,8.45){\large $\forall i \in t^{in}$}
        \rput[rb](4.5,8.45){\large $\forall l \in t^{late}$}
        \rput[lb](13.5,8.45){\large $\forall e \in t^{early}$}
        \rput(5,6){\large $\forall \omega.\, \ui(\omega, t)$}
      }

      \psset{linecolor=black}
      \a tfire->fired;
      \a fired->tundoa;
      \a tundoa->take;
      \a take->tundop;
      \a tundop->took;
      \a took->tundone;
      \a tundone->p2;
      \a p2->treseta;

      \psset{linecolor=green}
      \a in->tfire;
      \a tfire->out;
      \a tfire->far;
      \a far->tundop;
      \a tundone->early;
      \a treseta->late;
      \av early[180]-(8,8)->[15]tfire;
      \a late->tfire;

      \psset{linecolor=darkred}
      \a tundoa->pa;
      \a pa->treseta;
      \a undoa->tundoa;
      \a undoa->elidea;
      \a elidea->acka;
      \a treseta->acka;
      \a reseta->elidea;
      \a reseta->treseta;
    \end{petrinet}
  \end{center}
\vspace{-1ex}
\caption{A reversible transition and its macro expansion.}
\label{fig-reversible}
\vspace{-2pt}
\end{figure}
The arc weights on the green (or grey) arcs are inherited from the
untranslated net; the other arcs have weight~1.
Formally, a net $(S,T,\UI,\ui,F,M_0, \ell)$ with reversible
transitions translates into the Petri net containing all places $S$,
initially marked as indicated by $M_0$, all standard transitions in
$T$, labelled according to $\ell$, along with their pre- and
postplaces, and furthermore all net elements mentioned in
\reftab{reversible}.  Here \hypertarget{Tback}{}$T^\leftarrow$ denotes the set of
reversible transitions in $T$.

\begin{table}[ht]
\vspace*{-2.5ex}
\[
\begin{array}{@{}lclll@{}}
\textbf{Transition} & \textrm{label} & \textrm{Preplaces} & \textrm{Postplaces} & \textrm{for all} \\
\hline\rule[11pt]{0pt}{1pt}
t\cdot\fire   & \ell(t) & t^{in},~ t^{early},~ t^{late} &
                  \Fired(t),~ t^{out}, t^{\,\it far} & t \in T^\leftarrow \\
t\cdot\undo[\omega] & \tau & \undo[\omega](t),~\Fired(t) & \keep[\omega](t),~\take(f,t) &
  t \in T^\leftarrow,~\ui(\omega,t),~f \mathbin\in t^{\,\it far}\\
t\cdot\und    & \tau & \take(f,t),~ f & \took(f,t) & t\in T^\leftarrow,~f \in t^{\,\it far}\\
t\cdot\undone & \tau & \took(f,t) & \rho(t),~ t^{early} & t\in T^\leftarrow,~f \in t^{\,\it far}\\
t\cdot\reset[\omega] & \tau & \reset[\omega](t),~\keep[\omega](t),~\rho(t) & t^{late},~ \ack[\omega](t) &
 t \in T^\leftarrow,~\ui(\omega,t)\\
t\cdot\elide[\omega] & \tau & \undo[\omega](t),~\reset[\omega](t) & \ack[\omega](t) & t \in T^\leftarrow,~\ui(\omega,t)\\
\end{array}
\]
\vspace{-3ex}
\caption{Expansion of a Petri net with reversible transitions into a place/transition system.}
\label{tab-reversible}
\vspace{-3ex}
\end{table}

\subsection{The conflict replicating implementation}\label{sec-implementation}

Now we establish that a finitary plain structural conflict net that has no fully reachable
\visible pure \structuralM\ is distributable.  We do this by proposing the \emph{conflict
replicating implementation} of any such net, and show that this implementation is always
(a) essentially distributed, and (b) equivalent to the original net. In order to get the strongest
possible result, for (b) we use branching ST-bisimilarity with explicit divergence.

\begin{figure}
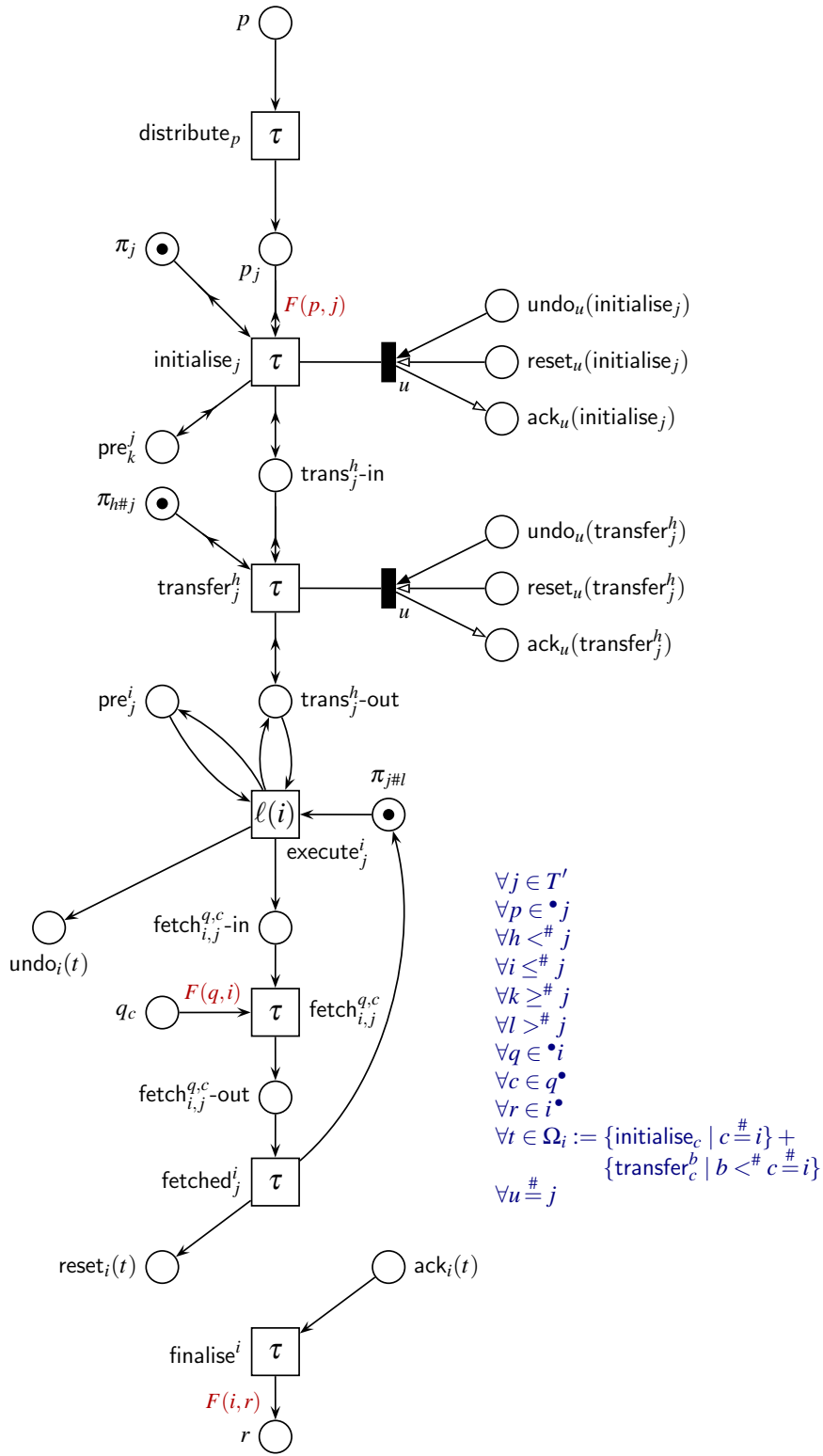

\vspace{-1em}
  \begin{center}
    \begin{petrinet}(17,26.8)
      {
        \darkblue

        \rput[lb](9,5){\large$\begin{array}{l}\displaystyle
          \forall j \in T'\\
          \forall p \in \precond{j}\\
          \forall h <^\# j\\
          \forall i \leq^\# j\\
          \forall k \geq^\# j\\
          \forall l >^\# j\\
          \forall q \in \precond{i}\\
          \forall c \in \postcond{q}\\
          \forall r \in \postcond{i\,}\\
          \forall t \in \UIij := \begin{array}[t]{@{}l@{}}
          \{\ini[c]\mid c\confeq i\} +\mbox{}\\ \{\trans[b]{c} \mid b <^\# c \confeq i\}
          \end{array}\\
          \forall u \confeq j\\
        \end{array}$}
      }

      {
        \darkred

        \rput[l](5.25,21){$F(p,j)$}
        \rput[r](4.85,1.6){$F(i,r)$}
        \rput[b](4,8.65){$F(q,i)$}
      }

      \ql(5,26):p:$p$;
      \ul(5,24):distributep:$\tau$:$\dist$;
      \qx(5,22):pj:$p_j$:(-0.45,-0.45);
      \ql(3,18.5):prejk:$\Pre^j_k$;
      \Ql(3,22):readyinitialisej:$\pi_j$;
      \ul(5,20):initialisej:$\tau$:$\ini$;
      \interface(7,20):initialiseji:initialisej:$u$;
      \qr(9,21):undoinij:$\undo[u](\ini)$;
      \qr(9,20):resetinij:$\reset[u](\ini)$;
      \qr(9,19):ackinij:$\ack[u](\ini)$;
      \qr(5,18):transin:$\transin{j}$;
      \Ql(3,17.5):pconh:$\pi_{h\#j}$;
      \ul(5,16):trans:$\tau$:$\trans{j}$;
      \interface(7,16):transi:trans:$u$;
      \qr(9,17):undotransj:$\undo[u](\trans{j})$;
      \qr(9,16):resettransj:$\reset[u](\trans{j})$;
      \qr(9,15):acktransj:$\ack[u](\trans{j})$;
      \qr(5,14):transout:$\transout{j}$;
      \ql(3,14):prehj:$\Pre^i_j$;
      \Qt(7,12):pconj:$\pi_{j\#l}$;
      \ul(5,12):executehj:\makebox[0pt]{$\ell(i)$}:;
      \rput[l](5.2,11.3){\large $\exec{j}$}
      \qb(1,10):undohjt:$\undo(t)$;
      \ql(5,10):fetchphjin:$\fetchin[q,c]$;
      \ql(3,8.5):pbackbottom:$q_c$;
      \ur(5,8.5):fetchphj:$\tau$:$\fetch[q,c]$;
      \ql(5,7):fetchphjout:$\fetchout[q,c]$;
      \ul(5,5.5):fetchedhj:$\tau$:$\fetched{j}$;
      \qr(7,4):ackhjt:$\ack(t)$;
      \ql(3,4):resethjt:$\reset(t)$;
      \ul(5,2.5):completehj:$\tau$:$\comp{j}$;
      \ql(5,1):r:$r$;

      \a p->distributep;
      \aBack pbacktop->distributep;
      \a distributep->pj;
      \aEarly pj->initialisej;
      \aLate readyinitialisej->initialisej;
      \aFar initialisej->prejk;
      \aFar initialisej->transin;
      \aEarly transin->trans;
      \aLate pconh->trans;
      \aFar trans->transout;
      \B prehj->executehj;
      \B executehj->prehj;
      \A transout->executehj;
      \A executehj->transout;
      \a pconj->executehj;
      \a executehj->undohjt;
      \a executehj->fetchphjin;
      \a fetchphjin->fetchphj;
      \a pbackbottom->fetchphj;
      \a fetchphj->fetchphjout;
      \a fetchphjout->fetchedhj;
      \a fetchedhj->resethjt;
      \a ackhjt->completehj;
      \a completehj->r;
      \BB fetchedhj->pconj;
      \aUndo undoinij->initialiseji;
      \aReset resetinij->initialiseji;
      \aAck initialiseji->ackinij;
      \aUndo undotransj->transi;
      \aReset resettransj->transi;
      \aAck transi->acktransj;
    \end{petrinet}
  \end{center}
  \vspace{-2ex}
  \caption{The conflict replicating implementation}
  \label{fig-conflictrepl}\vspace{-1pt}
\end{figure}

To define the conflict replicating implementation of a net $N=(S,T,F,M_0,\ell)$
we fix an arbitrary well-ordering $<$ on its transitions. We let
$b,c,g,h,i,j,k,l$ range over these ordered transitions, and write
\begin{list}{{\bf --}}{\leftmargin 18pt
                        \labelwidth\leftmargini\advance\labelwidth-\labelsep
                        \topsep 0pt \itemsep 0pt \parsep 0pt}
\item $i\mathbin\#j$ iff ~$i\neq j \wedge \precond{i} \cap \precond{j} \ne \varnothing$
  ~(transitions $i$ and $j$ are \emph{in conflict}),
  ~and $i \confeq j$ iff ~$i\mathbin\# j \vee i\mathbin=j$,
\item $i <^\#\! j$ iff ~$i<j \wedge i\mathbin{\#} j$,
  ~and $i \leqc j$ iff ~$i <^\#\! j \vee i=j$.
\end{list}
\reffig{conflictrepl} shows the conflict replicating implementation of $N$.  It is presented
as a Petri net $$\impl{N}=(S',T',F',\UI,\ui,M'_0,\ell')$$ with reversible transitions.  The set
$\UI$ of undo interfaces is $T$, and for $i\inp \UI$ we
have $\ui(i,t)$ iff $t\inp \UIij$, where the sets of transitions
\plat{$\UI_i\inp\nat^{T'}$} are specified in \reffig{conflictrepl}. The implementation $\impl{N}$
inherits the places of $N$ (\ie $S'\supseteq S$), and we postulate that
$M'_0\mathord\upharpoonright S = M_0$.  Given this, \reffig{conflictrepl} is not merely an
illustration of $\impl{N}$---it provides a complete and accurate description of it, thereby
defining the conflict replicating implementation of any net. In interpreting this figure
it is important to realise that net elements are completely determined by their name
(identity), and exist only once, even if they show up multiple times in the figure. For
instance, the place $\pi_{h\#j}$ with $h\mathord=2$ and $j\mathord=5$ (when using
natural numbers for the transitions in $T$) is the same as the place \plat{$\pi_{j\#l}$}
with $j\mathord=2$ and $l\mathord=5$; it is a standard preplace of \plat{$\exec{2}$} (for
all $i\leqc 2$), a
standard postplace of $\fetched{2}$, as well as a late preplace of \plat{$\trans[2]{5}$}.
A description of this net after expanding the macros for reversible transitions appears in
\reftab{conflictrepl} on Page~\pageref{tab-conflictrepl}.

The r\^ole of the transitions $\dist$ for $p\inp S$ is to distribute a
token in $p$ to copies $p_j$ of $p$ in the localities of all
transitions $j\inp T$ with $p\inp \precond{j}$.  In case $j$ is
enabled in $N$, the transition $\ini$ will become enabled in $\impl{N}$.
These transitions put tokens in the places \plat{$\Pre^j_k$}, which
are preconditions for all transitions \plat{$\exec[j]{k}$}, which
model the execution of $j$ at the location of $k$.  When two
conflicting transitions $h$ and $j$ are both enabled in $N$, the first
steps $\ini[h]$ and $\ini$ towards their execution in $\impl{N}$ can happen
in parallel. To prevent them from executing both, \plat{$\exec[j]{j}$}
(of $j$ at its own location) is only possible after \plat{$\trans{j}$},
which disables $\exec[h]{h}$.

The main idea behind the conflict replicating implementation is that
a transition $h\inp T$ is primarily executed by a sequential component of
its own, but when a conflicting transition $j$ gets enabled, the
sequential component implementing $j$ may ``steal'' the possibility to
execute $h$ from the home component of $h$, and keep the options to do
$h$ and $j$ open until one of them occurs. To prevent $h$ and $j$ from
stealing each other's initiative, which would result in deadlock, a
global asymmetry is built in by ordering the transitions.
Transition $j$ can steal the initiative from $h$ only when $h<j$.

In case $j$ is also in conflict with a transition $l$, with $j<l$,
the initiative to perform $j$ may subsequently be stolen by $l$.
In that case either $h$ and $l$ are in conflict too---then $l$ takes
responsibility for the execution of $h$ as well---or $h$ and $l$ are
concurrent---in that case $h$ will not be enabled, due to the absence
of fully reachable pure \structuralM s in $N$.
The absence of fully reachable pure \structuralM s also guarantees that it
cannot happen that two concurrent transitions $j$ and $k$ both
steal the initiative from an enabled transition $h$.

After the firing of $\exec{j}$ all tokens that were left behind in the process of
carefully orchestrating this firing will have to be cleaned up, in order to prepare the
net for the next activity in the same neighbourhood. This is the reason for
the reversibility of the transitions preparing the firing of \plat{$\exec{j}$}.
Hence there is an undo interface for each transition $i\in T'$, cleaning up the mess made in
preparation of firing $\exec{j}$ for some $j\geq^\# i$. $\UIij$ is the multiset of all transitions $t$ that could possibly have
contributed to this. For each of them the undo interface $i$ is activated, by
\plat{$\exec{j}$} depositing a token in $\undo(t)$. After all preparatory transitions that
have fired are undone, tokens appear in the places $p_c$ for all $p\inp\precond{i}$ and $c\inp\postcond{p}$.
These are collected by $\fetch$, after which all transitions in $\UIij$ get a reset signal.
Those that have fired and were undone are reset, and those that never fired perform $\elide(t)$.
In either case a token appears in $\ack(t)$. These are collected by $\comp{j}$, which
finishes the process of executing $i$ by depositing tokens in its postplaces.
\bigskip

\begin{figure}[hbt]
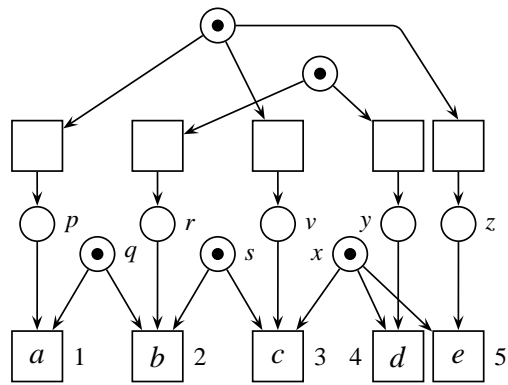

  \begin{center}
    \begin{petrinet}(10, 6)
      \P(4,6):inip1;
      \P(5.7,5.2):inip2;

      \t(1,4):init1:;
      \t(3,4):init2:;
      \t(5,4):init3:;
      \t(7,4):init4:;
      \t(8,4):init5:;

      \qr(1,2.7):p:$p$;
      \Qr(2,2.2):q:$q$;
      \qr(3,2.7):r:$r$;
      \Qr(4,2.2):s:$s$;
      \qr(5,2.7):v:$v$;
      \Ql(6.2,2.2):x:$x$;
      \ql(7,2.7):y:$y$;
      \qr(8,2.7):z:$z$;

      \ur(1,0.5):a:$a$:$1$;
      \ur(3,0.5):b:$b$:$2$;
      \ur(5,0.5):c:$c$:$3$;
      \ul(7,0.5):d:$d$:$4$;
      \ur(8,0.5):e:$e$:$5$;

      \a inip1->init1;
      \a inip1->init3;
      \av inip1[0]-(7,6)->[90]init5;
      \a inip2->init2;
      \a inip2->init4;

      \a init1->p;
      \a init2->r;
      \a init3->v;
      \a init4->y;
      \a init5->z;

      \a p->a;
      \a q->a;
      \a q->b;
      \a r->b;
      \a s->b;
      \a s->c;
      \a v->c;
      \a x->c;
      \a x->d;
      \a y->d;
      \a x->e;
      \a z->e;
    \end{petrinet}
  \end{center}
  \caption{An example net.}
  \label{fig-bigexampleoriginal}
\end{figure}

\begin{landfloat}{figure}{\rotateleft}
  \psscalebox{0.7}{
  \begin{petrinet}(38,22)
    {%
      \psset{linecolor=darkblue}
      \P(13,21):inip1;
      \P(21,21):inip2;

      \t(1,19):init1:;
      \t(9,19):init2:;
      \t(17,19):init3:;
      \t(25,19):init4:;
      \t(31,19):init5:;
    }%

    \ql(1,17):p:$p$;
    \Ql(5,17):q:$q$;
    \ql(9,17):r:$r$;
    \Ql(13,17):s:$s$;
    \ql(17,17):v:$v$;
    \Ql(21,17):x:$x$;
    \ql(25,17):y:$y$;
    \ql(31,17):z:$z$;

    \ur(1,15):distrp:$\tau$:$\dist[p]$;
    \ur(5,15):distrq:$\tau$:$\dist[q]$;
    \ur(9,15):distrr:$\tau$:$\dist[r]$;
    \ur(13,15):distrs:$\tau$:$\dist[s]$;
    \ur(17,15):distrv:$\tau$:$\dist[v]$;
    \ur(21,15):distrx:$\tau$:$\dist[x]$;
    \ur(25,15):distry:$\tau$:$\dist[y]$;
    \ur(31,15):distrz:$\tau$:$\dist[z]$;

    \ql(1,13):p1:$p_1$;
    \qr(3,13):q1:$q_1$;
    \ql(7,13):q2:$q_2$;
    \ql(9,13):r2:$r_2$;
    \qr(11,13):s2:$s_2$;
    \qr(15,13):s3:$s_3$;
    \qr(17,13):v3:$v_3$;
    \qr(19,13):x3:$x_3$;
    \qr(23,13):x4:$x_4$;
    \qr(25,13):y4:$y_4$;
    \ql(29,13):x5:$x_5$;
    \qr(31,13):z5:$z_5$;

    \ur(1,11):ini1:$\tau$:$\ini[1]$;
    \ur(9,11):ini2:$\tau$:$\ini[2]$;
    \ur(17,11):ini3:$\tau$:$\ini[3]$;
    \ur(25,11):ini4:$\tau$:$\ini[4]$;
    \ur(31,11):ini5:$\tau$:$\ini[5]$;

    {%
      \psset{linecolor=darkblue}
      \qr(7,9):trans12in:$\transin[1]{2}$;
      \qr(15,9):trans23in:$\transin[2]{3}$;
      \qr(23,9):trans34in:$\transin[3]{4}$;
      \qr(29,9):trans35in:$\transin[3]{5}$;
      \qr(33,9):trans45in:$\transin[4]{5}$;

      \ur(7,7):trans12:$\tau$:$\trans[1]{2}$;
      \ur(15,7):trans23:$\tau$:$\trans[2]{3}$;
      \ur(23,7):trans34:$\tau$:$\trans[3]{4}$;
      \ur(29,7):trans35:$\tau$:$\trans[3]{5}$;
      \ur(33,7):trans45:$\tau$:$\trans[4]{5}$;

      \ql(7,5):trans12out:$\transout[1]{2}$;
      \ql(15,5):trans23out:$\transout[2]{3}$;
      \ql(23,5):trans34out:$\transout[3]{4}$;
      \ql(32,5):trans35out:$\transout[3]{5}$;
      \qr(35,5):trans45out:$\transout[4]{5}$;
    }%

    \Ql(1,5):pi12:$\pi_{1\#2}$;
    \Ql(9,5):pi23:$\pi_{2\#3}$;
    \Ql(17,5):pi34:$\pi_{3\#4}$;
    \Ql(19,5):pi35:$\pi_{3\#5}$;
    \Ql(25,5):pi45:$\pi_{4\#5}$;

    \ur(1,3):exec11:$a$:$\exec[1]{1}$;
    \ul(7,3):exec12:$a$:$\exec[1]{2}$;
    \ur(9,3):exec22:$b$:$\exec[2]{2}$;
    \ul(15,3):exec23:$b$:$\exec[2]{3}$;
    \ur(17,3):exec33:$c$:$\exec[3]{3}$;
    \ul(23,3):exec34:$c$:$\exec[3]{4}$;
    \ur(25,3):exec44:$d$:$\exec[4]{4}$;
    \ul(31,3):exec35:$c$:$\exec[3]{5}$;
    \ur(34,3):exec45:$d$:$\exec[4]{5}$;
    \ur(37,3):exec55:$e$:$\exec[5]{5}$;

    \qr(1,1):pre11:$\Pre^1_1$;
    \ql(7,1):pre12:$\Pre^1_2$;
    \qr(9,1):pre22:$\Pre^2_2$;
    \ql(15,1):pre23:$\Pre^2_3$;
    \qr(17,1):pre33:$\Pre^3_3$;
    \ql(23,1):pre34:$\Pre^3_4$;
    \qr(25,1):pre44:$\Pre^4_4$;
    \ql(31,1):pre35:$\Pre^3_5$;
    \qr(34,1):pre45:$\Pre^4_5$;
    \qr(37,1):pre55:$\Pre^5_5$;

    {%
      \psset{linecolor=darkblue}
      \a inip1->init1;
      \a inip1->init3;
      \a inip1->init5;
      \a inip2->init2;
      \a inip2->init4;

      \a init1->p;
      \a init2->r;
      \a init3->v;
      \a init4->y;
      \a init5->z;
    }%

    \a p->distrp;
    \a q->distrq;
    \a r->distrr;
    \a s->distrs;
    \a v->distrv;
    \a x->distrx;
    \a y->distry;
    \a z->distrz;

    \a distrp->p1;
    \a distrq->q1;
    \a distrq->q2;
    \a distrr->r2;
    \a distrs->s2;
    \a distrs->s3;
    \a distrv->v3;
    \a distrx->x3;
    \a distrx->x4;
    \a distrx->x5;
    \a distry->y4;
    \a distrz->z5;

    \a p1->ini1;
    \a q1->ini1;
    \a q2->ini2;
    \a r2->ini2;
    \a s2->ini2;
    \a s3->ini3;
    \a v3->ini3;
    \a x3->ini3;
    \a x4->ini4;
    \a y4->ini4;
    \a x5->ini5;
    \a z5->ini5;

    \a ini2->trans12in;
    \a ini3->trans23in;
    \a ini4->trans34in;
    \a ini5->trans35in;
    \a ini5->trans45in;

    {%
      \psset{linecolor=darkblue}
      \a trans12in->trans12;
      \a trans23in->trans23;
      \a trans34in->trans34;
      \a trans35in->trans35;
      \a trans45in->trans45;

      \a trans12->trans12out;
      \a trans23->trans23out;
      \a trans34->trans34out;
      \a trans35->trans35out;
      \a trans45->trans45out;

      \a pi12->trans12;
      \a pi23->trans23;
      \a pi34->trans34;
      \a pi35->trans35;
      \a pi45->trans45;
    }%

    \a pi12->exec11;
    \a pi23->exec12;
    \a pi23->exec22;
    \a pi34->exec23;
    \a pi34->exec33;
    \a pi35->exec23;
    \a pi35->exec33;
    \a pi45->exec34;
    \a pi45->exec44;

    \a trans12out->exec12;
    \a trans12out->exec22;
    \a trans23out->exec23;
    \a trans23out->exec33;
    \a trans34out->exec34;
    \a trans34out->exec44;
    \a trans35out->exec35;
    \a trans35out->exec45;
    \a trans35out->exec55;
    \a trans45out->exec35;
    \a trans45out->exec45;
    \a trans45out->exec55;

    \a pre11->exec11;
    \a pre12->exec12;
    \a pre22->exec22;
    \a pre23->exec23;
    \a pre33->exec33;
    \a pre34->exec34;
    \a pre44->exec44;
    \a pre35->exec35;
    \a pre45->exec45;
    \a pre55->exec55;

    \av ini1[180]-(-1,9)(-1,0)(6,0)->[235]pre12;
    \av ini1[210]-(-0.8,9)(-0.8,0.2)(0,0.2)->[235]pre11;
    \av ini2[180]-(4,9.1)(4,-0.5)(14,-0.5)->[235]pre23;
    \av ini2[200]-(4.2,9)(4.2,-0.3)(8,-0.3)->[235]pre22;
    \av ini3[205]-(12.2,9)(12.2,0.2)(16,0.2)->[235]pre33;
    \av ini3[195]-(12.1,9.05)(12.1,0.1)(22,0.1)->[235]pre34;
    \av ini3[183]-(12,9.1)(12,0)(30,0)->[235]pre35;
    \av ini4[200]-(20.2,9)(20.2,-0.3)(24,-0.3)->[235]pre44;
    \av ini4[180]-(20,9.1)(20,-0.5)(33,-0.5)->[235]pre45;
    \av ini5[200]-(28,9)(28,-0.25)(36,-0.25)->[235]pre55;
  \end{petrinet}
  }
  \vspace{2ex}
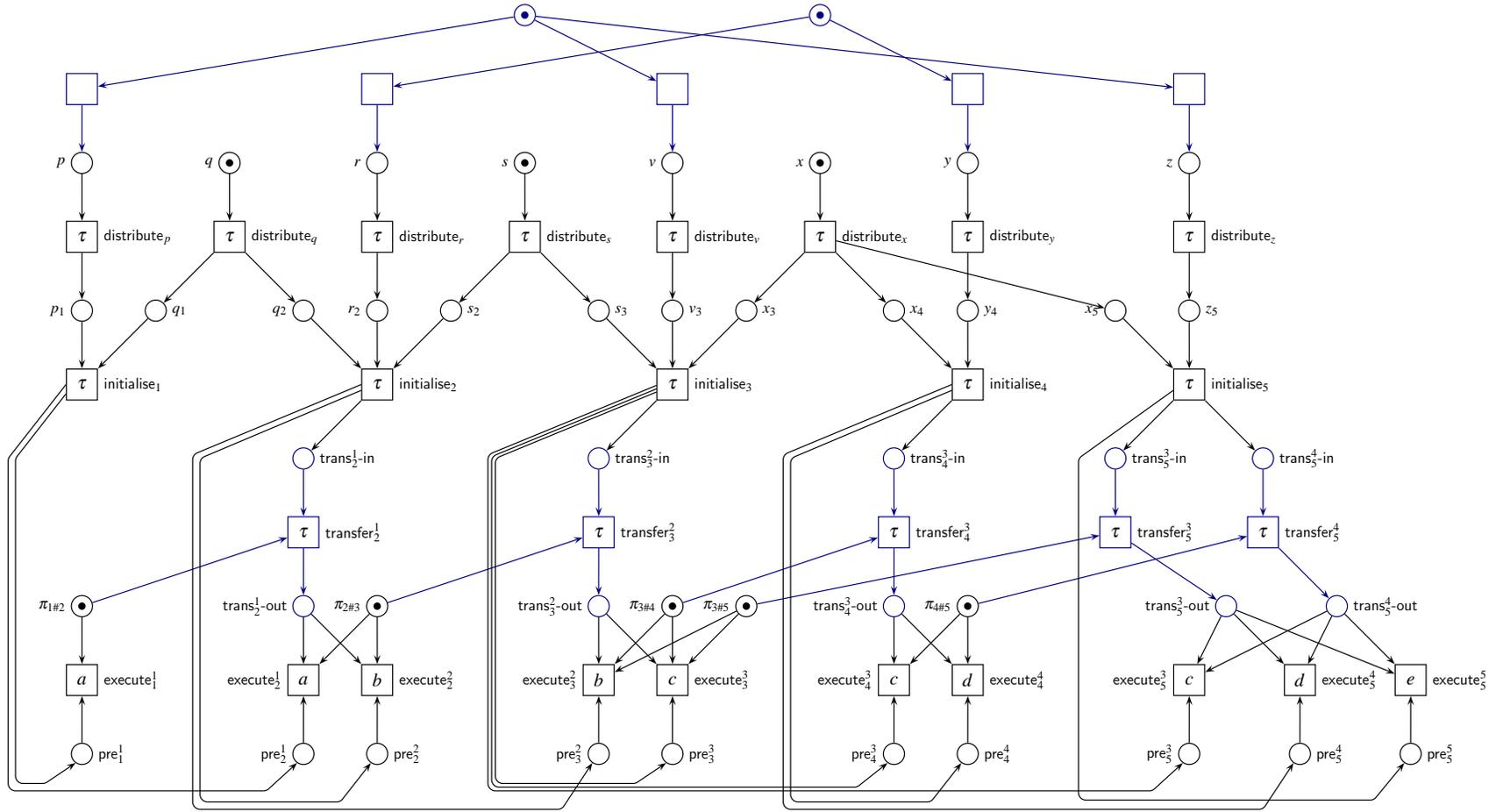
  \caption{The (relevant parts of the) conflict replicating implementation of the net in \reffig{bigexampleoriginal}.}
  \label{fig-bigexample}
\end{landfloat}

The conflict replicating implementation is illustrated by means of the
finitary plain structural conflict net $N$ of \reffig{bigexampleoriginal}.
The places and transitions $a$-$q$-$b$-$s$-$c$-$x$-$d$ in this net
constitute a \emph{Long \structuralM}: for each pair $a$-$b$,
$b$-$c$ and $c$-$d$ of neighbouring transitions, as well as for the
pair $a$-$d$ of extremal transitions, there exists a reachable marking
enabling them both. Moreover, neighbouring transitions in the long
\structuralM\ are in conflict: $a\mathbin \# b$, $b\mathbin \# c$ and $c\mathbin \# d$,
whereas the extremal transitions are concurrent: $a \smile d$.
However, $N$ has no fully reachable pure \structuralM: no \structuralM-shaped triple of transitions
$a$-$b$-$c$, $b$-$c$-$d$ or $b$-$c$-$e$ is ever simultaneously enabled.

In \cite{glabbeek08syncasyncinteractionmfcs} we gave a simpler
implementation, the \emph{transition-controlled choice
  implementation}, that works for all finitary plain 1-safe Petri nets
without such a long \structuralM. Hence $N$ constitutes an example
where that implementation does not apply, yet the conflict replicating
implementation does. In fact, when leaving out the $z$-$e$-branch it
may be the simplest example with these properties. We have added this
branch to illustrate the situation where three transitions are pairwise in conflict.

\reffig{bigexample} presents relevant parts of the conflict
replicating implementation $\impl{N}$ of $N$. The ten places of $N$
return in $\impl{N}$, but the transitions of $N$ are replaced by more
complicated net fragments. In \reffig{bigexample} we have simplified
the rendering of $\impl{N}$ by simply just copying the five topmost
transitions of $N$, instead of displaying the net fragments replacing them.
This simplification is possible since the top half of $N$ is already
distributed. To remind the reader of this, we left those transitions
unlabelled. 

In order to fix a well-ordering $<$ on the remaining transitions, we
named them after the first five positive natural numbers. The ordered
conflicts between those transitions now are $1 \mathord{\leq^\#} 2$,
$2 \mathord{\leq^\#} 3$, $3 \mathord{\leq^\#} 4$, $3 \mathord{\leq^\#} 5$
and $4 \mathord{\leq^\#} 5$. In \reffig{bigexample} we have skipped
all places, transitions and arcs involved in the cleanup of tokens
after firing of a transition. In this example the cleanup is not
necessary, as no place of $N$ is visited twice. Thus, we displayed only the non-reversible part
of the transitions $\ini$ and \plat{$\trans{j}$}---i.e.\ $\ini\cdot\fire$ and
\plat{$\trans{j}\cdot\fire$}---as well as the transitions $\dist$ and \plat{$\exec{j}$}.
Likewise, we omitted the outgoing arcs of \plat{$\exec{j}$}, the
places $\pi_j$, and those places that have arcs only to omitted transitions.
We leave it to the reader to check this net against the definition in
\reffig{conflictrepl}, and to play the token game on this net, to see 
that it correctly implements $N$.
\bigskip

  In \refsec{correctness} we will show,
  for any finitary plain structural conflict net
  without a fully reachable \visible pure \structuralM,
  that $\impl{N} \approx^\Delta_{bSTb} N$,
  and that $\impl{N}$ is essentially distributed.
  Hence $\impl{N}$ is an essentially distributed implementation of $N$.
  By \refpr{essentiallydistributedequal} this implies that
  $N$ is distributable up to $\approx^\Delta_{bSTb}$.
Together with \refthm{trulysyngltfullm} it follows that, for any equivalence between
$\approx_\mathscr{R}$ and $\approx^\Delta_{bSTb}$,
a finitary plain structural conflict net is distributable
iff it has no fully reachable \visible pure~\structuralM.

Given the complexity of our construction, no techniques known to us were adequate for
performing the equivalence proof. We therefore had to develop an entirely new method for
rigorously proving the equivalence of two Petri nets up to $\approx^\Delta_{bSTb}$, one of
which known to be plain. This method is presented in \refsec{method}.

\section{Proving Implementations Correct}\label{sec-method}

This section presents a method for establishing the equivalence of two Petri nets, one of
which known to be \hyperlink{plain}{plain}, up to branching ST-bisimilarity with explicit divergence.
It appears as \refthm{3ST}. First approximations of this method are presented in
Lemmas~\ref{lem-1ST} and~\ref{lem-2ST}. The progression from \reflem{1ST} to \reflem{2ST}
and to \refthm{3ST} makes the method more specific (so less general) and more powerful.
By means of a simplification a similar method can be obtained, also in three steps, for
establishing the equivalence of two Petri nets up to interleaving branching bisimilarity
with explicit divergence. This is elaborated at the end of this section.

\begin{definition}{deterministic}
A labelled transition system $(\st,\tr,\inist)$ is called \emph{deterministic}
if for all reachable states $\mathfrak{M}\in [\inist\rangle$ we have
$\mathfrak{M}\arrownot\goesto[\tau]$ and if $\mathfrak{M}\goesto[a]\mathfrak{M}'$
and $\mathfrak{M}\goesto[a]\mathfrak{M}''$ for some $a\in\act$ then $\mathfrak{M}'=\mathfrak{M}''$.
\end{definition}
Deterministic systems may not have reachable $\tau$-transitions at all;
this way, if $\mathfrak{M}\Goesto[\sigma]\mathfrak{M}'$ and
$\mathfrak{M}\Goesto[\sigma]\mathfrak{M}''$ for some $\sigma\in\act^*$ then
$\mathfrak{M}'=\mathfrak{M}''$.
Note that the labelled transition system associated to a
\hyperlink{plain}{plain} Petri net is deterministic; the same applies
to the ST-LTS, the split LTS or the step LTS associated to such a net.

\begin{lemma}{plain branching bisimilarity}
Let $(\st_1,\tr_1,\inist_1)$ and $(\st_2,\tr_2,\inist_2)$ be two
labelled transition systems, the latter being deterministic.
Suppose there is a relation $\Rel \subseteq \st_1\times\st_2$ such that
\begin{enumerate}[~~~(a)\,]
\item $\inist_1\Rel \inist_2$,
\item if $\mathfrak{M}_1\Rel \mathfrak{M}_2$ and
  $\mathfrak{M}_1\goesto[\tau]\mathfrak{M}'_1$ then $\mathfrak{M}'_1\Rel \mathfrak{M}_2$,
\item if $\mathfrak{M}_1\Rel \mathfrak{M}_2$ and
  $\mathfrak{M}_1\goesto[a]\mathfrak{M}'_1$ for some $a\in\act$ then
  $\exists \mathfrak{M}'_2.~\mathfrak{M}_2\goesto[a]\mathfrak{M}'_2 \wedge \mathfrak{M}'_1\Rel \mathfrak{M}'_2$,
\item if $\mathfrak{M}_1\Rel \mathfrak{M}_2$ and
  $\mathfrak{M}_2\goesto[a]$ for some $a\in\act$ then either
  $\mathfrak{M}_1 \goesto[a]$ or $\mathfrak{M}_1 \goesto[\tau]$
\item and there is no infinite sequence $\mathfrak{M}_1\goesto[\tau] \mathfrak{M}'_1\goesto[\tau] \mathfrak{M}''_1\goesto[\tau] \cdots$
  with $\mathfrak{M}_1\Rel \mathfrak{M}_2$ for some $\mathfrak{M}_2$.
\end{enumerate}
Then $\Rel$ is a branching bisimulation, and the two LTSs are
branching bisimilar with explicit divergence.
\end{lemma}

\begin{proofNobox}
It suffices to show that $\Rel $ satisfies Conditions 1--3 of \refdf{branching LTS};
the condition on explicit divergence follows immediately from (e),
using that a deterministic LTS admits no divergence at all.
\begin{enumerate}
\item By (a).
\item In case $\alpha=\tau$ this follows directly from (b), and otherwise from (c).
  In both cases $\mathfrak{M}^\dagger_2:=\mathfrak{M}_2$ and when $\alpha=\tau$ also $\mathfrak{M}'_2:=\mathfrak{M}_2$.
\item Suppose $\mathfrak{M}_1\Rel \mathfrak{M}_2$ and $\mathfrak{M}_2\goesto[\alpha]\mathfrak{M}'_2$.
  Since $(\st_2,\tr_2,\inist_2)$ is deterministic, $\alpha=a\in\Act$. By (d) we have either 
  $\mathfrak{M}_1 \goesto[a] \mathfrak{M}^1_1$ or $\mathfrak{M}_1 \goesto[\tau] \mathfrak{M}^1_1$ for some $\mathfrak{M}^1_1\in\st_1$.
  In the latter case (b) yields $\mathfrak{M}^1_1\Rel \mathfrak{M}_2$, and using (d) again,
  either $\mathfrak{M}^1_1 \goesto[a] \mathfrak{M}^2_1$ or $\mathfrak{M}^1_1 \goesto[\tau] \mathfrak{M}^2_1$ for some $\mathfrak{M}^2_1\in\st_1$.
  Repeating this argument, if the choice between $a$ and $\tau$ is made $k$ times in
  favour of $\tau$ (with $k\geq 0$), we obtain $\mathfrak{M}^{k}_1\Rel
  \mathfrak{M}_2$ (where $\mathfrak{M}^0_1:=\mathfrak{M}_1$) and
  either $\mathfrak{M}^{k}_1 \goesto[a] \mathfrak{M}^{k+1}_1$ or $\mathfrak{M}^k_1 \goesto[\tau] \mathfrak{M}^{k+1}_1$.
  By (e), at some point the choice must be made in favour of $a$, say at $\mathfrak{M}^k_1$.
  Thus $\mathfrak{M}_1\Goesto \mathfrak{M}^k_1 \goesto[a] \mathfrak{M}^{k+1}_1$, with $\mathfrak{M}^k_1\Rel \mathfrak{M}_2$.
  We take \plat{$\mathfrak{M}^\dagger_1$} and $\mathfrak{M}'_1$ from
  \refdf{branching LTS} to be $\mathfrak{M}^k_1$ and $\mathfrak{M}^{k+1}_1$.
  It remains to show that $\mathfrak{M}^{k+1}_1\Rel \mathfrak{M}'_2$.
  By (c) there is an $\mathfrak{M}''_2\in\st_2$ with $\mathfrak{M}_2\goesto[a]\mathfrak{M}''_2$ and $\mathfrak{M}^{k+1}_1\Rel \mathfrak{M}''_2$.
  Since $(\st_2,\tr_2,\inist_2)$ is deterministic, $\mathfrak{M}'_2=\mathfrak{M}''_2$.
  \hfill $\Box$
\end{enumerate}
\end{proofNobox}

\begin{lemma}{1ST}
Let $N=(S,T,F,M_0,\ell)$ and $N'=(S',T',F',M'_0,\ell')$ be two nets, $N'$ being plain.
Suppose there is a relation
$\Rel  \subseteq (\nat^{S}\times\nat^{T})\times (\nat^{S'}\times\nat^{T'})$
such that
\begin{enumerate}[~~~(a)\,]
\item $(M_0,\emptyset)\Rel (M'_0,\emptyset)$,
\item if $(M_1,U_1)\Rel (M_1',U'_1)$ and
  $(M_1,U_1)\goesto[\tau](M_2,U_2)$ then $(M_2,U_2)\Rel (M_1',U'_1)$,
\item if $(M_1,U_1)\Rel (M_1',U'_1)$ and
  $(M_1,U_1)\goesto[\eta](M_2,U_2)$ for some $\eta\in\Act^\pm$\\\mbox{}\qquad then
  $\exists (M'_2,U'_2).~(M'_1,U'_1)\goesto[\eta](M'_2,U'_2) \wedge (M_2,U_2)\Rel (M_2',U'_2)$,
\item if $(M_1,U_1)\Rel (M_1',U'_1)$ and
  $(M'_1,U'_1)\goesto[\eta]$ with $\eta\in\Act^\pm$ then either
  $\mathord{(M_1,U_1) \goesto[\eta]}$ or $\mathord{(M_1,U_1) \goesto[\tau]}$
\item and there is no infinite sequence $(M,U)\goesto[\tau] (M_1,U_1)\goesto[\tau] (M_2,U_2)
  \goesto[\tau] \cdots$ with $(M,U)\Rel (M',U')$ for some $(M',U')$.
\end{enumerate}
Then $\Rel $ is a branching split bisimulation, and $N \approx^\Delta_{bSTb} N'$.
\end{lemma}

\begin{proof}
That $N$ and $N'$ are branching split bisimilar with explicit
divergence follows directly from  Lemma \ref{lem-plain branching bisimilarity}
by taking $(\st_1,\tr_1,\inist_1)$ and $(\st_2,\tr_2,\inist_2)$ to be
the split LTSs associated to $N$ and $N'$ respectively.
Here we use that the split LTS associated to a plain net is deterministic.
The final conclusion follows by \refpr{split}.
\end{proof}
\reflem{1ST} provides a method for proving $N \approx^\Delta_{bSTb} N'$ that can
be more efficient than directly checking the definition. In particular, the
intermediate states $\mathfrak{M}^\dagger$ and the sequence of
$\tau$-transitions $\Goesto$ from \refdf{branching LTS} do not occur in
\reflem{plain branching bisimilarity}, and hence not in \reflem{1ST}. Moreover,
in Condition (d) one no longer has the match the targets of corresponding transitions.
\reflem{2ST} below, when applicable, provides an even more efficient method:
it is no longer needed to specify the branching split bisimulation $\Rel$,
and the targets have disappeared from the transitions in Condition~\ref{2cST} as well.
Instead, we have acquired Condition~\ref{clause1ST}, but this is structural
property, which is relatively easy to check.

\begin{lemma}{2ST}
Let $N=(S,T,F,M_0,\ell)$ be a net and $N'=(S',T',F',M'_0,\ell')$ be a plain
net with $S'\subseteq S$ and $M'_0=M_0\upharpoonright S'$.
Suppose:
\begin{enumerate}
\item $\forall t\inp T,~\ell(t)\neq\tau.~ \exists t'\inp T',~\ell(t')=\ell(t).~
       \exists G\fin \nat^T,~\ell(G)\equiv\emptyset.~ \marking{t'}=\marking{t+G}$.
      \label{clause1ST}
\item For any $G\fin \Int^T$ with $\ell(G)\equiv\emptyset$, ~\plat{$M'\inp\nat^{S'}$},
      ~\plat{$U'\inp\nat^{T'}$} and ~$U\inp\nat^T$ with ~$\ell'(U')\mathbin=\ell(U)$,
      ~$M'+\precond{U'}\in [M'_0\rangle_{N'}$
      and ~$M:=M'+\precond{U'}+(M_0-M'_0)+\marking{G}-\precond{U}\in\nat^S$ with $M+\precond{U}\in[M_0\rangle_N$,
      it holds that:\label{clause2ST}
\begin{enumerate}
\item there is no infinite sequence $M\goesto[\tau] M_1\goesto[\tau] M_2\goesto[\tau] \cdots$\label{2aST}
\item if $M' \goesto[a]$ with $a\in \Act$ then $M \goesto[a]$ or $M \goesto[\tau]$\label{2bST}
\item and if $M\goesto[a]$ with $a\inp \Act$ then $M'\goesto[a]$.\label{2cST}
\end{enumerate}
\end{enumerate}
Then $N \approx^\Delta_{bSTb} N'$.
\end{lemma}

\begin{proofNobox}
Define $\Rel \subseteq (\nat^{S}\times\nat^{T})\times (\nat^{S'}\times\nat^{T'})$ by
$(M,U) \Rel  (M',U') :\Leftrightarrow
\ell'(U')\mathbin=\ell(U) \wedge M'+\!\precond{U'}\inp [M'_0\rangle_{N'}\linebreak[3]
\wedge \exists G\fin\Int^T.~ \ell(G)\equiv\emptyset \wedge
M+\precond{U}=M'+\precond{U'}+(M_0-M'_0)+\marking{G}\in[M_0\rangle_N$.
It suffices to show that $\Rel $ satisfies Conditions (a)--(e) of \reflem{1ST}.
\begin{enumerate}[~~~(a)\,]
\item Take $G=\emptyset$.
\item Suppose $(M_1,U_1)\Rel (M_1',U'_1)$ and $(M_1,U_1)\goesto[\tau](M_2,U_2)$.
  Then $\ell'(U'_1)\mathbin=\ell(U_1) \wedge M'_1+\!\precond{U'_1}\inp [M'_0\rangle_{N'}\linebreak[2] \wedge
  \exists G\fin\Int^T.~ \ell(G)\mathbin\equiv\emptyset \wedge
  M_1=M'_1+\!\precond{U'_1}+(M_0-M'_0)+\marking{G}-\!\precond{U_1} \wedge M_1+\precond{U}\in[M_0\rangle_N$
  and moreover $M_1\goesto[\tau]M_2 \wedge U_2=U_1$.
  So $M_1 [t\rangle M_2$ for some $t\mathbin\in T$ with $\ell(t)\mathbin=\tau$. Hence
  $M_2=M_1+\marking{t}=M'_1+\!\precond{U'_1}+(M_0\mathord-M'_0)+\marking{G+t}\linebreak[2]-\!\precond{U_1}$.
  Since $(M_1+\precond{U_1}) [t\rangle (M_2+\precond{U_1})$, we have $M_2+\precond{U_1}\in[M_0\rangle_N$.
  Since also $\ell(G+t)\equiv\emptyset$ it follows that $(M_2,U_1)\Rel (M_1',U'_1)$.
\item Suppose $(M_1,U_1)\Rel (M_1',U'_1)$ and $(M_1,U_1)\goesto[\eta](M_2,U_2)$,
  with $\eta\in\Act^\pm$.
  Then $\ell'(U'_1)\mathbin=\ell(U_1)$, ~$M'_1+\!\precond{U'_1}\inp [M'_0\rangle_{N'}$ and
  \begin{equation}\label{G}
  \exists G\fin\Int^T.~ \ell(G)\mathbin\equiv\emptyset \wedge
  M_1+\!\precond{U_1}=M'_1+\!\precond{U'_1}+(M_0-M'_0)+\marking{G}\in[M_0\rangle_N.
  \end{equation}
  First suppose $\eta=a^+$. Then $\exists t\inp T.~ \ell(t)\mathbin=a \wedge M_1[t\rangle
  \wedge M_2=M_1-\precond{t} \wedge U_2=U_1+\{t\}$.
  Using that $M_1\goesto[a]$ with $a\in \Act$, by Condition~\ref{2cST}
  we have $M'_1\goesto[a]$, \ie $M'_1[t'\rangle$ for some
  $t'\in T$ with $\ell'(t')=a$.  Let $M'_2:=M'_1-\precond{t}$ and $U'_2:=U'_1+\{t'\}$.
  Then $(M'_1,U'_1)\goesto[a^+](M'_2,U'_2)$.
  Moreover, $\ell(U_2)=\ell(U'_2)$,
  ~$M'_2+\precond{U'_2} = M'_1+\precond{U'_1}\inp [M'_0\rangle_{N'}$ and
  $M_2+\precond{U_2} = M_1+\precond{U_1}$. In combination with (\ref{G}) this yields
  $$M_2+\!\precond{U_2}=M_1+\precond{U_1} =M'_1+\!\precond{U'_1}+(M_0-M'_0)+\marking{G}
   =M'_2+\!\precond{U'_2}+(M_0-M'_0)+\marking{G},$$
  so $(M_2,U_2)\Rel (M_2',U'_2)$.

  Now suppose $\eta=a^-$. Then $\exists t\inp U_1.\linebreak[3]\ \ell(t)\mathbin=a \wedge
  U_2\mathbin=U_1\mathord-\{t\} \wedge M_2=M_1+\postcond{t}$.  Since
  $\ell'(U'_1)\mathbin=\ell(U_1)$ there is a $t'\inp U'_1$ with
  $\ell(t')\mathbin=a$.  Let $M'_2:=M'_1+\postcond{t'}$ and
  $U'_2:=U'_1-\{t'\}$. Then $(M'_1,U'_1)\goesto[a^-](M'_2,U'_2)$.
  By construction, $\ell(U_2)=\ell(U'_2)$.
  Moreover, $M_2+\precond{U_2} = M_1+\postcond{t}+\precond{U_1}-\precond{t}=(
  M_1+\precond{U_1})+\marking{t}$, and likewise\vspace{-1ex}
  \begin{equation}\label{M'}
  M'_2+\precond{U'_2} = (M'_1+\precond{U'_1})+\marking{t'}
  \end{equation}
  so $(M'_1+\precond{U'_1})[t'\rangle (M'_2+\precond{U'_2})$.
  Since $M'_1+\precond{U'_1}\inp [M'_0\rangle_{N'}$, this yields 
  $M'_2+\precond{U'_2}\inp [M'_0\rangle_{N'}$.
  Moreover, $M_2+\!\precond{U_2} = M_1+\postcond{t}+\!\precond{U_1}-\!\precond{t} =
  M_1+\!\precond{U_1}+\marking{t} \in[M_0\rangle_N$.
  Furthermore, combining (\ref{G}) and (\ref{M'}) gives
  \begin{equation}\label{G2}
  \exists G\fin\Int^T.~ \ell(G)\mathbin\equiv\emptyset \wedge
  M_2+\!\precond{U_2}-\marking{t}=M'_2+\!\precond{U'_2}-\marking{t'}+(M_0-M'_0)+\marking{G}.
  \end{equation}
  By Condition~\ref{clause1ST} of \reflem{2ST}, $\exists t''\inp T',~\ell(t'')=\ell(t).~
       \exists G_t\fin \nat^T,~\ell(G_t)\equiv\emptyset.~ \marking{t}=\marking{t''-G_t}$.
  Since $N'$ is a \hyperlink{plain}{plain} net, it has only one
  transition $t^\dagger$ with $\ell(t^\dagger)\mathbin=a$, so $t''\mathbin=t'$.
  Substitution of $\marking{t'-G_t}$ for $t$\linebreak[3] in (\ref{G2}) yields\vspace{-2ex}
  \[\qquad
  \exists G\fin\Int^T.~ \ell(G)\mathbin\equiv\emptyset \wedge
  M_2+\!\precond{U_2}=M'_2+\!\precond{U'_2}+(M_0-M'_0)+\marking{G-G_t}.
  \]
  Since $\ell(G-G_t)\equiv\emptyset$ we obtain $(M_2,U_2)\Rel (M_2',U'_2)$.
\item Follows directly from Condition~\ref{2bST} and \refdf{split marking}.
\item Follows directly from Condition~\ref{2aST} and \refdf{split marking}.
  \hfill $\Box$
\end{enumerate}
\end{proofNobox}
In \reflem{2ST} a relation is explored between markings $M$ and $M+\marking{H}$
(where $M$ is $M'+\!\precond{U'}+(M_0-M'_0)$ of \reflem{2ST}, $H:=G$, and
$M+\marking{H}$ is $M+\!\precond{U}$ of \reflem{2ST}).
In such a case, we can think of $M$ as an ``original marking'', and of
$M+\marking{H}$ as a modification of this marking by the token replacement
$\marking{H}$. The next lemma provides a method to trace certain places $s$
marked by $M+\marking{H}$ (or transitions $t$ that are enabled under $M+\marking{H}$)
back to places that must have been marked by $M$ before taking into account the
token replacement $\marking{H}$. Such places are called \emph{faithful origins}
of $s$ (or $t$). In tracking the faithful origins of places and transitions, we
assume that the places marked by $M$ are taken from a set $S^+$ and the
transitions in $H$ from a set $T^+$. In \reflem{origin} we furthermore assume that the flow
relation restricted to $S\cup T^+$ is acyclic. We will need this lemma in proving the
correctness of our final method of proving $N \approx^\Delta_{bSTb} N'$.

\begin{definition}{faithful}
Let $N=(S,T,F,M_0,\ell)$ be a Petri net, $T^+\subseteq T$ a set of
transitions and $S^+\subseteq S$ a set of places.
\begin{itemize}
\item
A \emph{path} in $N$ is an alternating sequence $\pi=x_0 x_1 x_2 \cdots x_n \in (S\cup T)^*$ of
places and transitions, such that $F(x_i,x_{i+1})>0$ for $0\mathbin\leq i \mathbin< n$.
The \emph{arc weight} $F(\pi)$ of such a path is the product $\Pi_0^{n-1}F(x_i,x_{i+1})$.
\item
A place $s\in S$ is called \emph{faithful} w.r.t.\ $T^+$ and $S^+$
iff $|\{s\}\cap S^+| + \sum_{t\in T^+}F(t,s)=1$.
\item
A path $x_0 x_1 x_2 \cdots x_n \in (S\cup T)^*$ from $x_0$ to $x_n$ is \emph{faithful}
w.r.t.\ $T^+$ and $S^+$ iff all intermediate nodes $x_i$ for $0\leq i < n$ are either
transitions in $T^+$ or faithful places w.r.t.\ $T^+$ and $S^+$.
\item
For $x\in S\cup T$, the \emph{infinitary multiset} $^*x\in(\nat\cup\{\infty\})^{S^+}$
of \emph{faithful origins} of $x$ is given by\\
$^*x(s)=\sup\{F(\pi)\mid \pi \mbox{ is a faithful path from $s\in S^+$ to $x$}\}$. (So
$^*x(s)=0$ if no such path exists.)
\end{itemize}
\end{definition}
Suppose a marking $M_2$ is reachable from a marking $M_1\in \nat^{S^+}$
by firing transitions from $T^+$ only. Then, if a faithful place $s$
bears a token under $M_2$---i.e.\ $M_2(s)>0$---this token has a unique source:
if $s\in S^+$ it must stem from $M_1$ and otherwise it must be produced by the unique
transition $t\mathbin\in T^+$ with $F(t,s)\mathbin=1$.

In a net without arc weights, $^*x$ is always a set, namely the set of 
places $s$ in $S^+$ from which the flow relation of the net admits a path to $x$ that passes only
through faithful places and transitions from $T^+$ (with the possible exception of $x$ itself).
For nets with arc weights, the underlying set of $^*x$ is
the same, and the multiplicity of $s \in \mbox{}^*x$ is obtained by multiplying all
arc weights on the qualifying path from $s$ to $x$; in case of multiple such paths, we
take the upper bound over all such paths (which could yield the value $\infty$).

\begin{observation}{origin}
Let $(S,T,F,M_0,\ell)$ be a Petri net, $T^+\subseteq T$ a set of
transitions and ${S^+}\subseteq S$ a set of places.
For faithful places $s$ and transitions $t\in T$ we have\vspace{-1.5pt}
$$^*s=\left\{\begin{array}{@{}ll@{}}\{s\} & \mbox{if}~s\in {S^+}\\
                                   ^*t   & \mbox{if}~t\in T^+ \wedge F(t,s)=1
            \end{array}\right.
\qquad\qquad\vspace{-1.5pt}
^*t=\bigcup\{F(s,t)\cdot\mbox{}^*s \mid s\in \precond{t} \wedge s {\rm ~faithful}\}.$$
\end{observation}

\begin{lemma}{origin}\label{lem-faithful origins}
Let $(S,T,F,M_0,\ell)$ be a Petri net, $T^+\subseteq T$ a set of transitions
such that $F\upharpoonright(S\cup T^+)$ is acyclic, and ${S^+}\subseteq S$ a set of places. 
Let $M\in \nat^{S^+}$ and $H\fin \nat^{T^+}$, such that $M+\marking{H}\in\nat^S$. Then
\begin{enumerate}[(a)]
\item
for any faithful place $s$ w.r.t.\ $T^+$ and ${S^+}$ we have
$(M+\marking{H})(s)\cdot\mbox{}^*s \leq M$;
\item
for any $k\in\nat$, and any transition $t$ with $(M+\marking{H}) [k\cdot\{t\}\rangle$,
we have $k\cdot\mbox{}^*t\leq M$.
\end{enumerate}
\end{lemma}

\begin{proof}
We apply induction on $|H|$.\\[.5ex]
(a).
When $(M+\marking{H})(s)=0$ it trivially follows that $(M+\marking{H})(s)\cdot\mbox{}^*s \leq M$.
So suppose $(M+\marking{H})(s)>0$. Then either $s\in {S^+}$ or there is a
unique $t\in T^+$ with $H(t)>0$ and $F(t,s)=1$.
In the first case, using that $s\in\postcond{u}$ for no $u\in T^+$, we
have $(M+\marking{H})(s)\leq M(s)$, so $(M+\marking{H})(s)\cdot\mbox{}^*s
  \leq M(s)\cdot \{s\} \leq M$.

In the latter case, $(M+\marking{H})(s) \leq M(s)+\sum_{u\in T^+}H(u)\cdot F(u,s) =
H(t)$ and $^*s = \mbox{}^*t$.

Let $U:=\{u\in T^+\mid H(u)>0 \wedge u F^+ t\}$ be the set of
transitions occurring in $H$ from which the flow relation of the net
offers a non-empty path to $t$. As $F\upharpoonright(S\cup T^+)$ is acyclic,
$t\notin U$, so $H\!\upharpoonright\! U < H$.
Let $s'$ be any place with $s'\in\precond{u}$ for some transition
$u\in U$. Then, by construction of $U$, it cannot happen
that $s'\in \postcond{v}$ for some transition $v\notin U$ with $H(v)>0$.
Hence $(M+\marking{H\!\upharpoonright\! U})(s')\geq (M+\marking{H})(s')\geq 0$.\linebreak
Moreover, for any other place $s''$ we have
$\precond{(H\!\upharpoonright\! U)}(s'')=0$ and thus
$(M+\marking{H\!\upharpoonright\! U})(s'')\geq M(s'')\geq 0$.
It follows that $M+\marking{H\!\upharpoonright\! U}\in\nat^S$.

For each $s'''\in \precond{t}$ we have
$\postcond{(H-H\!\upharpoonright\! U)}(s''')=0$ and
$\precond{(H-H\!\upharpoonright\! U)}(s''') \ge H(t)\cdot \precond{t} (s''')$ and therefore
$0 \leq (M+\marking{H})(s''') \leq (M+\marking{H\!\upharpoonright\! U})(s''') - H(t)\cdot \precond{t}(s''')$,
and hence $H(t)\cdot\precond{t}\linebreak[3] \leq M+\marking{H\!\upharpoonright\! U}$.
It follows that $(M+\marking{H\!\upharpoonright\! U})[H(t)\cdot \{t\}\rangle$.
Thus, by induction,
$(M+\marking{H})(s)\cdot\mbox{}^*s \leq H(t)\cdot \mbox{}^*t \leq M$.\vspace{1ex}

\noindent
(b).
Let $(M+\marking{H}) [k\cdot\{t\}\rangle$.
For any faithful $s\in \precond{t}$ we have $(M+\marking{H})(s)\geq k\cdot F(s,t)$,
and thus, using (a),\vspace{-1.5pt}
 $$ k\cdot F(s,t)\cdot\mbox{}^*s
        \leq  (M+\marking{H})(s)\cdot\mbox{}^*s
        \leq M\;.\vspace{-1.5pt}$$
Therefore, by \refobs{origin}, $k\cdot\mbox{}^*t =
\bigcup\{k\cdot F(s,t)\cdot\mbox{}^*s \mid s\in \precond{t} \wedge s {\rm ~faithful}\}\leq M$.
\end{proof}
The following theorem is the main result of this section. It presents a method for proving
$N \approx^\Delta_{bSTb} N'$ for $N$ a net and $N'$ a plain net. Its main advantage w.r.t.\ 
directly using the definition, or w.r.t.\ application of \reflem{1ST} or \ref{lem-2ST}, is
the replacement of requirements on the dynamic behaviour of nets by structural requirements.
Such requirements are typically easier to check.
Replacing the requirement ``$M+\precond{U}\in[M_0\rangle_N$'' in Condition~\ref{lastST} by
``$M+\precond{U}\in\nat^S$'' would have yielded an even more structural version of this
theorem; however, that version turned out not to be strong enough for the verification
task performed in \refsec{correctness}.

\begin{theorem}{3ST}
Let $N=(S,T,F,M_0,\ell)$ be a net and $N'=(S',T',F',M'_0,\ell')$ be a plain
net with $S'\subseteq S$ and $M'_0=M_0\upharpoonright S'$.
Suppose there exist sets $T^+ \subseteq T$ and $T^-\subseteq T$
and a class $\NF\subseteq \Int^T$, such that
\begin{enumerate}
\item $F\upharpoonright(S\cup T^+)$ is acyclic.\label{acyclic+ST}
\item $F\upharpoonright(S\cup T^-)$ is acyclic.\label{acyclic-ST}
\item $\forall t\inp T,~\ell(t)\neq\tau.~ \exists t'\inp T',~\ell(t')=\ell(t).~
       \left(\precond{t'} \leq \mbox{}^*t \wedge
       \exists G\fin \nat^T,~\ell(G)\equiv\emptyset.~ \marking{t'}=\marking{t+G}\right)$.\\
  Here $\mbox{}^*t$ is the multiset of faithful origins of $t$ w.r.t.\ $T^+$ and
  $S'\cup\{s\in S \mid M_0(s)>0\}$.
  \label{matchingST}
\item There exists a function $f:T\rightarrow\nat$ with $f(t)>0$ for all $t\inp T$,
  extended to $\Int^T$ as in \refdf{multiset}, such that
  for each $G\fin \Int^T$ with $\ell(G)\equiv\emptyset$ there is an $H\fin \NF$
  with $\ell(H)\equiv\emptyset$, $\marking{H}=\marking{G}$ and $f(H)=f(G)$.
  \label{normalformST}
\item For every $M'\in\nat^{S'}$, $U'\in\nat^{T'}$ and $U\in\nat^{T}$ with
  $\ell(U)=\ell'(U')$ and $M'+\precond{U'}\in[M'_0\rangle_{N'}$, there is an \plat{$H_{M',U}\fin\nat^{T^+}$}
  with $\ell(H_{M',U})\equiv\emptyset$, such that\label{lastST}
  for each $H\mathbin{\fin} \NF$ with $M:=M'+\precond{U'}+(M_0-M'_0)+\marking{H}-\precond{U}\in\nat^S$
  and $M+\precond{U}\in[M_0\rangle_N$:
  \begin{enumerate}
  \item $M_{M',U}:=M'+\precond{U'}+(M_0-M'_0)+\marking{H_{M',U}}-\precond{U}\in\nat^S$,\label{markingST}
  \item if $M'\goesto[a]$ with $a\in\Act$ then \plat{$M_{M',U}\goesto[a]$},\label{matchST}
    \item $H\leq H_{M',U}$.\label{upperboundST}
    \item if $H(u)<0$ then $u\in T^-$,\label{T-ST}
    \item if $H(u)<0$ and $H(t)>0$ then $\precond{u} \cap \precond{t} = \emptyset$,
          \label{disjoint preplacesST}
    \item if $H(u)<0$ and $(M+\!\precond{U})[t\rangle$ with $\ell(t)\neq\tau$
          then $\precond{u} \cap \precond{t} = \emptyset$,
          \label{disjoint preplaces 2ST}
    \item if $(M+\!\precond{U})[\{t\}\mathord+\{u\}\rangle$ and
          and $t',u'\in T'$ with $\ell'(t')=\ell(t)$ and $\ell'(u')=\ell(u)$, then
          $\precond{t'}\cap\precond{u'}=\emptyset$.\label{concurrent}
  \end{enumerate}
\end{enumerate}
Then $N \approx^\Delta_{bSTb} N'$.
\end{theorem}

\begin{proofNobox}
It suffices to show that Condition~\ref{clause2ST} of \reflem{2ST} holds
(for Condition~\ref{clause1ST} of \reflem{2ST} is part of Condition~\ref{matchingST} above).
So let $G\fin \Int^T$ with $\ell(G)\equiv\emptyset$, ~\plat{$M'\inp\nat^{S'}$},
\plat{$U'\inp\nat^{T'}\!$} and $U\inp\nat^T\!$ with $\ell'(U')\mathbin=\ell(U)$,
      ~$M'\mathord+\!\precond{U'}\in [M'_0\rangle_{N'}$,
      ~$M:=M'\mathord+\!\precond{U'}\mathord+(M_0\mathord-M'_0)\mathord+\marking{G}
      \mathord-\!\precond{U}\inp\nat^S$ and $M+\precond{U}\in[M_0\rangle_N$.
\begin{enumerate}[(a)]
\item
Suppose $M\goesto[\tau]M_1\goesto[\tau]M_2\goesto[\tau] \cdots$.
Then there are transitions $t_i\in T$ with $\ell(t_i)=\tau$, for all $i\mathbin\geq 1$, such that
$M[t_1\rangle M_1[t_2\rangle M_2[t_3\rangle \cdots$.
As also $(M+\!\precond{U})[t_1\rangle (M_1+\!\precond{U})[t_2\rangle (M_2+\!\precond{U})[t_3\rangle \cdots$,
it follows that $(M_i+\!\precond{U})\inp[M_0\rangle_N$ for all $i\geq 1$.
Let $G_0:=G$ and for all $i\geq 1$ let $G_{i+1}:=G_i+\{t_i\}$.
Then $\ell(G_i)\equiv\emptyset$ and $M_i=M'+\!\precond{U'}+(M_0-M'_0)+\marking{G_i}-\!\precond{U}$.
Moreover, $f(G_{i+1})=f(G_i)+f(t_i) > f(G_i)$.
For all $i\geq 1$, using Condition~\ref{normalformST},
let $H_i\mathbin{\fin}\NF$ be so that $\marking{H_i}\mathbin=\marking{G_i}$ and $f(H_i)\mathbin=f(G_i)$.
Then $M_i=M'+\!\precond{U'}+(M_0-M'_0)+\marking{H_i}-\!\precond{U}$ and $f(H_0)<f(H_1)<f(H_2)<\cdots$.
However, from Condition~\ref{upperboundST} we get $f(H_i)\leq f(H_{M'})$ for all $i\geq 1$.
The sequence $M\goesto[\tau]M_1\goesto[\tau]M_2\goesto[\tau] \cdots$ therefore must be finite.
\item[(b)]
Now suppose $M' \goesto[a]$ with $a\in\Act$.
By Condition~\ref{normalformST} above there exists an $H\fin\NF$
such that $\ell(H)\equiv\emptyset$ and $\marking{H}=\marking{G}$, and hence
$M=M'+\!\precond{U'}+(M_0-M'_0)+\marking{H}-\!\precond{U}$.
Let $H^-:=\{u\in T\mid H(u)<0\}$.
\begin{itemize}
\item First suppose $H^-\neq\emptyset$.
By Condition~\ref{T-ST}, $H^-\subseteq T^-$.
  By Condition~\ref{acyclic-ST}, $<^-:=(F\upharpoonright(S\cup T^-))^+$
  is a partial order on $S\cup T^-$, and hence on $H^-$.
  Let $u$ be a minimal transition in $H^-$ w.r.t.\ $<^-$.
  By definition, for all $s\in S$,
  \vspace{-3ex}

  {\small
  \begin{equation}\label{token countST}
  M(s)=M'(s)+\!\precond{U'}(s)+(M_0-M'_0)(s)+\!\sum_{t\in T}H(t)\cdot F(t,s)+\!\sum_{t\in T}\!-H(t)\cdot F(s,t)
  +\!\sum_{t\in U}\!-U(t)\cdot F(t,s).\vspace{-2ex}
  \end{equation}}%
  As $M'_0=M_0\upharpoonright S'$, we have $M'_0\leq M_0$.
  Hence the first three summands in this equation are always positive (or $0$).
  Now assume $s\in\precond{u}$. Since $u$ is minimal w.r.t.\ $<^-$,
  there is no $t\in T$ with $H(t)<0$ and $F(t,s)\neq 0$.
  Hence also all summands $H(t)\cdot F(t,s)$ are positive.
  By Condition~\ref{disjoint preplacesST}, there is no
  $t\in T$ with $H(t)>0$ and $F(s,t)\neq 0$,
  so all summands $-H(t)\cdot F(s,t)$ are positive as well.
  By Condition~\ref{disjoint preplaces 2ST}, there is no
  $t\in T$ with $U(t)>0$ and $F(s,t)\neq 0$,
  for this would imply that $\ell(t)\neq\tau$ and $(M+\!\precond{U})[t\rangle$, so
  no summands in (\ref{token countST}) are negative.
  Thus $0\leq -H(u)\cdot F(s,u) \leq M(s)$.
  Since $H(u)\leq -1$, this implies $M(s)\geq F(s,u)$.
  Hence $u$ is enabled in $M$. As $\ell(u)=\tau$, we have $M\goesto[\tau]$.
\item Next suppose $H^-\!=\emptyset$ but $H\neq H_{M',U}$.
  Let $H^\smile:=\{u\in T\mid H_{M',U}(u)-H(u)>0\}$.
  Then $H^\smile\neq\emptyset$ by Condition~\ref{upperboundST}.
  Since \plat{$H_{M',U}\fin\nat^{T^+}\!\!$}, $H^\smile\subseteq T^+$.
  By Condition~\ref{acyclic+ST}, $<^+:=(F\upharpoonright(S\cup T^+))^+$
  is a partial order on $S\cup T^+$, and hence on $H^\smile$.
  Let $u$ be a minimal transition in $H^\smile$ w.r.t.\ $<^+$.
  We have $M=M'+\!\precond{U'}+(M_0-M'_0)+\marking{H_{M',U}+(H-H_{M',U})}-\!\precond{U}=M_{M',U}+\marking{H-H_{M',U}}$.
  Hence, for all $s\in S$,
  \begin{equation}\label{token count 2ST}
  M(s)=M_{M',U}(s)+\sum_{t\in T}(H-H_{M',U})(t)\cdot F(t,s)+\sum_{t\in T}-(H-H_{M',U})(t)\cdot F(s,t)\;.
  \end{equation}
  By Condition~\ref{markingST}, $M_{M',U}\in\nat^S$.
  By Condition~\ref{upperboundST}, $H- H_{M',U}\leq 0$.
  For $s\in \precond{u}$ there is moreover no $t\in H^\smile$ with $s\in\postcond{t}$,
  so no $t\in T$ with $(H-H_{M',U})(t)<0$ and $F(t,s)\neq 0$.
  Hence no summands in (\ref{token count 2ST}) are negative.
  It follows that $0\leq -(H\mathord-M_{M',U})(u)\cdot F(s,t) \leq M(s)$.
  Since $(H\mathord-H_{M',U})(u)\leq -1$, this implies $M(s)\geq F(s,u)$.
  Hence $u$ is enabled in $M$. As $\ell(u)=\tau$, we have $M\goesto[\tau]$.
\item Finally suppose $H= H_{M',U}$. Then $M=M_{M',U}$ and
  $M\goesto[a]$ follows by Condition~\ref{matchST}.
\end{itemize}
\item[(c)]
  Next suppose $M\goesto[a]$ with $a\in\Act$.
  Then there is a $t\in T$ with $\ell(t)=a\neq\tau$ and $M[t\rangle$.
  So $(M+\!\precond{U})[t\rangle$.
  We will first show that $(M'+\!\precond{U'})\goesto[a]$.
  By Condition~\ref{normalformST} there exists an $H_0\fin\NF\subseteq \nat^T$
  such that $\ell(H_0)\equiv\emptyset$ and $\marking{H_0}=\marking{G}$, and hence
  $M+\!\precond{U}=M'+\!\precond{U'}+(M_0-M'_0)+\marking{H_0}\in[M_0\rangle_N$.
  For our first step, it suffices to show that whenever $H \mathbin{\fin} \NF$ with
  $M_H := M'+\!\precond{U'}+(M_0-M'_0)+\marking{H}\inp[M_0\rangle$\linebreak
  and $M_H [t\rangle$, then $(M'+\!\precond{U'})\goesto[a]$.
  We show this by induction on $f(H_{M',U}-H)$, observing that
  $f(H_{M',U}-H)\in\nat$ by Conditions~\ref{upperboundST} (with empty $U$) and~\ref{normalformST}.

  We consider two cases, depending on the emptiness of $H^-:=\{u\in T\mid H(u)<0\}$.

  First assume $H^-\!\mathbin=\emptyset$. Then $H\mathbin{\fin}\nat^T\!$.
  By Condition~\ref{upperboundST} (with empty $U$) we even have $H\mathbin{\fin}\nat^{T^+}\!\!$. 
  Let $\mbox{}^*t$ denote the multiset of faithful origins of $t$ w.r.t.\ $T^+$ and
  $S^+ := S'\cup\{s\in S \mid M_0(s)>0\}$.
  By \reflem{faithful origins}(b), taking $k\mathbin=1$, substituting $M'+\!\precond{U'}+(M_0-M'_0)$ for
  the ``$M$'' of that lemma, and using Condition~\ref{acyclic+ST} of \refthm{3ST},
  $^*t \leq M'+\!\precond{U'}+(M_0-M'_0)$. So by Condition~\ref{matchingST} of \refthm{3ST}
  there is a $t'\in T'$ with $\ell(t')=\ell(t)$ and $\precond{t'} \leq M'+\!\precond{U'}+(M_0-M'_0)$.
  Since $\precond{t'} \in \nat^{S'}$ and $M'_0=M_0\!\upharpoonright\! S'$, this implies
  $\precond{t'} \leq M'+\!\precond{U'}$.
  It follows that $(M'+\!\precond{U'})[t'\rangle_{N'}$ and hence $(M'+\!\precond{U'})\goesto[a]$.

  Now assume $H^- \neq \emptyset$.
  By the same proof as for (b) above, case $H^- \neq \emptyset$,
  there is a transition $u\in H^-$ that is enabled in $M_H$.
  So $M_H[u\rangle M_1$ for some $M_1\in[M_0\rangle_N$, and $M_1=M'+\!\precond{U'}+(M_0-M'_0)+\marking{H+u}$.
  By Condition~\ref{disjoint preplaces 2ST} of \refthm{3ST} (still with empty $U$),
  $\precond{u}\cap\precond{t}=\emptyset$, and thus $M_1[t\rangle$.
  By Condition~\ref{normalformST} of \refthm{3ST} there exists an $H_1\mathbin{\fin}\NF$
  such that $\ell(H_1)\mathbin{\equiv}\emptyset$, $\marking{H_1}\mathbin=\marking{H+u}$, and
  $f(H_1)\mathbin=f(H+u)\mathbin>f(H)$. Thus $M_1=M_{H_1}$ and $f(H_{M',U}-H_1)<f(H_{M',U}-H)$.
  By induction we obtain $(M'+\!\precond{U'})\goesto[a]$.

  By the above reasoning, there is a $t'\in T'$ such that
  $\ell'(t')=\ell(t)$ and $(M'+\!\precond{U'})[t'\rangle$.
  Now take any $u'\in U'$. Then there must be an $u\in U$ with
  $\ell'(u')=\ell(u)$. Since $M[t\rangle$, we have $(M+\!\precond{U})[\{t\}\mathord+\{u\}\rangle$ and
  by Condition~\ref{concurrent} we obtain $\precond{t'}\cap\precond{u'}=\emptyset$.
  It follows that $M'[t'\rangle$, and hence $M'\goesto[a]$.
  \hfill $\Box$
\end{enumerate}
\end{proofNobox}

\subsection*{Digression: Interleaving semantics}

Above, a method is presented for establishing the equivalence of two Petri nets, one of
which known to be \hyperlink{plain}{plain}, up to branching ST-bisimilarity with explicit divergence.
Here, we simplify this result into a method for establishing the equivalence of the two
nets up interleaving branching bisimilarity with explicit divergence.
This result is not applied in the current paper.

\begin{lemma}{1}
Let $N=(S,T,F,M_0,\ell)$ and $N'=(S',T',F',M'_0,\ell')$ be two nets, $N'$ being plain.
Suppose there is a relation $\Rel  \subseteq \nat^{S}\times\nat^{S'}$ such that
\begin{enumerate}[~~~(a)\,]
\item $M_0\Rel M'_0$,
\item if $M_1\Rel M_1'$ and
  $M_1\goesto[\tau]M_2$ then $M_2\Rel M_1'$,
\item if $M_1\Rel M_1'$ and
  $M_1\goesto[a]M_2$ for some $a\in\Act$ then
  $\exists M'_2.~M'_1\goesto[a]M'_2 \wedge M_2\Rel  M'_2$,
\item if $M_1\Rel M_1'$ and
  $M'_1\goesto[a]$ for some $a\in\Act$ then either
  $\mathord{M_1 \goesto[a]}$ or $\mathord{M_1 \goesto[\tau]}$
\item and there is no infinite sequence $M\goesto[\tau] M_1\goesto[\tau] M_2\goesto[\tau] \cdots$
  with $M\Rel M'$ for some $M'$.
\end{enumerate}
Then $N$ and $N'$ are interleaving branching bisimilar with explicit divergence. 
\end{lemma}

\begin{proof}
This follows directly from \reflem{plain branching bisimilarity}
by taking $(\st_1,\tr_1,\inist_1)$ and $(\st_2,\tr_2,\inist_2)$ to be
the interleaving LTSs associated to $N$ and $N'$ respectively.
Here we use that the LTS associated to a plain net is deterministic.
\end{proof}

\begin{lemma}{2}
Let $N=(S,T,F,M_0,\ell)$ be a net and $N'=(S',T',F',M'_0,\ell')$ be a plain
net with $S'\subseteq S$ and $M'_0=M_0\upharpoonright S'$.
Suppose:
\begin{enumerate}
\item $\forall t\inp T,~\ell(t)\neq\tau.~ \exists t'\inp T',~\ell(t')=\ell(t).~
       \exists G\fin \nat^T,~\ell(G)\equiv\emptyset.~ \marking{t'}=\marking{t+G}$.
      \label{clause1}
\item For any $G\mathbin{\fin} \Int^T$ with $\ell(G)\mathbin\equiv\emptyset$, $M'\inp[M'_0\rangle_{N'}$
      and $M:=M'\mathord+(M_0\mathord-M'_0)\mathord+\marking{G}\inp[M_0\rangle_N$, it holds that:\label{clause2}
\begin{enumerate}
\item there is no infinite sequence $M\goesto[\tau] M_1\goesto[\tau] M_2\goesto[\tau] \cdots$,\label{2a}
\item if $M' \goesto[a]$ with $a\in \Act$
  then $\mathord{M \goesto[a]}$ or $\mathord{M \goesto[\tau]}$\label{2b}
\item and if $M \goesto[a]$ with $a\in \Act$ then $M' \goesto[a]$.\label{2c}
\end{enumerate}
\end{enumerate}
Then $N$ and $N'$ are interleaving branching bisimilar with explicit divergence.
\end{lemma}

\begin{proofNobox}
Define $\Rel \subseteq \nat^{S}\times \nat^{S'}$ by
$M \Rel  M' :\Leftrightarrow M'\inp[M'_0\rangle_{N'} \wedge \exists
G\mathbin{\fin}\Int^T.~M=M'\mathord+(M_0\mathord-M'_0)\mathord+\marking{G}\inp[M_0\rangle_N\linebreak[3]
\wedge \ell(G)\equiv\emptyset$.
It suffices to show that $\Rel $ satisfies Conditions (a)--(e) of \reflem{1}.
\begin{enumerate}[~~~(a)\,]
\item Take $G=\emptyset$.
\item Suppose $M_1\Rel M_1'$ and $M_1\goesto[\tau]M_2$.
  Then $\exists G\fin\Int^T\!\!.~M_1=M'_1+(M_0-M'_0)+\marking{G} \wedge \ell(G)\equiv\emptyset$
  and $\exists t\mathbin\in T.~\ell(t)=\tau \wedge M_2=M_1+\marking{t}=M'_1+(M_0-M'_0)+\marking{G+t}$.
  Moreover, $M_1\in[M_0\rangle_N$ and hence $M_2\in[M_0\rangle_N$.
  Furthermore, $M_1'\in[M'_0\rangle_{N'}$ and $\ell(G+t)\equiv\emptyset$, so $M_2\Rel M'_1$.
\item Suppose $M_1\Rel M_1'$ and $M_1\goesto[a]M_2$.
  Then $\exists G\fin\Int^T\!\!.~M_1=M'_1+(M_0-M'_0)+\marking{G} \wedge \ell(G)\equiv\emptyset$
  and $\exists t\mathbin\in T.~\ell(t)=a\neq\tau \wedge M_2=M_1+\marking{t}=M'_1+(M_0-M'_0)+\marking{G+t}$.
  Moreover, $M_1\in[M_0\rangle_N$ and hence $M_2\in[M_0\rangle_N$.
  Furthermore, $M_1'\in[M'_0\rangle_{N'}$.
  By Condition~\ref{clause1} of \reflem{2}, $\exists t'\inp T',~\ell(t')\mathbin=\ell(t).\linebreak[3]\
       \exists G_t\mathbin{\fin} \nat^T,~\ell(G_t)\equiv\emptyset.~ \marking{t}=\marking{t'-G_t}$.
  Substitution of $\marking{t'-G_t}$ for $t$ yields $M_2=M'_1+\marking{t'}+(M_0\mathord-M'_0)+\marking{G-G_t}$.
  By Condition~\ref{2c}, $M'_1\goesto[a]$, so $M'_1\goesto[a] M'_2$ for some $M'_2\in[M'_0\rangle_{N'}$.
  As $t'$ is the only transition in $T'$ with $\ell'(t')=a$, we must have $M'_1[t'\rangle M'_2$.
  So $M'_1+\marking{t'}=M'_2$. Since $\ell(G-G_t)\equiv\emptyset$ it follows that $M_2\Rel M'_2$.
\item Follows directly from Condition~\ref{2b}.
\item Follows directly from Condition~\ref{2a}.
  \hfill $\Box$
\end{enumerate}
\end{proofNobox}
The above is a variant of this \reflem{2ST} that requires Condition~\ref{clause2ST} only
for $U=U'=\emptyset$, and allows to conclude that $N$ and $N'$ are interleaving branching
bisimilar (instead of branching ST-bisimilar) with explicit divergence.
Likewise, the below is a variant of \refthm{3ST} that requires Condition~\ref{lastST} only
for $U=U'=\emptyset$, and misses Condition~\ref{concurrent}.

\begin{theorem}{3}
Let $N=(S,T,F,M_0,\ell)$ be a net and $N'=(S',T',F',M'_0,\ell')$ be a plain
net with $S'\subseteq S$ and $M'_0=M_0\upharpoonright S'$.
Suppose there exist sets $T^+ \subseteq T$ and $T^-\subseteq T$
and a class $\NF\subseteq \Int^T$, such that
\begin{enumerate}
\item[1--4.] Conditions~\ref{acyclic+ST}--\ref{normalformST} from \refthm{3ST} hold, and
\item[5.] For every reachable marking $M'\in [M'_0\rangle_{N'}$ there is an $H_{M'}\fin\nat^{T^+}$
  with $\ell(H_{M'})\equiv\emptyset$, such that for each $H\fin \NF$ with
  $M:=M'+(M_0-M'_0)+\marking{H}\in[M_0\rangle_N$ one has:
  \begin{enumerate}
  \item $M_{M'}:=M'+(M_0-M'_0)+\marking{H_{M'}}\in\nat^S$,\label{marking}
  \item if $M'\goesto[a]$ with $a\in\Act$ then $M_{M'}\goesto[a]$,\label{match}
    \item $H\leq H_{M'}$,\label{upperbound}
    \item if $H(u)<0$ then $u\in T^-$,\label{T-}
    \item if $H(u)<0$ and $H(t)>0$ then $\precond{u} \cap \precond{t} = \emptyset$,
          \label{disjoint preplaces}
    \item if $H(u)<0$ and $M[t\rangle$ with $\ell(t)\neq\tau$ then $\precond{u} \cap \precond{t} = \emptyset$.
          \label{disjoint preplaces 2}
  \end{enumerate}
\end{enumerate}
Then $N$ and $N'$ are interleaving branching bisimilar with explicit divergence.
\end{theorem}

\begin{proof}
A straightforward simplification of the proof of \refthm{3ST}.
\end{proof}

\section{The Correctness Proof}\label{sec-correctness}

We now apply the preceding theory to prove the correctness of the conflict replicating
implementation.

\begin{theorem}{correctness}
Let $N$ be a finitary plain structural conflict net without a
fully reachable pure $\structuralM$.\\ Then $\impl{N} \approx^\Delta_{bSTb} N$.
\end{theorem}

\begin{proofNobox}
In this proof the given finitary plain structural conflict net without a fully reachable
pure $\structuralM$ will be $N'=(S',T',F',M'_0,\ell')$, and its conflict replicated
implementation $\impl{N'}$ is called $N=(S,T,F,M_0,\ell)$. This convention matches
the one of \refsec{method}, but is the reverse of the one used in \refsec{distributable};
it pays off in terms of a significant reduction in the number of primes in this paper.

For future reference, \reftab{conflictrepl} provides a place-oriented
representation of the conflict replicating implementation of a given net
$N'=(S',T',F',M'_0,\ell')$, with the macros for reversible transitions expanded.
\hypertarget{far}{Here $\mbox{\hyperlink{Tback}{$T^\leftarrow$}}=\{\ini \mid j\inp T'\} \cup
\{\trans{j} \mid h <^\# j\inp T'\}$, whereas \plat{$(\trans{j})^{\,\it far}=\{\transout{j}\}$}
and \plat{$(\ini)^{\,\it far}=\{\Pre^j_k \mid k\geq^\# j\} \cup \{\transin{j}\mid h<^\# j\}$}.}

We will obtain \refthm{correctness} as an application of \refthm{3ST}.
Following the construction of $N$ described in \refsec{implementation},
we indeed have $S'\subseteq S$ and $M'_0=M_0\upharpoonright S'$.
Let $T^+\subseteq T$ be the set of transitions
\begin{equation}\label{positive transitions}
  \dist      \qquad
  \ini\cdot\fire       \qquad
  \trans{j}\cdot\fire  \qquad
\end{equation}
for any applicable values of $p\inp S'$ and $h,j\inp T'\!$.
Furthermore, $T^-:=(T\setminus (T^+ \cup \{\exec{j}\mid i\leqc j \in T'\}))$.
We start with checking Conditions \ref{acyclic+ST}, \ref{acyclic-ST} and \ref{matchingST} of \refthm{3ST}.
\begin{enumerate}
\item[\ref{acyclic+ST}.] Let $<^+$ be the partial order on $T^+$ given by the order of
  listing in (\ref{positive transitions})---so $\ini[i]\cdot\fire <^+ \trans{j}\!\cdot\fire$, for any
  $i\in T'$ and $h<^\#j\in T'$, but
  the transitions $\trans{j}\cdot\fire$ and $\trans[k]{l}\cdot\fire$ for $(i,j)\neq(k,l)$ are unordered.
  By examining  \reftab{conflictrepl}
  we see that for any place with a pretransition $t$ in $T^+$, all its posttransitions $u$ in $T^+$
  appear higher in the $<^+$-ordering: $t<^+ u$. From this it follows that
  $F\upharpoonright(S\cup T^+)$ is acyclic.
\item[\ref{acyclic-ST}.] Let $<^-\!$ be the partial order on $T^-\!$ given by the row-wise order of
  the following enumeration of $T^-\!$:
\[\begin{array}{l@{\quad}l@{\quad}l@{\quad}ll}
                       t\cdot\undo &

  \trans{j}\cdot\und  & \trans{j}\cdot\undone &
  \ini\cdot\und        &  \ini\cdot\undone  \\
  \fetch               &  \fetched{j}         &
  t\cdot\reset[i] \qquad & t\cdot\elide[i] &  \comp{j}
\end{array}\]
  for any $t\in\{\ini,~ \trans{j}\}$ and any applicable
  values of $f\inp S$, $p\inp S'$, and $h,i,j,c\inp T'\!$.
  By examining \reftab{conflictrepl}
  we see that for any place with a pretransition $t$ in $T^-$, all its posttransitions $u$ in $T^-$
  appear higher in the $<^-$-ordering: $t<^- u$. From this it follows that
  $F\upharpoonright(S\cup T^-)$ is acyclic.
\end{enumerate}
\[
\begin{array}{@{}llll@{}}
~\\[-7ex]
\textbf{Place} &
\textrm{Pretransitions}\hfill\scriptstyle\rm{arc~weights} & \textrm{Posttransitions}\hfill\scriptstyle\rm{arc~weights} & \textrm{for all} \\
\hline
p                & \comp{j}\weight{i,p} & \dist ~~~\mbox{\scriptsize (if $\postcond{p}\mathbin{\neq}\emptyset$)}
                     & p\inp S',~ i\in \precond{p} \\
p_c            & \left\{
                          \begin{array}{@{}l@{}}\dist\\\ini[c]\cdot\undone~~\weight{p,c}\!\!\end{array}\right. &
                          \begin{array}{@{}l@{}}\ini[c]\cdot\fire~~~~~~~~~~~\weight{p,c}\\\fetch\weight{p,i}\end{array} &
                            \begin{array}{@{}l@{}}p\inp S',~c\in\postcond{p}\\j\geq^\# i \in \postcond{p}\end{array} \\
\pi_c ~\hfill\mbox{(marked)} & \ini[c]\cdot\reset & \ini[c]\cdot\fire & i\confeq c \in T'\\
\Pre^i_j           & \left\{
                          \begin{array}{@{}l@{}}\ini[i]\cdot\fire\\\exec{j}\end{array}\right. &
                          \begin{array}{@{}l@{}}\exec{j}\\\ini[i]\cdot\und[\Pre^i_j]\end{array} &
                          \begin{array}{@{}l@{}}j\geq^\# i\in T'\end{array} \\
\transin{j}          & \left\{
                          \begin{array}{@{}l@{}}\ini\cdot\fire\\\trans{j}\cdot\undone\end{array}\right. &
                          \begin{array}{@{}l@{}}\trans{j}\cdot\fire\\\ini\cdot\und[\transin{j}]\end{array} &
                                    h <^\# j\in T' \\
\transout{j}         & \left\{
                          \begin{array}{@{}l@{}}\trans{j}\cdot\fire\\\exec{j}\end{array}\right. &
                          \begin{array}{@{}l@{}}\exec{j}\\\trans{j}\cdot\und[\transout{j}]\!\!\!\end{array} &
                                          h<^\# j\in T',~ i\leqc j \\
\pi_{j\#l} ~\hfill\mbox{(marked)} &  \left\{
                          \begin{array}{@{}l@{}}\fetched{j}\\\trans[j]{l}\cdot\reset[c]\end{array}\right. &
                          \begin{array}{@{}l@{}}\exec{j}\\\trans[j]{l}\cdot\fire\end{array} &
                          \begin{array}{@{}l@{}}i\leqc j <^\# l \in T',~ c\confeq l \end{array} \\
\fetchin           & \exec{j} & \fetch & j\geq^\# i\inp T',~p\inp\precond{i},~c\inp\postcond{p} \\
\fetchout          & \fetch & \fetched{j} &        j\geq^\# i\inp T',~p\inp\precond{i},~c\inp\postcond{p} \\
\hline
\rule{0pt}{12pt}
\undo(t)      & \exec[i]{j}\cdot\fire & t\cdot\undo,\quad t\cdot\elide & j\geq^\# i\in T',~ t\in \UIij \\
\reset(t)      & \fetched[i]{j}      & t\cdot\reset,\quad t\cdot\elide & j\geq^\# i\in T',~ t\in \UIij \\
\ack(t)      & t\cdot\reset,\quad t\cdot\elide & \comp[i]{j} & i\in T',~ t\in \UIij \\
\Fired(t)       & t\cdot\fire & t\cdot\undo & t\in T^\leftarrow,~ \UIij\ni t \\
\keep(t)    & t\cdot\undo & t\cdot\reset & t\in T^\leftarrow,~ \UIij\ni t \\
\take(f,t)  & t\cdot\undo & t\cdot\und  & t\in T^\leftarrow,~ \UIij\ni t,~ f\inp t^{\,\it far} \\
\took(f,t)  & t\cdot\und & t\cdot\undone  & t\in T^\leftarrow,~ f\in t^{\,\it far} \\
\rho(t)       & t\cdot\undone & t\cdot\reset & t\in T^\leftarrow,~ \UIij\ni t \\
\end{array}
\]
\begin{center}
\vspace{1ex}
\refstepcounter{table}
Table \thetable: The conflict replicating implementation.
\label{tab-conflictrepl}
\vspace{3ex}
\end{center}
\begin{enumerate}
\item[\ref{matchingST}.]
  The only transitions $t\in T$ with $\ell(t)\neq\tau$ are $\exec{j}$, with
  $i\leqc j\in T'$. So take $i\leqc j\in T'$. Then the only transition $t'\inp T'$
  with $\ell'(t')\mathbin=\ell(\exec{j})$ is $i$. Now two statements regarding $i$ and $\exec{j}$
  need to be proven. For the first, note that, for any $p\in\precond{i}$, the places
  $p$, $p_i$ and $\Pre^i_j$ are faithful w.r.t.\ $T^+$ and $S'\cup\{s\in S \mid M_0(s)>0\}$.
  Hence $~ p ~~ \dist ~~ p_i ~~ \ini[i]\cdot\fire ~~ \Pre^i_j ~~ \exec{j} ~$
  is a faithful path from $p$ to $\exec{j}$. The arc weight of this path is
  $F'(p,i)$. Thus $\precond{i} \leq \mbox{}^*\exec{j}$.

  The second statement holds because, for all $i\leqc j\in T'$,
  \begin{equation}\label{mimic}
  \marking{i} = \llbracket
     \exec{j} + \!\!\sum_{p\in\precond{i}}\big(F'(p,i)\cdot\dist +
     \!\!\sum_{c\in\postcond{p}}\fetch\big) + \fetched{j} + \comp{j} +
     \sum_{t\in\UIij} t\cdot\elide
     \rrbracket.
  \end{equation}
  To check that these equations hold, note that
  $$\begin{array}{@{}l@{~}c@{~}l@{}}
    \marking{\dist}&=&-\{p\} + \{p_c \mid c\in \postcond{p}\},\\
    \marking{\exec{j}}&=&-\{\pi_{j\#l}\mid l\geq^\# j\}+\{\fetchin \mid p\inp\precond{i},~
    c\inp\postcond{p}\}+\{\undo(t) \mid t\in\UIij\},\\
    \marking{\fetch}&=&-\{\fetchin\} - F'(p,i)\cdot\{p_c\}+\{\fetchout\},\\
    \marking{\fetched{j}}&=&-\{\fetchout\mid p\inp\precond{i},~ c\inp\postcond{p}\}
    +\{\pi_{j\#l}\mid l\geq^\# j\}+\{\reset(t)\mid t\in \UIij\},\\
    \marking{t\cdot\elide}&=&-\{\undo(t),~\reset(t)\mid t\in \UIij\}+\{\ack(t) \mid t\in\UIij\},\\
    \marking{\comp{j}} &=&-\{\ack(t) \mid t\in\UIij\}+
                           \plat{$\displaystyle\sum_{r\in\postcond{i}}F'(i,r)\cdot\{r\}$}.\\[1ex]
  \end{array}$$
\end{enumerate}
Before we define the class $\NF\subseteq \Int^T$ of signed multisets of transitions in
normal form, and verify conditions \ref{normalformST} and \ref{lastST}, we derive
some properties of the conflict replicating implementation $N=\impl{N'}$.

\begin{claim}{G-properties}
  For any $M'\in\Int^{S'}$ and $G\fin \Int^T$ such that
  $M:=M'+(M_0-M'_0)+\marking{G}\in\nat^S$ we have
  \begin{eqnarray}
  G(t\cdot\elide)+G(t\cdot\undo) &\!\!\!\!\leq\!\!\!\!& \sum_{j\geq^\# i}G(\exec[i]{j})\label{undo}\\
  \hspace{-2em}
  G(\comp[i]{j}) \leq G(t\cdot\elide)+G(t\cdot\reset) &\!\!\!\!\leq\!\!\!\!& \sum_{j\geq^\# i}G(\fetched[i]{j})\label{reset}\\
  G(t\cdot\reset) &\!\!\!\!\leq\!\!\!\!& G(t\cdot\undo)\label{interfacecount}
  \end{eqnarray}
  for each $i\in T'$ and $t\inp \UIij$.
  Moreover, for each \hyperlink{far}{$t\in T^\leftarrow$ and $f\in t^{\,\it far}$},
  \begin{equation}\label{took}\label{undocount}
  \sum_{\{\omega\mid t\in \UI_\omega\}}\!\!\!\!\!\!\!\! G(t\cdot\reset[\omega])
  \leq G(t\cdot\undone) \leq G(t\cdot\und)
  \leq \!\!\!\!\!\!\!\!\sum_{\{\omega\mid t\in \UI_\omega\}}\!\!\!\!\!\!\!\! G(t\cdot\undo[\omega]) \leq G(t\cdot\fire)
  \end{equation}
  and for each appropriate $c,h,i,j,l\in T'$ and $p\in S'$:
  \begin{eqnarray}
  G(\fetched[i]{j}) \leq G(\fetch) &\!\!\!\!\leq\!\!\!\!& G(\exec{j})\label{fetch}
  \\[2pt]\label{pi-j}
  G(\ini\cdot\fire) &\!\!\!\!\leq\!\!\!\!& 1+\sum_{\omega }G(\ini\cdot\reset[\omega])
  \\[-6pt]\label{transin}\hspace{-2em}
  G(\trans{j}\cdot\fire) - G(\trans{j}\cdot\undone) &\!\!\!\!\leq\!\!\!\!&
  G(\ini\cdot\fire) - G(\ini\cdot\und[\transin{j}])
  \\[2pt]\label{pi}
  G(\trans[j]{l}\!\cdot\fire) + \sum_{i\leqc j}G(\exec{j}) &\!\!\!\!\leq\!\!\!\!&
  1+\sum_{\omega}G(\trans[j]{l}\!\cdot\reset[\omega])+ \sum_{i\leqc j}G(\fetched{j})
  \\[-5pt]\label{pre}
  \mbox{if ~$M[\exec{j}\rangle$~ then}\quad
  1 &\!\!\!\!\leq\!\!\!\!& G(\ini[i]\cdot\fire) - G(\ini[i]\cdot\und[\Pre^i_j])
  \\[2pt]\label{transout}
  \mbox{if ~$\exists i.~M[\exec{j}\rangle$~ then}\quad
  1 &\!\!\!\!\leq\!\!\!\!& G(\trans{j}\cdot\fire) - G(\trans{j}\cdot\und[\transout{j}])
  \end{eqnarray}\vspace{-15pt}%
    \begin{equation}\label{p_j}
      F'(p,c)\mathord\cdot \big(G(\ini[c]\!\cdot\fire) \mathord- G(\ini[c]\!\cdot\undone)\big)
      +  \hspace{-.5em}\sum_{j\geq^\# i\in \postcond{p}}\hspace{-.5em}
      F'(p,i) \cdot G(\fetch) \leq G(\dist)
     \vspace{-15pt}
    \end{equation}
  \begin{eqnarray}\label{p}\hspace{-1.67em}
  G(\dist) &\!\!\!\!\leq\!\!\!\!&
  M'(p)+ \hspace{-1em}\sum_{\{i\in T'\mid p\in\postcond{i}\}}\hspace{-1em} G(\comp{j}).
  \end{eqnarray}
\end{claim}

\begin{proofclaim}
  For any $i\in T'$ and $t\in\UIij$, we have
  $$M(\undo(t))=\big(\sum_{j\geq^\# i}G(\exec[i]{j})\big)-G(t\cdot\elide)-G(t\cdot\undo)\geq 0,$$ given that
  $M'(\undo(t))=(M_0-M'_0)(\undo(t))=\emptyset$.
  In this way, the place $\undo(t)$ gives rise to the inequation (\ref{undo}) about $G$.
  Likewise, the places $\ack(t)$, $\reset(t)$ and $\keep(t)$,
  respectively, contribute (\ref{reset}) and (\ref{interfacecount}), whereas
  $\rho(t)$, $\took(t)$, $\take(t)$ and $\Fired(t)$ yield (\ref{took}).
  The remaining inequations arise from $\fetchout$, $\fetchin$, $\pi_j$,
  $\transin{j}$, $\pi_{j\#l}$, $\Pre^i_j$, $\transout{j}$, $p_c$ and $p$, respectively.
\end{proofclaim}
  (\ref{pi}) can be rewritten as $T^j_l+\sum_{i\leqc j} E^i_j \leq 1$, where
  $T^j_l:=G(\trans[j]{l}\cdot\fire) - \sum_{\omega}G(\trans[j]{l}\cdot\reset[\omega])$ and
  $E^i_j:=G(\exec[i]{j}) - G(\fetched[i]{j})$.
  By (\ref{undocount})
  $\sum_\omega G(\trans[j]{l}\cdot\reset)\leq G(\trans[j]{l}\cdot\fire)$, so $T^j_l\geq 0$,
  and likewise, by (\ref{fetch}), $E^i_j\geq 0$ for all $i\leqc j$.
  Hence, for all $i\leqc j <^\# l\in T'$,\vspace{-1ex}
  \begin{equation}\label{pi-execute}
  0\leq T^j_l\leq 1 \qquad 0\leq E^i_j\leq 1\qquad T^j_l+\sum_{i\leqc j} E^i_j \leq 1.
  \vspace{-2ex}
  \end{equation}

\newcommand{\follow}{\textit{next}}
\noindent
In our next claim we study triples $(M, M', G)$ with
\begin{enumerate}[(A)]
\item $M\in[M_0\rangle_N$, $M'\in[M'_0\rangle_{N'}$ and $G \fin \Int^T$,\label{r0}
\item $M=M'+(M_0-M'_0)+\marking{G}$,\label{r1}
\item $G(\comp{j}) = 0$ for all $i\in T'$,\label{r2}
\item $G(\dist) \leq M'(p)$ for all $p\in S'$,\label{r3}
\item $G(\fetched[k]{l})\geq 0$ for all $k\leqc l\in T'$,\label{rFp}
\item \plat{$\displaystyle G(\dist) \geq F'(p,i)\cdot G(\exec{j})$}
  for all $i\leqc j\in T'$ and $p\in\precond{i}$,\label{r4}
\item \plat{$0\leq G(\exec{j})\leq 1$} for all $i\leqc j\in T'$,\label{rExec}
\item \plat{$\displaystyle G(\dist) \geq F'(p,j)\cdot G(\exec{j})$}
  for all $i\leqc j\in T'$ and $p\in\precond{j}$,\label{r5}
\item (in the notation of (\ref{pi-execute}))
  if $E^i_j=1$ with $i\leqc j\in T'$ then \plat{$T^h_j=1$} for all $h <^\# j$,\label{rtrans}
\item there are no $j\geq^\# i \confeq k \leq^\# l \in T'$ with $(i,j)\neq (k,\ell)$,
  \plat{$G(\exec[i]{j})>0$} and \plat{$G(\exec[k]{l})>0$},\label{rExec2}
\item there are no $i\leq^\# j \confeq k \leq^\# l \in T'$ with $(i,j)\neq (k,\ell)$,
  \plat{$G(\exec[i]{j})>0$} and \plat{$G(\exec[k]{l})>0$}.\label{rExec3}
\label{rLast}
\end{enumerate}
Given such a triple $(M_1,M'_1,G_1)$ and a transition $t\in T$,
we define $\follow(M_1, M'_1, G_1, t) =: (M, M', G)$ as follows:
Let $G_2:=G_1+\{t\}$.
Take $M:=M_1+\marking{t} = M'_1+(M_0-M'_0)+\marking{G_2}$.
In case $t$ is not of the form $\comp{j}$ we take $M':=M'_1\in[M'_0\rangle_{N'}$ and $G:=G_2\fin\Int^T$.
In case $t\mathbin=\comp{j}$ for some $i\in T'$ we have
\plat{$1=G_2(\comp{j}) \leq \sum_{j\geq^\# i} G_2(\exec{j})
  =\sum_{j\geq^\# i} G_1(\exec{j})$} by (\ref{r2}), (\ref{reset}) and (\ref{fetch}),
so by (\ref{rExec}) and (\ref{rExec2}) there is a unique $j\geq^\# i$ with
\plat{$G_1(\exec{j})=1$}. We take $M':=M'_1+\marking{i}$ and
\plat{$G:=G_2-G^i_{\!\!j}$}, where $G^i_{\!\!j}$ is the right-hand side of (\ref{mimic}).

\begin{claim}{reachable}
\begin{enumerate}[(1)]
\item If $M_1[t\rangle$ and $(M_1, M'_1, G_1)$ satisfies (\ref{r0})-(\ref{rExec3}),
  then so does $\follow(M_1, M'_1, G_1, t)$.\label{follow1}
\item For any $M\in[M_0\rangle_{N}$ there exist $M'$ and $G$ such that
  (\ref{r0})-(\ref{rExec3}) hold.\label{follow2}
\end{enumerate}
\end{claim}

\begin{proofclaimNobox}
(\ref{follow2}) follows from (\ref{follow1}) via induction on the reachability of $M$.
In case $M=M_0$ we take $M':=M'_0$ and $G:=\emptyset$. Clearly,
(\ref{r0})--(\ref{rLast}) are satisfied.

Hence we now show (\ref{follow1}). Let $(M, M', G) := \follow(M_1, M'_1, G_1, t)$.
We check that $(M,M',G)$ satisfies the requirements (\ref{r0})--(\ref{rLast}).
\begin{enumerate}[(A)]
\item By construction, $M\in[M_0\rangle_N$ and $G \fin \Int^T$.
  If $t$ is not of the form $\comp{j}$ we have \mbox{$M'\mathbin=M_1\inp[M'_0\rangle_{N'}$}.
  Otherwise, by (\ref{r3}) and (\ref{r4}) we have $M'_1(p)\geq G_1(\dist)\geq F'(p,i)$
  for all $p\in \precond{i}$, and hence $M'_1[i\rangle$. 
  This in turn implies that $M'=M'_1+\marking{i}\in[M'_0\rangle_{N'}$.
\item  In case $t$ is not of the form $\comp{j}$ we have
  $$M=M_1+\marking{t}=M'_1+(M_0-M'_0)+\marking{G_1+t}= M'+(M_0-M'_0)+\marking{G}.$$
  In case $t=\comp{j}$ we have $M=M'_1+(M_0-M'_0)+\marking{G_2}=
  M'+(M_0-M'_0)+\marking{G}$, using that $\marking{i}=\marking{G^i_{\!\!j}}$.
\item In case $t=\comp{j}$ we have $G(\comp{j}) = G_1(\comp{j})+1-G^i_{\!\!j}(\comp{j}) =
  0+1-1=0$.\\ Otherwise $G(\comp{j}) = G_1(\comp{j})+0=0+0=0$.
\item This follows immediately from (\ref{r2}) and (\ref{p}).
\item  
  The only time that this invariant is in danger is when $t=\comp{j}$.
  Then $G=G_1+\{\comp{j}\}-G^i_{\!\!j}$ for a certain $j\geq^\# i$ with $G_1(\exec{j})=1$.
  By (\ref{rExec2})\footnote{We use (\ref{rExec2}) and (\ref{rFp}) for $G_1$ only, making
  use of the induction hypothesis.}\let\fnote\thefootnote\
  $G_1(\exec{l})\leq 0$ for all $l\geq^\# i$ with $l\neq j$.
  Hence by (\ref{fetch}) $G_1(\fetched{l})\leq 0$ for all such $l$.
  By (\ref{r2}) $G_2(\comp{j}) = G_1(\comp{j})+1 = 1$, so by (\ref{reset})
  $\sum_{l\geq^\# i} G_1(\fetched{l}) = \sum_{l\geq^\# i} G_2(\fetched{l})>0$;
  hence it must be that $G_1(\fetched{j})>0$. By (\ref{rFp})$^\fnote$ $G_1(\fetched[k]{l})\geq 0$
  for all $k\leqc l\in T'$.
  Given that $G^i_{\!\!j}(\fetched{j})=1$ and $G^i_{\!\!j}(\fetched[k]{l})=0$ for all $(k,l)\neq(i,j)$,
  we obtain $G(\fetched[k]{l})\geq 0$ for all $k\leqc l\in T'$.
\item Take $i\mathbin{\leqc} j \inp T'$ and $p\inp\precond{i}$.
  There are two occasions where the invariant is in danger: when $t=\exec{j}$
  and when $t=\comp[k]{l}$ with $k\in T'$.  First let $t=\exec{j}$.
  Then $M_1[\exec{j}\rangle$.  Thus,\vspace{-1ex}
  \hypertarget{proofr4}{$$\begin{array}[b]{@{}r@{~\geq~}ll@{}}
  \multicolumn{2}{@{}l}{G(\dist)}\\
  \mbox{}& \displaystyle
    F'(p,i)\cdot\big(G(\ini[i]\cdot\fire) - G(\ini[i]\cdot\undone)\big)
    + \hspace{-.5em}\sum_{h\geq^\# g\in\postcond{p}}\hspace{-.5em}F'(p,g)\cdot G(\textsf{fetch}_{g,h}^{p,i})
  & \mbox{(by (\ref{p_j}))} \\
  & \displaystyle
    F'(p,i)\cdot\big(G(\ini[i]\cdot\fire) - G(\ini[i]\cdot\undone)\big)
    + \hspace{-.5em}\sum_{h\geq^\# g\in\postcond{p}}\hspace{-.5em}F'(p,g)\cdot G(\fetched[g]{h})
  & \mbox{(by (\ref{fetch}))} \\
  & F'(p,i)\cdot\big(G(\ini[i]\cdot\fire) - G(\ini[i]\cdot\undone)\big) + F'(p,i)\cdot G(\fetched{j})
  & \mbox{(by (\ref{rFp}))} \\
  & F'(p,i)\cdot\left(\big(G(\ini[i]\cdot\fire) - G(\ini[i]\cdot\und[\Pre^i_j])\big) + G(\fetched{j})\right)
  & \mbox{(by (\ref{took}))} \\
  & F'(p,i)\cdot\big(1 + G(\fetched{j})\big)
  & \mbox{(by (\ref{pre}))} \\
  & F'(p,i)\cdot G(\exec{j})
  & \mbox{(by (\ref{pi-execute}))}.
  \end{array}$$}
  Now let $t=\comp[k]{l}$ with $k\in T'$. 
  By (\ref{undocount}) $G(\ini[i]\cdot\fire) - G(\ini[i]\cdot\undone)\geq 0$.
  So by (\ref{p_j}), (\ref{rFp}), and (\ref{fetch}) $G(\dist)\geq 0$.
  For this reason we may assume, w.l.o.g., that $G(\exec{j}) \geq 1$.

  We have $G=G_1+\{\comp[k]{l}\}-G^k_l$ for certain $l\geq^\# k$ with $G_1(\exec[k]{l})\mathbin=1$.
  Since \plat{$G^i_{\!\!j}(\exec{j})\mathbin\geq 0$}, we also have $G_1(\exec{j}) \geq 1$.
  By (\ref{rExec2}) this implies that $\neg(i\confeq k)$ or $(i,j)=(k,l)$.
  In the latter case we have \plat{$G(\exec{j})=\mbox{}$}
  \plat{$G_1(\exec{j})-G^i_{\!\!j}(\exec{j})=1-1=0$},
  contradicting our assumption.
  In the former case $p\notin\precond{k}$, so $G^k_l(\dist)=0$ and hence
  $G(\dist)=G_1(\dist)\geq F'(p,i)\cdot G_1(\exec{j})=F'(p,i)\cdot G(\exec{j})$.
\item That $G(\exec{j})\geq 0$ follows from (\ref{rFp}) and (\ref{fetch}).
  If $G(\exec{j})\geq 2$ for some $i\leqc j\in T'$ then
  $M'(p)\geq G(\dist) \geq 2\cdot F'(p,i)$ for all $p\in \precond{i}$, using (\ref{r3}) and
  (\ref{r4}), so $M'[2\cdot\{i\}\rangle_{N'}$. Since $N'$ is
  a \hyperlink{finitary}{finitary} \hyperlink{scn}{structural conflict net}, it has no
  self-concurrency, so this is impossible.
\item Take $i\mathbin{\leqc} j \inp T'$ and $p\inp\precond{j}$. The case $i=j$ follows
  from (\ref{r4}), so assume $i<^\# j$.
  By (\ref{undocount}) we have $G(\ini[i]\cdot\fire) - G(\ini[i]\cdot\undone)\geq 0$.
  So by (\ref{p_j}), (\ref{rFp}), and (\ref{fetch}) $G(\dist)\geq 0$.
  Hence, using (\ref{rExec}), we may assume, w.l.o.g., that \plat{$G(\exec{j})=1$}.
  We need to investigate the same two cases as in the proof of (\ref{r4}) above.
  First let $t=\exec{j}$.  Then $M_1[\exec{j}\rangle$.  Thus,\vspace{-1ex}
  \hypertarget{proofr5}{$$\begin{array}[b]{@{}r@{~\geq~}lr@{}}
  \multicolumn{2}{@{}l}{G(\dist)}\\
  \mbox{}& \displaystyle
    F'(p,j)\cdot \big(G(\ini\cdot\fire) - G(\ini\cdot\undone)\big)
    + \hspace{-.5em}\sum_{h\geq^\# g\in\postcond{p}}\hspace{-.5em}F'(p,g)\cdot G(\textsf{fetch}_{g,h}^{p,j})
  & \mbox{(by (\ref{p_j}))} \\[-10pt]
  & F'(p,j)\cdot \big(G(\ini\cdot\fire) - G(\ini\cdot\undone)\big)
  & \hspace{-4em}\mbox{(by (\ref{rFp}) and (\ref{fetch}))} \\
  & F'(p,j)\cdot \big(G(\ini\cdot\fire) - G(\ini\cdot\und[\mbox{$\transin[i]{j}$}])\big)
  & \mbox{(by (\ref{took}))} \\
  & F'(p,j)\cdot \big(G(\trans[i]{j}\cdot\fire) - G(\trans[i]{j}\cdot\undone)
  & \mbox{(by (\ref{transin}))} \\
  & F'(p,j)\cdot \big(G(\trans[i]{j}\cdot\fire) - G(\trans[i]{j}\cdot\und[\mbox{$\transout[i]{j}$}])\big)
  & \mbox{(by (\ref{took}))} \\
  & F'(p,j)
  & \mbox{(by (\ref{transout}))}\makebox[0pt]{\,.}
  \end{array}$$}
  Now let $t=\comp[k]{l}$ with $k\in T'$. 
  We have $G=G_1+\{\comp[k]{l}\}-G^k_l$ for certain $l\geq^\# k$ with $G_1(\exec[k]{l})=1$.
  Since $G^i_{\!\!j}(\exec{j})\mathbin\geq 0$, we also have $G_1(\exec{j}) \geq 1$.
  By (\ref{rExec3}) this implies that $\neg(j\confeq k)$ or $(i,j)=(k,l)$.
  In the latter case \plat{$G(\exec{j})=\mbox{}$}
  $G_1(\exec{j})-G^i_{\!\!j}(\exec{j})=1-1=0$, contradicting our assumption.
  In the former case $p\notin\precond{k}$, so \plat{$G^k_l(\dist)=0$} and hence
  $G(\dist)=G_1(\dist)\geq F'(p,j)\cdot G_1(\exec{j})=F'(p,j)\cdot G(\exec{j})$.

\item Let $i\mathbin{\leqc} j \inp T'$ and $h<^\# j$.
  Since, for all $k\mathbin{\leqc} l\inp T'$,
  $G^k_l(\trans{j}\!\cdot\fire)\mathbin=\sum_\omega G^k_l(\trans{j}\!\cdot\reset[\omega])\mathbin=0$ and
  \plat{$G^k_l(\exec{j})=G^k_l(\fetched{j})$},
  the invariant is preserved when $t$ has the form \plat{$\comp[b]{c}\!$}. Using
  (\ref{pi-execute}), it is in danger only when $t=\exec{j}$ or $t=\trans{j}\!\cdot\reset[\omega]$
  for some $\omega$ with $\trans{j}\inp\UI_\omega$.

  First assume $M_1[\exec{j}\rangle$ and $T^h_j=G_1(\trans{j}\cdot\fire)-\sum_\omega G_1(\trans{j}\cdot\reset[\omega])=0$.
  Then
  $$\begin{array}[b]{r@{~\leq~}ll}
  \multicolumn{1}{r@{~\leq~}}{1}
  & G_1(\trans{j}\cdot\fire) - G_1(\trans{j}\cdot\und[\mbox{$\transout{j}$}])
  & \mbox{(by (\ref{transout}))} \\
  & G_1(\trans{j}\cdot\fire) - \sum_{\omega} G_1(\trans{j}\cdot\reset[\omega]) =0
  & \mbox{(by (\ref{took}))}, \\
  \end{array}$$
  which is a contradiction.

  Next assume \plat{$t=\trans{j}\!\cdot\reset[k]$} with $k \confeq j$, and $E^i_j=1$.
  By (\ref{rFp}) and (\ref{rExec}) the latter implies that \plat{$G_1(\exec{j})=1$} and \plat{$G_1(\fetched{j})=0$}\textsl{}.
  Then  
  $$\begin{array}[b]{r@{~\leq~}ll}
  \multicolumn{1}{r@{~=~}}{0}
  & G_1(\comp[k]{l})
  & \mbox{(by (\ref{r2}))} \\
  & G_1(\trans{j}\cdot\elide[k])+G_1(\trans{j}\cdot\reset[k])
  & \mbox{(by (\ref{reset}))} \\
  \multicolumn{1}{r@{~<~}}{} & G(\trans{j}\cdot\elide[k])+G(\trans{j}\cdot\reset[k]) \\
  & \sum_{l\geq^\#k}G(\fetched[k]{l})
  & \mbox{(by (\ref{reset}))}.
  \end{array}$$
  Hence $G_1(\fetched[k]{l})=G(\fetched[k]{l})>0$ for some $l\geq^\#k$, and by
  (\ref{fetch}) also $G_1(\exec[k]{l})>0$.
  Using (\ref{rExec3}) we obtain $(i,\!j)\mathop=(k,l)$, thereby obtaining a contradiction
  (\plat{$0\mathbin=G_1(\fetched{j})\mathbin=G_1(\fetched[k]{l})\mathbin>0$}).

\item 
Let $j\geq^\# i \confeq k \leq^\# l \in T'$ with $(i,j)\neq (k,\ell)$.
  The invariant is in danger only when \plat{$t\mathbin=\exec{j}$} or $t\mathbin=\exec[k]{l}$.
  W.l.o.g.\ let $t\mathbin=\exec[k]{l}$, with
  $G_1(\exec[k]{l})\mathbin=0$ and $G_1(\exec{j})\mathbin\geq 1$.

  Making a case distinction, first assume \plat{$G(\fetched{j})\mathbin\geq 1$}.
  Using (\ref{r3}), (\ref{r4}) and that $G(\exec[k]{l})=1$, $M'(p) \geq G(\dist)\geq
  F'(p,k)$ for all $p\in\precond{k}$.   Likewise, $M'(p) \geq G(\dist)\geq F'(p,i)$ for all $p\in\precond{i}$.
  Moreover, just as in \hyperlink{proofr4}{the proof of} (\ref{r4}), we derive, for all
  $p\in\precond{i}\cap\precond{k}$,
  $$\begin{array}[b]{@{}r@{~\geq~}ll@{}}
  \multicolumn{2}{@{}l}{M'(p)\geq G(\dist)} & \mbox{(by (\ref{r3}))} \\
  \mbox{}& \displaystyle
    F'(p,k)\cdot\big(G(\ini[k]\cdot\fire) - G(\ini[k]\cdot\undone)\big)
    + \hspace{-.7em}\sum_{h\geq^\# g\in\postcond{p}}\hspace{-.5em}F'(p,g)\cdot G(\textsf{fetch}_{g,h}^{p,k})
  & \mbox{(by (\ref{p_j}))} \\
  & \displaystyle
    F'(p,k)\cdot\big(G(\ini[k]\cdot\fire) - G(\ini[k]\cdot\undone)\big)
    + \hspace{-.7em}\sum_{h\geq^\# g\in\postcond{p}}\hspace{-.5em}F'(p,g)\cdot G(\fetched[g]{h})
  & \mbox{(by (\ref{fetch}))} \\
  & F'(p,k)\cdot\big(G(\ini[k]\cdot\fire) - G(\ini[k]\cdot\undone)\big) + F'(p,i)\cdot G(\fetched{j})
  & \mbox{(by (\ref{rFp}))} \\
  & F'(p,k)\cdot\big(G(\ini[k]\cdot\fire) - G(\ini[k]\cdot\und[\Pre^k_l])\big) + F'(p,i)\cdot G(\fetched{j})
  & \mbox{(by (\ref{took}))} \\
  & F'(p,k) + F'(p,i)
  & \mbox{(by (\ref{pre}))}.
  \end{array}$$
  It follows that $M'[\{k\}\mathord+\{i\}\rangle$. As $i\confeq k$ and $N'$ is a \hyperlink{finitary}{finitary}
  \hyperlink{scn}{structural conflict net}, this is impossible. (Note that this argument
  holds regardless whether $i=k$.)

  Now assume \plat{$G(\fetched{j})\leq 0$}.
  Then, in the notation of (\ref{pi-execute}), \plat{$E^i_j=1$}.
  Since $G_1(\exec[k]{l})=0$, (\ref{rFp}) and (\ref{fetch}) yield $G_1(\fetched[k]{l})=0$.
  Hence $G(\exec[k]{l})=1$ and $G(\fetched[k]{l})= 0$, so $E^k_l=1$.
  We will conclude the proof by deriving a contradiction from \plat{$E^i_j=E^k_l=1$}.
  In case $j=l$ this contradiction emerges immediately from (\ref{pi-execute}).
  By symmetry it hence suffices to consider the case $j< l$.

  By (\ref{r3}) and (\ref{r5}) we have $M'(p)\geq G(\dist)\geq F'(p,j)$ for all
  $p\in\precond{j}$, so $M'[j\rangle$. Likewise $M'[l\rangle$ and, using (\ref{r4}),
  $M'[i\rangle$ and $M'[k\rangle$. Since $j \confeq i \confeq k$ and
  $N'$ has no \hyperlink{M}{fully reachable \visible pure \structuralM},
  $j \confeq k$. Since $j \confeq k \confeq l$ and
  $N'$ has no \hyperlink{M}{fully reachable \visible pure \structuralM},
  $j \confeq l$. So $j<^\# l$.
  By (\ref{pi-execute}), using that \plat{$E^i_j=1$}, \plat{$T^j_l=0$}. This is in contradiction with
  $E^k_l=1$ and (\ref{rtrans}).

\item Suppose that $G(\exec{j})>0$ and $G(\exec[k]{l})>0$, with $i\leq^\# j \confeq k \leq^\# l \in T'$.
  By (\ref{r3}) and (\ref{r5}) we have $M'(p)\mathbin\geq G(\dist)\mathbin\geq F'(p,j)$ for all
  $p\inp\precond{j}$, so $M'[j\rangle$. Likewise, using (\ref{r4}),
  $M'[i\rangle$ and $M'[k\rangle$. Since $i\confeq j \confeq k$ and
  $N'$ has no \hyperlink{M}{fully reachable \visible pure \structuralM},
  $i \confeq k$. Using this, the result follows from (\ref{rExec2}).
  \hfill\filledbox
\end{enumerate}
\end{proofclaimNobox}

\begin{claim}{extra}
For any $M\in[M_0\rangle_N$ there exist $M'\in[M'_0\rangle_{N'}$ and
$G\fin \Int^T$ satisfying (\ref{r0})--(\ref{rLast}) from \refcl{reachable}, and
\begin{enumerate}[(A)]
\setcounter{enumi}{11}
\item there are no $j\geq^\# i \confeq k \leq^\# l \in T'$ with
  \plat{$M[\exec[i]{j}\rangle$} and \plat{$G(\exec[k]{l})>0$},\label{rExec4}
\item there are no $i\leq^\# j \confeq k \leq^\# l \in T'$ with
  \plat{$M[\exec[i]{j}\rangle$} and \plat{$G(\exec[k]{l})>0$},\label{rExec5}
\item if \plat{$M[\exec[i]{j}\rangle$} for $i\leqc j \in T'$ then $M'[j\rangle$.\label{rVeryLast}
\end{enumerate}
\end{claim}
\begin{proofclaimNobox}
Given $M$, by \refcl{reachable}(2) there are $M'$ and $G$ so that
the triple $(M,M',G)$ satisfies (\ref{r0})--(\ref{rLast}).
Assume \plat{$M[\exec[i]{j}\rangle$} for some $i\leqc j \in T'$.
Let $M_1:=M+\marking{\exec[i]{j}}$ and $G_1:=G+\{\exec{j}\}$.
By (\ref{rExec}) $G(\exec{j})\geq 0$, so $G_1(\exec{j})>0$.
By \refcl{reachable}(1) the triple ($M_1,M',G_1$) satisfies (\ref{r0})--(\ref{rLast}).
\begin{enumerate}[(A)]
\setcounter{enumi}{11}
\item
  Suppose \plat{$G(\exec[k]{l})>0$} for certain $l\geq^\# k \confeq i$.\
  In case $(i,j)=(k,\ell)$ we have $G_1(\exec{j})\geq 2$, contradicting (\ref{rExec}).
  In case $(i,j)\neq (k,\ell)$, $G_1$ fails (\ref{rExec2}), also a contradiction.
\item
  Suppose \plat{$G(\exec[k]{l})>0$} for certain $l\geq^\# k \confeq j$.
  Then $G_1$ fails (\ref{rExec}) or (\ref{rExec3}), a contradiction.
\item
  By (\ref{r3}) and (\ref{r5}) $M'(p)\geq G_1(\dist)\geq F(p,j)$ for all $p\in\precond{j}$, so $M'[j\rangle$.
\hfill\filledbox
\end{enumerate}
\end{proofclaimNobox}

\begin{claim}{concurrency}
If \plat{$M[\{\exec[i]{j}\}\mathord+\{\exec[k]{l}\}\rangle$} for some $M\in[M_0\rangle_N$
then $\neg (i \confeq k)$.
\end{claim}
\begin{proofclaim}
Suppose \plat{$M[\{\exec[i]{j}\}\mathord+\{\exec[k]{l}\}\rangle$} for some $M\in[M_0\rangle_N$.
By \refcl{reachable}(2) there exist $M'\in[M'_0\rangle_{N'}$ and
$G\fin \Int^T$ satisfying (\ref{r0})--(\ref{rLast}).
Let $M_1:=M+\marking{\exec[k]{l}}$ and \plat{$G_1:=G\mathbin+\{\exec[k]{l}\}$}.
By \refcl{reachable}(1) the triple $(M_1,M',G_1)$ satisfies (\ref{r0})--(\ref{rLast}).
Let $M_2:=M_1+\marking{\exec{j}}$ and \plat{$G_2:=G_1\mathbin+\{\exec{j}\}$}.
Again by \refcl{reachable}(1), also the triple $(M_2,M',G_2)$ satisfies (\ref{r0})--(\ref{rLast}).
By (\ref{rExec}) \plat{$G(\exec{j})\mathbin\geq 0$}, so in case $(i,j)\mathbin=(k,l)$ we obtain
\plat{$G_2(\exec{j})\mathbin\geq 2$}, contradicting
(\ref{rExec}). Hence $(i,j)\mathbin{\neq}(k,l)$. Moreover, $G_2(\exec[k]{l})>0$ and $G_2(\exec{j})>0$.
Now (\ref{rExec2}) implies $\neg (i \confeq k)$.
\end{proofclaim}
For any $t\in\{\ini,~ \trans{j}\}$ with
$h,j\inp T'$, and any \mbox{$\omega\inp \UI$} with $t\in\UI_\omega$, we write
$$t(\omega) := t\cdot\fire + t\cdot\undo[\omega] +
               \big(\sum_{f\in t^{\,\it far}} t\cdot\und\big) +
               t\cdot\undone+t\cdot\reset[\omega]\;.\vspace{-1ex}$$
The transition $t$ has no preplaces of type {\scriptsize \it in}, nor postplaces of type {\scriptsize \it out}.
By checking in \reftab{reversible} or \reffig{reversible} that each other place
occurs as often in $\precond{u(\omega)}+\postcond{(u\cdot\elide[\omega])}$ as in
$\postcond{u(\omega)}+\precond{(u\cdot\elide[\omega])}$, one verifies, for any $\omega\in
\UI$ with $t\in \UI_\omega$, that\vspace{-1ex}
\begin{equation}\label{elide}
\marking{t(\omega)} = \marking{t\cdot\elide[\omega]}.
\end{equation}
Let $\equiv$ be the congruence relation on finite signed multisets of transitions
generated by
\begin{eqnarray}\label{elideNF}
t(\omega) &\equiv& t\cdot\elide[\omega]
\end{eqnarray}
for all $t \in \{\ini,~ \trans{j} \mid h,j\inp T'\}$ and $\omega\in \UI$ with
$\UI_\omega\ni t$.
Here \emph{congruence} means that $G_1\mathbin\equiv G_2$ implies $k\cdot G_1 \mathbin\equiv k\cdot G_2$ and
$G_1+H \mathbin\equiv G_2+H$ for all $k\inp\Int$ and $H\fin \Int^T$.
Using (\ref{elide}) $G_1\equiv G_2$ implies $\marking{G_1}=\marking{G_2}$.

\begin{claim}{0}
If $M'=\marking{G}$ for $M'\in \Int^{S'}$ and $G\fin \Int^T$ such that for all $i\in T'$
we have $G(\comp{j})=0$ and either $\forall j\geq^\# i.~G(\exec{j})\geq 0$ or $\forall
j\geq^\# i.~G(\exec{j})\leq 0$, then $G \equiv \emptyset$.
\end{claim}

\begin{proofclaim}
  Let $M'$ and $G$ be as above. 
  W.l.o.g.\ we assume $G(t\cdot\elide[\omega])=0$ for all $t\in\{\ini,~
  \trans{j}\}$ and all $\omega\in\UI$ with $t\in \UI_\omega$, for any $G$ can be brought
  into that form by applying (\ref{elideNF}).
  For each $s\in S\setminus S'$ we have $M'(s)=0$, and using this the inequations
  (\ref{undo})--(\ref{fetch}) and (\ref{p_j}) of \refcl{G-properties} turn into equations.
  For each $i\in T'$ we have $G(\sum_{j\geq^\# i}\exec{j})=0$, using (the equational form
  of) (\ref{undo})--(\ref{interfacecount}), and that $G(\comp{j})=0$. Since
  \plat{$G(\exec{j})\geq 0$} (or $\mbox{}\leq 0$) for all $j\geq^\# i$,
  this implies that \plat{$G(\exec{j})= 0$} for each $i\leqc j\in T'$.
  With (\ref{fetch}) we obtain $G(\fetched{j})=G(\fetch)=0$ for each applicable $p,c,i,j$.
  Using that $G(t\cdot\elide[\omega])=0$ for each applicable $t$ and $\omega$, with
  (\ref{reset})--(\ref{took}) and (\ref{p_j}) we find $G(t)=0$ for all $t\in T$.
\end{proofclaim}

\begin{claim}{D}
Let $M:=M'+(M_0\mathord-M'_0)+\marking{H}\in[M_0\rangle_N$ for $M'\inp[M'_0\rangle_{N'}$ and
$H\fin \Int^T$ with \plat{$H(\exec{j})\mathbin= 0$} for all $i\leqc j\in T'$.
\begin{enumerate}[(a)]
\item If \plat{$H(\comp[i]{j})<0$} and \plat{$H(\comp[k]{l})<0$} for certain
  $i,k \in T'$ then $\neg(i \mathrel{\#} k)$.\label{Hcomp}
\item If \plat{$M[\exec{j}\rangle$} and \plat{$H(\comp[k]{l})<0$} for certain
  $i,k \in T'$ then $\neg(i \confeq k)$ and $\neg(j \confeq k)$.\label{Hexec}
\item $H(\dist)\geq 0$ for all $p\in S'$ (with $\postcond{p}\neq\emptyset$).\label{dist-positive-H}
\item Let $c\confeq i \in T'$.
  If $H(\dist)\geq F'(p,c)$ for all $p\in\precond{c}$, then $H(\comp{j})=0$.\label{dist-final}
\item If $M[\exec{j}\rangle$ with $i\leqc j\in T'$ then $M'[j\rangle$.\label{Hexecj}
\end{enumerate}
\end{claim}

\begin{proofclaimNobox}
By \refcl{extra} there exist $M'_1\in[M'_0\rangle_{N'}$ and $G_1\fin \Int^T$
satisfying (\ref{r1})--(\ref{rVeryLast}) (with $M$, $M'_1$ and $G_1$ playing the r\^oles of $M$, $M'$ and $G$).
In particular, $M=M'_1+(M_0-M'_0)+\marking{G_1}$, $G_1(\comp{j}) = 0$ for all $i\in T'$,
and $G_1(\exec{j})\geq 0$ for all $i\leqc j\in T'$. Using (\ref{rExec2}), for each $i\in
T'$ there is at most one $j\geq^\#i$ with \plat{$G_1(\exec{j})>0$}; we denote this $j$ by $f(i)$,
and let $f(i):=i$ when there is no such $j$. This makes $f:T'\rightarrow T'$ a function,
satisfying $G_1(\exec{j})=0$ for all $j\geq^\# i$ with $j\neq f(i)$.

Given that \plat{$H(\exec{j})\mathbin=0$} for all $i\leqc j\in T'$,
(\ref{undo})--(\ref{interfacecount}) (or (\ref{reset}) and (\ref{fetch})) imply
$H(\comp{j})\leq 0$ for all $i\in T'$.
Let $M'_2:=M'+\sum_{i \inp T'} H(\comp{j})\cdot \marking{i}$
and $G_2:=H-\sum_{i \inp T'}H(\comp{j})\cdot G^i_{\!f(i)}$, where $G^i_{\!\!j}$ is the
right-hand side of (\ref{mimic}).
Then $M = M'+(M_0-M'_0)+\marking{H} = M'_2+(M_0-M'_0)+\marking{G_2}$,
using that \plat{$\marking{i}=\marking{G^i_{\!f(i)}}$}.
Moreover, $G_2(\comp{j})=0$ for all $i\inp T'$, using that \plat{$G^i_{\!f(i)}(\comp{j})=1$}.

It follows that $M'_1-M'_2 = \marking {G_2-G_1}$. Moreover, we have $(G_2-G_1)(\comp{j})=0$ for
all $i\in T'$. We proceed to show that $G_2-G_1$ satisfies the remaining precondition of \refcl{0}.
So let $i \in T'$. In case \plat{$H(\comp{j})=0$}, for all $j\geq^\# i$ we have
\plat{$G_2(\exec{j})= 0$}, and $G_1(\exec{j})\geq 0$ by (\ref{rExec}). Hence \plat{$(G_2-G_1)(\exec{j})\leq 0$}.
In case \plat{$H(\comp{j})<0$}, we have \plat{$G_2(\exec{{f(i)}})\geq 1$}, and hence, using (\ref{rExec}),
\plat{$(G_2-G_1)(\exec{{f(i)}})\geq 0$}. Furthermore, for all $j\neq f(i)$,
\plat{$G_2(\exec{j})\geq 0$} and \plat{$G_1(\exec{j})=0$}, so again $(G_2-G_1)(\exec{j})\geq 0$.

Thus we may apply \refcl{0}, which yields $G_2 \equiv G_1$. It follows that $M'_2=M'_1\in[M'_0\rangle_{N'}$.
\begin{enumerate}[(a)]
\item
Suppose that \plat{$H(\comp[i]{j})<0$} and \plat{$H(\comp[k]{l})<0$} for certain $i\mathrel\#k \in T'$.
Then \plat{$G_2(\exec{{f(i)}})>0$} and $G_2(\exec[k]{{f(k)}})>0$, so
$G_1(\exec{{f(i)}})>0$ and $G_1(\exec[k]{{f(k)}})>0$, contradicting (\ref{rExec2}).
\item
Suppose that \plat{$M[\exec{j}\rangle$} and \plat{$H(\comp[k]{l})<0$} for certain
  $k\confeq i$ or $k\confeq j$.
Then \plat{$G_1(\exec[k]{{f(k)}})=\mbox{}$}\linebreak[1]\plat{$G_2(\exec[k]{{f(k)}})>0$},
contradicting (\ref{rExec4}) or (\ref{rExec5}).
\item
By (\ref{Hcomp}), for any given $p\in S'$ there is at most one $i\in\postcond{p}$ with $H(\comp{j})<0$.
For all $i\in T'$ with $i\notin\postcond{p}$ we have $G^i_{\!f(i)}(\dist)=0$.
First suppose $k\in\postcond{p}$ satisfies $H(\comp[k]{{f(k)}})<0$.
Then 
$$G_1(\exec[k]{{f(k)}})\begin{array}[t]{@{~=~}l}G_2(\exec[k]{{f(k)}})\\
   H(\exec[k]{{f(k)}})-\sum_{i \in T'}H(\comp{j})\cdot G^i_{\!f(i)}(\exec[k]{{f(k)}})\\
   0-H(\comp[k]{{f(k)}}),\end{array}$$
so by (\ref{r4}) $G_1(\dist)\geq -F'(p,k)\cdot H(\comp[k]{{f(k)}})$.
Hence
$$H(\dist)~\begin{array}[t]{@{}l}=~G_2(\dist)+\sum_{i\in T'}H(\comp{j})\cdot G^i_{\!f(i)}(\dist)\\
   =~ G_1(\dist)+H(\comp[k]{{f(k)}})\cdot G^{k}_{f(k)}(\dist)\\
   \geq~ -F'(p,k)\cdot H(\comp[k]{{f(k)}}) + H(\comp[k]{{f(k)}})\cdot F'(p,k)=0.
  \end{array}$$
In case there is no $i\in\postcond{p}$ with $H(\comp{j})<0$ we have
$$H(\dist)=G_2(\dist)+\sum_{i\in T'}H(\comp{j})\cdot G^i_{\!f(i)}(\dist)=G_1(\dist)\geq 0\vspace{-2ex}$$
by (\ref{r4}) and (\ref{rExec}).

\item
Since \plat{$H(\comp{j})\leq 0$} and \plat{$G^i_{\!f(i)} (\dist)\geq 0$} for all $i \inp T'$,
also using (\ref{dist-positive-H}),
all summands in \plat{$H(\dist)+\sum_{i \in T'}-H(\comp{{f{i}}})\cdot G^i_{\!f(i)}(\dist)$} are positive.
Now suppose \plat{$H(\comp{j})<0$} for certain $i\inp T'$.
Then, using (\ref{r3}), for all $p\in\precond{i}$,
$$M'_1(p)\geq G_1(\dist) = G_2(\dist) \geq G^i_{\!f(i)}(\dist)=F'(p,i).$$
Furthermore, let $c \confeq i$ and suppose $H(\dist)\geq F'(p,c)$ for all $p\in\precond{c}$.
Then, using (\ref{r3}),
$$M'_1(p)\geq G_1(\dist) = G_2(\dist) \geq H(\dist)\geq F'(p,c)$$
for all $p\in\precond{c}$.
Moreover, if $p\in \precond{c}\cap\precond{i}$ then
$$M'_1(p)\geq G_2(\dist) \geq H(\dist)+G^i_{\!f(i)}(\dist) \geq F'(p,c)+F'(p,i).$$
Hence $M'_2 [\{c\}\mathord+\{i\}\rangle$. However, since $c \confeq i$ and $N'$ is a
\hyperlink{scn}{structural conflict net}, this is impossible.

\item Suppose $M[\exec{j}\rangle$ with $i\leqc j\in T'$.
Then $M'_1[j\rangle$ by (\ref{rVeryLast}).
Now $M'=M'_1+\sum_{k \inp T'} -H(\comp[k]{{f(k)}})\cdot \marking{k}$, with $-H(\comp[k]{j})\geq 0$
for all $k\in T'$. Whenever $-H(\comp[k]{j})>0$ then $\neg(j \confeq k)$ by (\ref{Hexec}).
Hence $M'[j\rangle$.
\hfill\filledbox
\end{enumerate}
\end{proofclaimNobox}
\bigskip
We now define the class $\NF\subseteq \Int^T$ of signed multisets of transitions in
\emph{normal form} by $H\in\NF$ iff $\ell(H)\equiv\emptyset$ and, for all $t\in\{\ini,~
\trans{j} \mid h,j\inp T'\}$:
\begin{enumerate}[(NF-1)]
\item $H (t\cdot\elide[\omega]) \leq 0$ for each  $\omega\inp \UI$,\label{NF1}
\item $H (t\cdot\undo[\omega]) \geq 0$ for each  $\omega\inp \UI$, or $H (t\cdot\fire) \geq 0$,\label{NF2}
\item and if $H (t\cdot\elide[\omega]) < 0$ for any  $\omega\inp \UI$,
  then $H (t\cdot\undo[\omega]) \leq 0$ and $H (t\cdot\fire) \leq 0$.\label{NF3}
\end{enumerate}
We proceed verifying the remaining conditions of \refthm{3ST}.
\begin{enumerate}
\item[\ref{normalformST}.]
By applying (\ref{elideNF}), each signed multiset $G\fin\Int^T$ with $\ell(G)\equiv\emptyset$
can be converted into a signed multiset \mbox{$H\fin\NF$} with $\ell(H)\equiv\emptyset$, such that
$\marking{H}=\marking{G}$. Namely, for any $t\in\{\ini,~
\trans{j} \mid h,j\inp T'\}$, first of all perform the following
three transformations, until none is applicable:
\begin{enumerate}[(i)]
\item correct a positive count of a transition $t\cdot\elide[\omega]$ in $G$ by adding
  $t(\omega)-t\cdot\elide[\omega]$ to $G$;
\item if both $H(t\cdot\undo[\omega])<0$ for some $\omega$ and $H(t\cdot\fire)<0$,
  correct this in the same way;
\item and if, for some $\omega$, $t\mathord\cdot\elide[\omega]$ has a negative and
  $t\mathord\cdot\undo[\omega]$ a positive count, add $t\cdot\elide[\omega]-t(\omega)$.
\end{enumerate}
Note that transformation (iii) will never be applied to the same $\omega$ as (i) or (ii),
so termination is ensured. Properties (NF-\ref{NF1}) and (NF-\ref{NF2}) then hold for $t$.
After termination of (i)--(iii), perform
\begin{enumerate}[(i)]
\item[(iv)] if, for some $\omega$, $H(t\cdot\elide[\omega])<0$ and
  $H(t\cdot\fire)>0$, add $t\cdot\elide[\omega]-t(\omega)$.
\end{enumerate}
This will ensure that also (NF-\ref{NF3}) is satisfied, while preserving (NF-\ref{NF1}) and (NF-\ref{NF2}).

Define the function $f:T\rightarrow\nat$ by $f(u):=1$ for all $u\in T$ not of the form
$u=t\cdot\elide[\omega]$, and $f(t\cdot\elide[\omega]):=f(t(\omega))$ (applying the last
item of \refdf{multiset}). Then surely $f(G)=f(H)$.

\item[\ref{lastST}.] 
  Let $M'\in\nat^{S'}$, $U'\in\nat^{T'}$ and $U\in\nat^{T}$ with $\ell(U)=\ell'(U')$
  and $M'+\!\precond{U'}\in[M'_0\rangle_{N'}$.
  Since $N'$ is a \hyperlink{finitary}{finitary} \hyperlink{scn}{structural conflict net}, it admits no
  self-concurrency, so, as $\precond{U'}\leq M'+\!\precond{U'}\in[M'_0\rangle_{N'}$,
  the multiset $U'$ must be a set. As $N'$ is \hyperlink{plain}{plain}, this implies that the multiset $\ell'(U')$ is a set.
  Since $\ell(U)=\ell'(U')$, also $\ell(U)$, and hence $U$, must be a set.
  All its elements have the form $\exec{j}$ for  $i\leqc j\in T'$,
  since these are the only transitions in $T$ with visible labels.
  Note that $U'$ is completely determined by $U$, namely by
  \plat{$U'=\{i\mid \exists j.~\exec{j}\in U\}$}.
  We take\vspace{-1ex}
  $$H_{M',U}:= \sum_{p\in S'} (M'\mathord+\!\precond{U'})(p)\cdot\{\dist\} +
  \!\!\!\!\!\!\!\! \sum_{(M'+\!\precond{U'})[j\rangle}\!\!\!\!\left(\{\ini\cdot\fire\} + 
  \hspace{-2.7em} \sum_{h<^\#j,~\nexists\exec[g]{h}\in U}\hspace{-2.6em} \{\trans[h]{j}\cdot\fire\}\right)
  $$
  Since $N'$ is \hyperlink{finitary}{finitary}, \plat{$H_{M',U}\fin\nat^{T^+}$}. Moreover,
  $\ell(H_{M',U})\equiv\emptyset$.

  Let $H\mathbin{\fin} \NF$ with $M:=M'+\!\precond{U'}+(M_0\mathord-M'_0)+\marking{H}-\precond{U}\in\nat^S$
  and $M+\precond{U}\in[M_0\rangle_N$.
  Since $H\inp\NF$, and thus $\ell(H)\equiv\emptyset$, $H(\exec[i]{j})=0$.
  From here on we apply \refcl{G-properties} and \refcl{D} with $M+\precond{U}$ and
  $M'+\precond{U'}$ playing the r\^oles of $M$ and $M'$.
  Note that the preconditions of these claims are met.

  That $H(\exec[i]{j})=0$ for all $i\leqc j\inp T'$, together with (\ref{undo}) and the
  requirements (NF-\ref{NF1}) and \mbox{(NF-\ref{NF3})} for normal forms, yields
  $H(t\cdot\elide)\leq 0$ as well as $H(t\cdot\undo)\leq 0$.
  Using this, (\ref{reset})--(\ref{fetch}) imply that
  \begin{equation}\label{T-negative}
  H(u)\leq 0 ~\mbox{ for each }~ u\in T^-.
  \vspace{-4ex}
  \end{equation}
  \begin{claim}{C} Let $c \inp T'$ and $p\in\precond{c}$. Then
  \begin{itemize}
  \item if $H(\ini[c]\cdot\fire)>0$ then $H(\fetch)=0$ for all $i\in\postcond{p}$ and $j\geq^\# i$, and
  \item if $H(\trans[b]{c}\cdot\fire)>0$ for some $b<^\#c$ then $H(\fetch)=0$ for all $i\in\postcond{p}$ and $j\geq^\# i$.
  \end{itemize}
  \end{claim}
  \begin{proofclaim} Suppose that $H(t\cdot\fire)>0$, for $t=\ini[c]$ or \plat{$t=\trans[b]{c}$}.
  Then (\ref{pi-j}) resp.\ (\ref{pi-execute}) together with (\ref{T-negative}) implies that
  $H(t\cdot\reset[\omega])=0$ for each $\omega$ with $t\in\UI_\omega$.
  In order words, $H(t\cdot\reset)=0$ for each $i\confeq c$, so in particular for each $i\in\postcond{p}$.
  Furthermore, $H(t\cdot\elide)\geq 0$, by requirement (NF-\ref{NF3}) of normal forms.
  With (\ref{reset}), this yields $\sum_{j\geq^\#i}H(\fetched{j})\geq 0$, and
  (\ref{T-negative}) implies $H(\fetched{j})= 0$ for each $j\geq^\#i$.
  Now (\ref{fetch},\,\ref{T-negative}) gives $H(\fetch)= 0$ for each $j\geq^\#i\in\postcond{p}$.
  \end{proofclaim}
  We proceed to verify the requirements (\ref{markingST})--(\ref{concurrent}) of \refthm{3ST}.

  \begin{enumerate}[(a)]
  \item[(\ref{markingST})] To show that $M_{M',U}\in\nat^S$, it suffices to apply it to the preplaces of
    transitions in $H_{M',U}+U$:\vspace{-1.5ex}
    $$\begin{array}{@{}l@{~=~}ll}
    M_{M',U}(p) & 0 & \mbox{for all }p\in S'\;;\\
    M_{M',U}(p_j) & \left\{\begin{array}{@{}ll@{}}
     (M'+\!\precond{U'})(p)-F'(p,j) & \mbox{if } (M'+\!\precond{U'})[j\rangle  \\
     (M'+\!\precond{U'})(p)           & \mbox{otherwise}
     \end{array}\right. & \mbox{for }p\inp S',~j\inp\postcond{p};\\
    M_{M',U}(\pi_j) & \left\{\begin{array}{@{}l@{\qquad\;\quad}l@{}}
     \phantom{-}0 & \mbox{if } (M'+\!\precond{U'})[j\rangle  \\
     \phantom{-}1 & \mbox{otherwise}
     \end{array}\right. & \mbox{for }j\in T';\\
    M_{M',U}(\Pre^j_k) & \hspace{-1.6pt}\left\{\begin{array}{@{}l@{\qquad\;\quad}l@{}}
     \phantom{-}1 & \mbox{if } (M'+\!\precond{U'})[j\rangle \wedge \exec[j]{k}\notin U \\
     -1 & \mbox{if } \neg(M'+\!\precond{U'})[j\rangle \wedge \exec[j]{k}\in U \\
     \phantom{-}0 & \mbox{otherwise}
     \end{array}\right. & \mbox{for }j\leqc k\in T';\\
    M_{M',U}(\pi_{h\#j}) &  \left\{\begin{array}{@{}l@{\qquad\;\quad}l@{}}
     \phantom{-}0 & \mbox{if } \exists\exec[g]{h}\in U \vee (M'+\!\precond{U'})[j\rangle\\
     \phantom{-}1 & \mbox{otherwise}
     \end{array}\right. & \mbox{for }h<^\# j\in T'\\
    M_{M',U}(\transin{j}) &  \left\{\begin{array}{@{}l@{\qquad\;\quad}l@{}}
     \phantom{-}1 & \mbox{if } (M'+\!\precond{U'})[j\rangle \wedge \exists\exec[g]{h}\in U \\
     \phantom{-}0 & \mbox{otherwise}
     \end{array}\right. & \mbox{for }h<^\#j\in T';\\
    M_{M',U}(\transout{j}) & \left\{\begin{array}{@{}l@{\qquad\;\quad}l@{}}
     \phantom{-}1 & \makebox[0pt][l]{if $(M'+\!\precond{U'})[j\rangle \wedge \nexists\exec[g]{h}\in U 
                                                   \wedge \nexists\exec{j}\in U$} \\
     -1 & \makebox[0pt][l]{if $\big(\neg(M'+\!\precond{U'})[j\rangle \vee \exists\exec[g]{h}\in U\big) 
                                                   \wedge \exists\exec{j}\in U$} \\
     \phantom{-}0 & \mbox{otherwise}
     \end{array}\right. & \begin{array}{@{}l@{}}\mbox{}\\\mbox{}\\\mbox{for }h<^\#j\in T'.\end{array}\\
    \end{array}$$
    For all these places $s$ we indeed have that $M_{M',U}(s)\geq 0$,
    for the circumstances yielding the two exceptions above cannot occur:
    \begin{itemize}
    \item Suppose \plat{$\exec[j]{k}\in U$} with $j\leqc k\in T'$. Then $j\in U'$, so
      $\precond{j} \subseteq M'+\!\precond{U'}$ and $(M'+\!\precond{U'})[j\rangle$.
      Consequently, $M_{M',U}(\Pre^j_k) \neq -1$ for all $j\leqc k \in T'$.
    \item Suppose $\exec{j}\in U$ with $i\leqc j\in T'$. Then $\precond{\exec{j}}\leq
      \precond{U}$, so $(M+\!\precond{U})[\exec{j}\rangle$.
      \refcl{D}(\ref{Hexecj}) with $M+\!\precond{U}$ and $M'+\!\precond{U'}$ in the
      r\^oles of $M$ and $M'$ yields $(M'+\precond{U'})[j\rangle$.

      If moreover $\exec[g]{h}\inp U$ with $g\mathbin{\leqc} h \mathbin{<^\#}\! j$, then
      $\{g\}\mathord+\{i\}\subseteq U'$, so $\precond\{g\}\mathord+\!\precond\{i\}\subseteq
      M'\mathord+\!\precond{U'}$ and
      $(M'+\precond{U'})[\{g\}\mathord+\{i\}\rangle$. In particular, $g\concurrent i$, and since $N'$
      is a \hyperlink{scn}{structural conflict net}, $\precond{g}\cap\precond{i}=\emptyset$.
      By \refcl{D}(\ref{Hexecj})---as above---$(M'\mathord+\!\precond{U'})[h\rangle$, so
      $\precond{g}\cup\precond{h}\cup \precond{j}\cup\precond{i}\subseteq
        M'\mathord+\!\precond{U'} \inp [M'_0\rangle_{N'}$.
      Moreover, since $g \leqc h <^\# j \geq^\# i$, we have
      $\precond{g}\cap\precond{h}\neq\emptyset$,
      $\precond{h}\cap\precond{i}\neq\emptyset$ and $\precond{i}\cap\precond{j}\neq\emptyset$.
      Now in case also $\precond{h}\cap\precond{i}\neq\emptyset$, the transitions $g$, $h$
      and $i$ constitute a \hyperlink{M}{fully reachable pure $\structuralM$};
      otherwise $h\concurrent i$ and  $h$, $j$
      and $i$ constitute a \hyperlink{M}{fully reachable pure $\structuralM$}.
      Either way, we obtain a contradiction.
      Consequently, $M_{M',U}(\transout{j}) \neq -1$ for all $h<^\# j \in T'$.
    \end{itemize}
  \item[(\ref{matchST})] Suppose $M'\goesto[a]$; say $M'[i\rangle$ with $\ell'(i)=a$.
    Let $j$ be the largest transition in $T'$ w.r.t.\ the well-ordering $<$ on $T$
    such that $i\leqc j$ and $(M'+\!\precond{U'})[j\rangle$.
    It suffices to show that \plat{$M_{M',U} [\exec{j}\rangle$}, \ie that
    $M_{M',U}(\Pre^i_j)\mathord=1$, $M_{M',U}(\transout{j})\mathord=1$ for all $h\mathbin{<^\#}\!j$,
    and $M_{M',U}(\pi_{j\#l})\mathord=1$ for all $l\mathbin{>^\#}\!j$.
 
    If \plat{$\exec{j}\in U$} we would have $i\in U'$
    and hence $(M'+\!\precond{U'})[2\cdot\{i\}\rangle$.
    Since $N'$ is a \hyperlink{finitary}{finitary} \hyperlink{scn}{structural conflict net}, this is impossible.
    Therefore \plat{$\exec{j}\not\in U$} and, using the calculations from (a) above,
    $M_{M',U}(\Pre^i_j)=1$.
   
    Let $h<^\#j$. To establish that $M_{M',U}(\transout{j})=1$ we need to show that
    there is no $k\leqc j$ with \plat{$\exec[k]{j}\in U$} and no $g\leqc h$ with
    \plat{$\exec[g]{h}\in U$}. First suppose \plat{$\exec[k]{j}\in U$} for some $k\leqc j$.
    Then $k\in U'$ and hence $(M'+\!\precond{U'})[\{i\}\mathord+\{k\}\rangle$.
    This implies $i \smile k$, and, as $N'$ is a structural conflict net, $\precond{i}\cap \precond{k}=\emptyset$.	
    Hence the transitions $i$, $j$ and $k$ are all different, with $\precond{i}\cap \precond{j}\neq\emptyset$ and
    $\precond{j}\cap \precond{k}\neq\emptyset$ but $\precond{i}\cap \precond{k}=\emptyset$.
    Moreover, the reachable marking $M'+\!\precond{U'}$
    enables all three of them. Hence $N'$ contains a \hyperlink{M}{fully reachable pure $\structuralM$},
    which contradicts the assumptions of \refthm{correctness}.

    Next suppose \plat{$\exec[g]{h}\in U$} for some $g\leqc h$.
    Then $(M+\!\precond{U})[\exec[g]{h}\rangle$, so $(M'+\!\precond{U'})[h\rangle$
    by \refcl{D}(\ref{Hexecj}). Moreover, $g \in U'$, so $(M'+\!\precond{U'})[\{i\}\mathord+\{g\}\rangle$.
    This implies $g \smile i$, and $\precond{g}\cap \precond{i}=\emptyset$.
    Moreover, $\precond{g}\cap \precond{h}\neq\emptyset$, $\precond{h}\cap \precond{j}\neq\emptyset$ and
    $\precond{j}\cap \precond{i}\neq\emptyset$, while the reachable marking $M'+\!\precond{U'}$
    enables all these transitions. Depending on whether $\precond{h}\cap \precond{i}=\emptyset$,
    either $h$, $j$ and $i$, or $g$, $h$ and $i$ constitute a 
    \hyperlink{M}{fully reachable pure $\structuralM$},
    contradicting the assumptions of \refthm{correctness}.

    Let $l>^\#j$. To establish that \plat{$M_{M',U}(\pi_{j\#l})=1$} we need to show that
    there is no $k\leqc j$ with \plat{$\exec[k]{j}\in U$}---already done above---and that
    $\neg(M'+\!\precond{U'})[l\rangle$.
    Suppose $(M'+\!\precond{U'})[l\rangle$.
    Considering that $j$ was the largest transition with $i\leqc j$ and
    $(M'+\!\precond{U'})[j\rangle$, we cannot have $i<^\# l$.
    Hence the transitions $i$, $j$ and $l$ are all different, with $\precond{i}\cap \precond{j}\neq\emptyset$ and
    $\precond{j}\cap \precond{l}\neq\emptyset$ but $\precond{i}\cap \precond{l}=\emptyset$.
    Moreover, the reachable marking $M'+\!\precond{U'}$
    enables all three of them. Hence $N'$ contains a \hyperlink{M}{fully reachable pure $\structuralM$},
    which contradicts the assumptions of \refthm{correctness}.
  \item [(\ref{upperboundST})]
  We have to show that $H(t)\leq H_{M',U}(t)$ for each $t\in T$.
  \begin{enumerate}[i.]
  \item[$\bullet$]
    In case $t\in T^-$ this follows from (\ref{T-negative}) and \plat{$H_{M',U}\in\nat^{T^+}\!\!$}.
  \item[$\bullet$]
    In case $t=\exec{j}$ it follows since $\ell(H)\equiv\emptyset$.
  \item[$\bullet$]
    In case $t=\dist$ it follows from (\ref{p}) and (\ref{T-negative}).
  \item[$\bullet$]
    Next let $t=\ini[c]\cdot\fire$ for some $c\in T'$.
    In case $H(\ini[c]\cdot\fire)\leq 0$ surely we have $H(\ini[c]\cdot\fire)\leq H_{M',U}(\ini[c]\cdot\fire)$.
    So without limitation of generality we may assume that $H(\ini[c]\cdot\fire)>0$.
    By (\ref{pi-j},\,\ref{T-negative}) we have $H(\ini[c]\!\cdot\fire)=1$.
    Using (\ref{p_j}), \refcl{C}, (\ref{T-negative}) and (\ref{p}) we obtain, for all $p\in\precond{c}$,
    $$F'(p,c)\cdot H(\ini[c]\cdot\fire) \leq H(\dist) \leq (M'+\!\precond{U'})(p).$$
    Hence $c$ is enabled under  $M'+\!\precond{U'}$, which implies $H_{M',U}(\ini[c]\cdot\fire)=1$.
  \item[$\bullet$]
    Let $t\mathbin=\trans[b]{c}\cdot\fire$ for some $b\mathbin{<^\#}\! c\inp T'\!$.
    As above, we may assume $H(\trans[b]{c}\!\cdot\fire)\mathbin>0$.
    By (\ref{pi-execute},\,\ref{T-negative}) we have $H(\trans[b]{c}\!\cdot\fire)=1$.
    Using (\ref{T-negative}) and that $H(\exec[g]{b})=0$ for all $g\leqc b$, it follows
    that $(M+\!\precond{U})(\pi_{b\#c})=0$. Hence $\neg(M+\!\precond{U})[\exec[g]{b}\rangle$ for
    all $g\leqc b$, and thus $\nexists \exec[g]{b} \in U$.
    For all $p\in\precond{c}$ we derive
    $$\begin{array}{@{}r@{~\leq~}ll}
      \multicolumn{2}{@{}l@{}}{F'(p,c)\cdot H(\trans[b]{c}\cdot\fire)}\\
    \mbox{} & F'(p,c)\cdot\big(H(\trans[b]{c}\cdot\fire)-H(\trans[b]{c}\cdot\undone)\big)&
      (\ref{T-negative})\\
    & F'(p,c)\cdot\big(H(\ini[c]\cdot\fire)-H(\ini[c]\cdot\und[\mbox{$\transin[b]{c}$}])\big)&
      (\ref{transin})\\
    & F'(p,c)\cdot\big(H(\ini[c]\cdot\fire)-H(\ini[c]\cdot\undone)\big)&
      (\ref{took})\\
    \multicolumn{1}{r@{~=~}}{\mbox{}}
    & \displaystyle
      \mbox{[the same as above]}
      + \hspace{-.5em}\sum_{j\geq^\# i\in \postcond{p}}\hspace{-.5em}
      F'(p,i) \cdot H(\fetch) &
      (\mbox{\refcl{C}})\\[-10pt]
    & H(\dist) & (\ref{p_j}) \\
    & \displaystyle
      (M'+\!\precond{U'})(p)+ \hspace{-.5em}\sum_{\{i\in T'\mid
        p\in\postcond{i}\}}\hspace{-.5em} H(\comp{j}) & (\ref{p})\\[-10pt]
    & (M'+\!\precond{U'})(p) & (\ref{T-negative}).
    \end{array}$$
    Hence $(M'+\!\precond{U'})[c\rangle$, and thus $H_{M',U}(\trans[b]{c})=1$.
  \end{enumerate}
  \item[(\ref{T-ST})]
    If $u\notin T^-$, yet $H(u)\neq 0$, then $u$ is either $\dist$, $\ini\cdot\fire$
    or $\trans{j}\cdot\fire$ for suitable $p\in S'$ or $h,j\in T'$.
    For $u=\dist$ the requirement follows from \refcl{D}(\ref{dist-positive-H});
    otherwise Property (NF-\ref{NF2}), together with (\ref{took}), guarantees that
    $H(u)\geq 0$.
  \item[(\ref{disjoint preplacesST})]
    If $H(t)\mathbin>0$ and $H(u)\mathbin<0$, then $t\inp T^+$ and $u\inp T^-$.
    The only candidates for $\precond{t}\cap\precond{u}\neq\emptyset$ are
    \begin{itemize}
    \item \plat{$p_c \in \precond{(\ini[c]\cdot\fire)} \cap \precond{(\fetch)}$}
      for $p\in S'$, $c,i\in\postcond{p}$ and $j\geq^\# i$,
    \item $\transin[b]{c} \in \precond{(\trans[b]{c}\cdot\fire)}
                    \cap \precond{(\ini[c]\cdot\und[{\transin[b]{c}}])}$ for $b\leqc c\in T'$.
    \end{itemize}
    We investigate these possibilities one by one.
    \begin{itemize}
    \item $H(\ini[c]\cdot\fire)>0 \wedge H(\fetch)<0$ cannot occur by \refcl{C}.
    \item Suppose $H(\trans[b]{c}\cdot\fire)>0$.
      By (\ref{pi-execute},\,\ref{T-negative}) we have $H(\trans[b]{c}\!\cdot\fire)=1$.
      Through the derivation above, in the
      proof of requirement (c), using (\ref{T-negative},\,\ref{transin},\,\ref{took}),
      \refcl{C} and (\ref{p_j}), we obtain $H(\dist)\geq F'(p,c)$ for all $p\in\precond{c}$.
      Now \refcl{D}(\ref{dist-final}) yields $H(\comp{j})=0$ for all $i\confeq c$. By (\ref{reset}) and
      (\ref{T-negative}) we obtain $H(\ini[c]\!\cdot\reset)\mathbin=0$ for each such $i$.
      Hence \plat{$\sum_{i\confeqscript c} H(\ini[c]\!\cdot\reset)\mathbin=0$}, and thus
      $H(\ini[c]\cdot\und[{\transin[b]{c}}])=0$ by (\ref{took},\,\ref{T-negative}).
    \end{itemize}
  \item[(\ref{disjoint preplaces 2ST})]
    If $H(u)<0$ and $(M+\!\precond{U})[t\rangle$ with $\ell(t)\neq\tau$,
    then $t=\exec{j}$ for some $i\leqc j\in T'$ and $u\inp T^-$.
    The only candidates for $\precond{t}\cap\precond{u}\neq\emptyset$ are
    \begin{itemize}
    \item $\Pre^i_j \in \precond{(\exec{j})}
                    \cap \precond{(\ini\cdot\und[\Pre^i_j])}$ and
    \item $\transout{j} \in \precond{(\exec{j})}
                    \cap \precond{(\trans{j}\cdot\und[\transout{j}])}$ for $h<^\# j$.
    \end{itemize}
    We investigate these possibilities one by one.
    \begin{itemize}
    \item Suppose $(M+\!\precond{U})[\exec{j}\rangle$.
      By \refcl{D}(\ref{Hexec}), \plat{$H(\comp[k]{j})\geq 0$} for each $k\confeq i$. By (\ref{reset}) and
      (\ref{T-negative}) we obtain $H(\ini[i]\!\cdot\reset[k])\mathbin=0$ for each such $k$.
      Hence \hspace{-.6pt}\plat{$\displaystyle\sum_{k\confeqscript i} H(\ini[i]\!\cdot\reset[k])\mathbin=0$}, and thus
      $H(\ini[i]\cdot\und[\Pre^i_j])=0$ by (\ref{took},\,\ref{T-negative}).
    \item Suppose $(M+\!\precond{U})[\exec{j}\rangle$ and $h<^\#j$.
      By \refcl{D}(\ref{Hexec}), \plat{$H(\comp[k]{j})\geq 0$} for each $k\confeq j$. By (\ref{reset}) and
      (\ref{T-negative}) \plat{$H(\trans{j}\!\cdot\reset[k])\mathbin=0$} for each such $k$.
      So \plat{$\displaystyle\sum_{k\confeqscript j} H(\trans{j}\!\cdot\reset[k])\mathbin=0$}, and
      $H(\trans{j}\cdot\und[\transout{j}])=0$ by (\ref{took},\,\ref{T-negative}).
    \end{itemize}
  \item[(\ref{concurrent})]
    Suppose $(M+\!\precond{U})[\{t\}\mathord+\{u\}\rangle_N$, and $i,k\in T'$
    with $\ell'(i)=\ell(t)$ and $\ell'(k)=\ell(u)$.
    Since the net $N'$ is \hyperlink{plain}{plain}, $t$ and $u$ must have the form $\exec{j}$ and $\exec[k]{j}$
    for some $j>^\#i$ and $l>^\#k$. \refcl{concurrency} yields $\neg(i\confeq k)$ and hence
    $\precond{i}\cap\precond{k}=\emptyset$.
    \hfill$\Box$
  \end{enumerate}
\end{enumerate}
\end{proofNobox}
Thus, we have established that the conflict replicating implementation $\impl{N'}$ of a finitary plain
structural conflict net $N'$ without a fully reachable pure $\structuralM$ is branching
ST-bisimilar with explicit divergence to $N'$. It remains to be shown that $\impl{N'}$ is
essentially distributed.

\begin{lemma}{S-invariant}
Let $N$ be the conflict replicating implementation of a finitary net $N'=(S',T',F',M'_0,\ell')$;
let $j,l\in T'\!$, with $l\mathbin{>^\#} j$.
Then no two transitions from the set \plat{$\{\exec{j}\mid i\leqc j\}
  \cup\{\trans[j]{l}\cdot\fire\} \cup \mbox{}$} \plat{$\{\trans[j]{l}\cdot\und[\mbox{$\transout[j]{l}$}]\}
  \cup\{\exec[k]{l}\mid k\leqc l\}$} can fire concurrently.
\end{lemma}

\begin{proof}
  For each \plat{$i\mathbin{\leqc} j$} pick an arbitrary preplace $q_i$ of $i$.
  The set \plat{$
  \{\textsf{fetch}^{q_i,i}_{i,j}\textsf{-in},~\textsf{fetch}^{q_i,i}_{i,j}\textsf{-out}\mid i\leqc j\}
  \cup \mbox{}$}
  \plat{$\{\pi_{j\#l},~\transout[j]{l},~\took(\transout[j]{l},\trans[j]{l}),~\rho(\trans[j]{l}\}$}
  is an \emph{S-invariant}: there is always exactly one token in this set. This is the case because
  each transition from $N$ has as many preplaces as postplaces in this set.
  The transitions from $\{\exec{j}\mid i\leqc j\}
  \cup\{\trans[j]{l}\cdot\fire\} \cup\{\trans[j]{l}\cdot\und[\mbox{$\transout[j]{l}$}]\}\vspace{-2pt}
  \cup\{\exec[k]{l}\mid k\leqc l\}$ each have a preplace in this set.
  Hence no two of them can fire concurrently.
\end{proof}

\begin{lemma}{essentially distributed}
Let $N$ be the conflict replicating implementation $\impl{N'}$ of a finitary plain
structural conflict net $N'=(S',T',F',M'_0,\ell')$ without a fully reachable pure $\structuralM$.
Then for any $i\leqc j \confeq c\in T'$ and \hyperlink{far}{$f\in (\ini[c])^{\,\it far}$},
the transitions \plat{$\exec{j}$} and $\ini[c]\cdot\und[f]$ cannot fire concurrently.
\end{lemma}
\begin{proof}
Suppose these transitions can fire concurrently, say from the marking $M\in[M_0\rangle_N$.
By \refcl{extra}, there are $M'\in[M'_0\rangle_{N'}$ and $G\fin\Int^T$ such that
(\ref{r1})--(\ref{rVeryLast}) hold. Let $t:=\ini[c]$, $G_1:=G+\{t\cdot\und[f]\}$ and
$M_1\mathbin{:=}M+\marking{t\mathord\cdot\und[f]}$.
Then (\ref{took}), applied to the triples $(M,M',G)$ and $(M_1,M',G_1)$, yields
\[
  \sum_{\makebox[1em][l]{$\scriptstyle\{\omega\mid t\in \UI_\omega\}$}} G(t\cdot\reset[\omega])
  \leq G(t\cdot\und) < G_1(t\cdot\und)
  \leq \sum_{\makebox[1em]{$\scriptstyle\{\omega\mid t\in \UI_\omega\}$}} G_1(t\cdot\undo[\omega])
  = \sum_{\makebox[1em]{$\scriptstyle\{\omega\mid t\in \UI_\omega\}$}} G(t\cdot\undo[\omega]).
\]
Hence, there is an $\omega$ with $t\in \UI_\omega$ and $G(t\cdot\reset[\omega])< G(t\cdot\undo[\omega])$.
This $\omega$ must have the form $k\in T'$ with $k\confeq c$. We now obtain
  $$\begin{array}[b]{r@{~\leq~}ll}
  \multicolumn{1}{r@{~=~}}{0}
  & G(\comp[k]{l})
  & \mbox{(by (\ref{r2}))} \\
  & G(t\cdot\elide[k])+G(t\cdot\reset[k])
  & \mbox{(by (\ref{reset}))} \\
  \multicolumn{1}{r@{~<~}}{} & G(t\cdot\elide[k])+G(t\cdot\undo[k]) \\
  & \sum_{l\geq^\#k}G(\exec[k]{l})
  & \mbox{(by (\ref{undo}))}.
  \end{array}$$
Hence, there is an $l\geq^\# k \confeq c$ with $G(\exec[k]{l})>0$.
By (\ref{rExec5}) we obtain $\neg(j\confeq k)$, so $\precond{j}\cap\precond{k}=\emptyset$.
Additionally, we have $\precond{j}\cap\precond{c}\neq\emptyset$ and
$\precond{c}\cap\precond{k}\neq\emptyset$.
By (\ref{rVeryLast}) we obtain $M'[j\rangle$, and
by (\ref{r3}) and (\ref{r4}) $M'[k\rangle$.
Furthermore, by (\ref{took}), $G(t\cdot\und)<G_1(t\cdot\und)\leq G_1(t\cdot\fire) =
G(t\cdot\fire)$, so, for all $p\inp\precond{c}$,
  $$\begin{array}[b]{r@{~\leq~}ll}
  F'(p,c)
  & F'(p,c)\cdot \big(G(t\cdot\fire)-G(t\cdot\und)\big)\\
  & F'(p,c)\cdot \big(G(t\cdot\fire)-G(t\cdot\undone)\big)
  & \mbox{(by (\ref{took}))} \\
  & G(\dist) - \sum_{j\geq^\# i\in \postcond{p}} F'(p,i) \cdot G(\fetch)
  & \mbox{(by (\ref{p_j}))} \\
  & G(\dist)
  & \mbox{(by (\ref{rFp}) and (\ref{fetch}))} \\
  & M'(p)
  & \mbox{(by (\ref{r3}).} \\
  \end{array}$$
It follows that $M'[c\rangle$.
Thus $N'$ contains a \hyperlink{M}{fully reachable pure $\structuralM$},
which contradicts the assumptions of \reflem{essentially distributed}.
\end{proof}

\begin{theorem}{cri-distributed}
  Let $N$ be the conflict replicating implementation $\impl{N'}$ of a finitary plain
  structural conflict net $N'$ without a fully reachable pure $\structuralM$.
  Then $N$ is essentially distributed.
\end{theorem}
\begin{proof}
  We take the canonical distribution $D$ of $N$, in which $\equiv_D$ is the
  equivalence relation on places and transitions generated by Condition (1) of \refdf{distributed}.
  We need to show that this distribution satisfies Condition ($2'$) of \refdf{externally distributed}.
  A given transition $t$ with $\ell(t)\neq\tau$ must have the form \plat{$\exec{j}$} for some $i\leqc j\in T'$.
  By following the flow relation of $N$ one finds the places and transitions that, under
  the canonical distribution, are co-located with \plat{$\exec{j}$}:
  $$\begin{array}{@{}l@{}}
  \pi_{j\#l} \rightarrow \trans[j]{l}\cdot\fire \leftarrow \transin[j]{l} \rightarrow
  \ini[l]\cdot\und[\mbox{$\transin[j]{l}$}] \leftarrow \take(\transin[j]{l},\ini[l]) \\
  ~~\downarrow\\
  \hspace*{-1em}\exec{j} \\
  ~~\uparrow\\
  \transout{j}  \rightarrow \trans{j}\cdot\und[\transout{j}] \leftarrow \take(\transout{j},\trans{j}) \\
  ~~\downarrow\\
  \exec[g]{j} \\
  ~~\uparrow\\
  \Pre^g_j  \rightarrow \ini[g]\cdot\und[\Pre^g_j] \leftarrow \take(\Pre^g_j,\ini[g])
  \end{array}$$
  for all $l\mathbin{>^\#} j$, $h\mathbin{<^\#} j$ and $g\leqc j$.
  We need to show that none of these transitions can happen concurrently with \plat{$\exec{j}$}.
  For transitions $\trans[j]{l}\cdot\fire$ and $\exec[g]{j}$ this follows directly from \reflem{S-invariant}.
  For \plat{$\trans{j}\cdot\und[\transout{j}]$} this also follows from \reflem{S-invariant}, in
  which $j$, $k$ and $l$ play the r\^ole of the current $h$, $i$ and $j$.
  For the transitions \plat{$\ini[l]\cdot\und[\mbox{$\transin[j]{l}$}]$}
  and \plat{$\ini[g]\cdot\und[\Pre^g_j]$} this has been established in \reflem{essentially distributed}.
\end{proof}
Our main result follows by combining Theorems~\ref{thm-correctness}
and~\ref{thm-cri-distributed} and \refpr{essentiallydistributedequal}:

\begin{theorem}{fullmgttrulysync}
  Let $N$ be a finitary plain structural conflict net
  without a fully reachable \visible pure \structuralM.
  Then $N$ is distributable up to $\approx^\Delta_{bSTb}$.
\end{theorem}

\begin{corollary}{fullmeqtrulysync}
  Let $N$ be a finitary  plain structural conflict net.
  Then $N$ is distributable iff it has no fully reachable
  \visible pure~\structuralM.
\end{corollary}

\section{Conclusion}

In this paper, we have given a precise characterisation of
distributable Petri nets in terms of a semi-structural property. Moreover, we
have shown that our notion of distributability corresponds to an
intuitive notion of a distributed system by establishing that any
distributable net may be implemented as a network of asynchronously
communicating components.

In order to formalise what qualifies as a valid implementation, we needed a suitable
equivalence relation. We have chosen step readiness equivalence for showing the
impossibility part of our characterisation, since it is one of the simplest and least
discriminating semantic equivalences imaginable that abstracts from internal actions but
preserves branching time, concurrency and divergence to some small degree. For the
positive part, stating that all other nets are implementable, we have introduced a
combination of several well known rather discriminating equivalences, namely a divergence
sensitive version of branching bisimulation adapted to ST-semantics.  Hence our
characterisation is rather robust against the chosen equivalence; it holds in fact for all
equivalences between these two notions.  However, ST-equivalence (and our version of it)
preserves the causal structure between action occurrences only as far as it can be expressed in terms
of the possibility of durational actions to overlap in time. Hence a natural question is
whether we could have chosen an even stronger causality sensitive equivalence for our
implementability result, respecting e.g.\ pomset equivalence or history preserving
bisimulation.  Our conflict replicating implementation does not fully preserve the causal
behaviour of nets; we are convinced that we have chosen the strongest possible equivalence
for which our implementation works.  It is an open problem to find a class of nets that
can be implemented distributedly while preserving divergence, branching time and causality
in full.
Another line of research is to investigate which Petri nets can be
implemented as distributed nets when relaxing the requirement of
preserving the branching structure.
If we allow linear time correct implementations (using a step trace
equivalence), we conjecture that all Petri nets become distributable.
However, also in this case it is problematic, in fact even impossible in our setting, to
preserve the causal structure, as has been shown in \cite{EPTCS64.9}.
A similar impossibility result has been obtained in the world of the $\pi$-calculus in
\cite{EPTCS64.7}.

The interplay between choice and synchronous communication has already been investigated
in quite a number of approaches in different frameworks. We refer to
\cite{glabbeek08syncasyncinteractionmfcs} for a rather comprehensive overview and
concentrate here on recent and closely related work.

The idea of modelling asynchronously communicating sequential components by sequential
Petri nets interacting though buffer places has already been considered in
\cite{reisig82buffersync}. There Wolfgang Reisig introduces a class of systems, represented
as Petri nets, where the relative speeds of different components are guaranteed to be
irrelevant.  His class is a strict subset of our LSGA nets, requiring additionally,
amongst others, that all choices in sequential components are free, \ie do not depend upon
the existence of buffer tokens, and that places are output buffers of only one component.
Another quite similar approach was taken in \cite{EHH10}, where
transition labels are classified as being either input or output.
There, asynchrony is introduced by adding new buffer places during
net composition.
This framework does not allow multiple senders for a single receiver.

Other notions of distributed and distributable Petri nets are proposed in
\cite{hopkins91distnets,BCD02,BD11}. In these works, given a distribution of the
transitions of a net, the net is distributable iff it can be implemented by a net that is
distributed w.r.t.\ that distribution. The requirement that concurrent transitions may not
be co-located is absent; given the fixed distribution, there is no need for such a
requirement. These papers differ from each other, and from ours, in what counts as a valid
implementation. A comparison of our criterion with that of Hopkins
\cite{hopkins91distnets} is provided in \cite{glabbeek08syncasyncinteractionmfcs}.

In \cite{glabbeek08syncasyncinteractionmfcs} we have obtained a characterisation similar to
Corollary~\ref{cor-fullmeqtrulysync}, but for a much more restricted notion of distributed
implementation (\emph{plain distributability}), disallowing nontrivial transition
labellings in distributed implementations.  We also proved that fully reachable pure
\structuralM s are not implementable in a distributed way, even when using transition
labels (\refthm{trulysyngltfullm}). However, we were not able to show that this upper
bound on the class of distributable systems was tight.  Our current work implies the
validity of Conjecture 1 of \cite{glabbeek08syncasyncinteractionmfcs}.
While in \cite{glabbeek08syncasyncinteractionmfcs} we considered only one-safe
place/transition systems, the present paper employs a more general class of
place/transition systems, namely structural conflict nets. This enables us to give a
concrete characterisation of distributed nets as systems of sequential components
interacting via non-safe buffer places.

\bibliographystyle{eptcs}

\end{document}